\documentclass[a4paper,UKenglish]{lipics-v2019}

\usepackage[noend]{algpseudocode}
\usepackage[ruled]{algorithm}
\usepackage[dvipsnames,table]{xcolor}
\usepackage{etoolbox,mathtools,tikz,xspace}
\usepackage{extraref}

\usetikzlibrary{ipe}
\graphicspath{{figs/}}

\title{Fr\'echet Distance for Uncertain Curves}
\author{Kevin Buchin}{Department of Mathematics and Computer Science, TU Eindhoven, Netherlands}{k.a.buchin@tue.nl}{https://orcid.org/0000-0002-3022-7877}{}
\author{Chenglin Fan}{Department of Computer Science, University of Texas at Dallas, Richardson, TX 75080, USA}{cxf160130@utdallas.edu}{}{}
\author{Maarten L\"offler}{Department of Information and Computing Sciences, Utrecht University, Netherlands}{m.loffler@uu.nl}{}{}
\author{Aleksandr Popov}{Department of Mathematics and Computer Science, TU Eindhoven, Netherlands}{a.popov@tue.nl}{https://orcid.org/0000-0002-0158-1746}{Netherlands Organisation for Scientific Research (NWO), project no.\@ 612.001.801.}
\author{Benjamin Raichel}{Department of Computer Science, University of Texas at Dallas, Richardson, TX 75080, USA}{benjamin.raichel@utdallas.edu}{}{}
\author{Marcel Roeloffzen}{Department of Mathematics and Computer Science, TU Eindhoven, Netherlands}{m.j.m.roeloffzen@tue.nl}{}{Netherlands Organisation for Scientific Research (NWO), project no.\@ 628.011.005.}
\authorrunning{K.~Buchin, C.~Fan, M.~L\"offler, A.~Popov, B.~Raichel, and M.~Roeloffzen}
\Copyright{Kevin Buchin, Chenglin Fan, Maarten L\"offler, Aleksandr Popov, Benjamin Raichel, and Marcel Roeloffzen}
\relatedversion{Results on upper bounds and expected values appear in the master thesis of A.~Popov
at \url{https://research.tue.nl/en/studentTheses/similarity-of-uncertain-trajectories}.}
\funding{\textit{Chenglin Fan and Ben Raichel}: NSF CRII Award 1566137 and CAREER Award 1750780.}
\hideLIPIcs
\nolinenumbers

\newcommand*{\dsh}{---}
\newcommand*{\reals}[1][P]{\Subset_{\mathbb{#1}}}
\newcommand*{\Real}[1]{\Realop(\mathcal{#1})}
\newcommand*{\diam}[1]{\diamop(#1)}
\newcommand*{\concat}{\mathbin{\Vert}}
\newcommand*{\fr}{d_\mathrm{F}}
\newcommand*{\dfr}{d_\mathrm{dF}}
\newcommand*{\dfrmax}{\dfr^{\,\max}}
\newcommand*{\dfrmin}{\dfr^{\,\min}}
\newcommand*{\dfrexp}{\dfr^{\,\mathbb{E}}}
\newcommand*{\dfrexpp}[1][P]{\dfr^{\,\mathbb{E}(\mathbb{#1})}}
\newcommand*{\frmax}{\fr^{\,\max}}
\newcommand*{\frmin}{\fr^{\,\min}}
\newcommand*{\frexp}{\fr^{\,\mathbb{E}}}
\newcommand*{\frexpp}[1][P]{\fr^{\,\mathbb{E}(\mathbb{#1})}}
\newcommand*{\cnfsat}{CNF-SAT}
\newcommand*{\np}{NP}
\newcommand*{\nump}{\#P}
\newcommand*{\True}{\textsc{True}}
\newcommand*{\False}{\textsc{False}}

\newcommand*{\lcg}{\mathrm{lcg}}
\newcommand*{\ucg}{\mathrm{ucg}}
\newcommand*{\cg}{\mathrm{cg}}
\newcommand*{\g}{\mathrm{g}}
\newcommand*{\intrv}{\mathrm{int}}
\newcommand*{\midpt}{\mathrm{mid}}
\newcommand*{\gs}{\mathrm{gs}}
\newcommand*{\opt}{\mathrm{opt}}

\DeclareMathOperator*{\argmin}{arg\,min}
\DeclareMathOperator*{\Concat}{\Big\Vert}
\DeclareMathOperator{\Feas}{Feas}
\DeclareMathOperator{\Prop}{Prop}
\DeclareMathOperator{\Reach}{Reach}
\DeclareMathOperator{\Realop}{Real}
\DeclareMathOperator{\diamop}{diam}
\DeclareMathOperator{\Stab}{Stab}
\DeclareMathOperator*{\width}{width}
\DeclareMathOperator{\EG}{EG}
\DeclareMathOperator{\Thick}{Thick}
\newcolumntype{C}{>{$}c<{$}}

\newcommand*{\curveA}{\pi}
\newcommand*{\curveB}{\sigma}
\newcommand*{\R}{\mathbb{R}}
\newcommand*{\U}{\mathcal{U}}
\newcommand*{\W}{\mathcal{W}}
\newcommand*{\ur}{U}

\newcommand*{\eps}{\varepsilon}
\newcommand*{\Frechet}{Fr\'{e}chet\xspace}
\newcommand{\pth}[2][\!]{#1\left({#2}\right)}

\newcommand*{\decider}{\textsc{Decider}\xspace}
\newcommand*{\gridDecider}{\textsc{GridDecider}\xspace}

\MakeRobust{\Call}
\makeatletter
\patchcmd{\ALG@doentity}{\item[]\nointerlineskip}{}{}{}
\let\OldStatex\Statex
\renewcommand{\Statex}{\OldStatex\hskip\ALG@thistlm}
\algnewcommand{\LineComment}[1]{\Statex \(\triangleright\) #1}
\makeatother
\algnewcommand{\To}{~\textbf{to}}

\newtheorem{problem}[theorem]{Problem}
\newtheorem{observation}[theorem]{Observation}
\addrefclass{prob}{Problem}

\begin{document}
\maketitle

\begin{abstract}
In this paper we study a wide range of variants for computing the (discrete and continuous) \Frechet
distance between uncertain curves.
We define an uncertain curve as a sequence of \emph{uncertainty regions}, where each region is a
disk, a line segment, or a set of points.
A \emph{realisation} of a curve is a polyline connecting one point from each region.
Given an uncertain curve and a second (certain or uncertain) curve, we seek to compute the lower and
upper bound \Frechet distance, which are the minimum and maximum \Frechet distance for any
realisations of the curves.

We prove that both the upper and lower bound problems are \np-hard for the continuous Fr\'echet
distance in several uncertainty models, and that the upper bound problem remains hard for the
discrete Fr\'echet distance.
In contrast, the lower bound (discrete~\cite{ahn:2012} and continuous) Fr\'echet distance can be
computed in polynomial time.
Furthermore, we show that computing the expected discrete Fr\'echet distance is \nump-hard when the
uncertainty regions are modelled as point sets or line segments.
The construction also extends to show \nump-hardness for computing the continuous Fr\'echet distance
when regions are modelled as point sets.

On the positive side, we argue that in any constant dimension there is a FPTAS for the lower bound
problem when $\Delta/\delta$ is polynomially bounded, where $\delta$ is the Fr\'echet distance and
$\Delta$ bounds the diameter of the regions.
We then argue there is a near-linear-time $3$-approximation for the decision problem when the
regions are convex and roughly $\delta$-separated.
Finally, we also study the setting with Sakoe--Chiba time bands, where we restrict the alignment
between the two curves, and give polynomial-time algorithms for upper bound and expected discrete
and continuous Fr\'echet distance for uncertainty regions modelled as point sets.
\end{abstract}

\section{Introduction}\label{sec:intro}
In this paper we investigate the well-studied topic of curve similarity in the context of the
burgeoning area of geometric computing under uncertainty. 
While classical algorithms in computational geometry typically assume the input point locations are
known exactly, in recent years there has been a concentrated effort to adapt these algorithms to
uncertain inputs, which can more faithfully model real-world inputs.
The need to model such uncertain inputs is perhaps no more clear than for the location data of a
moving object obtained from physical devices, which is inherently imprecise due to issues such as
measurement error, sampling error, and network latency~\cite{pfoser:1999, prasadsistla:1998}.
Moreover, to ensure location privacy, one may purposely add uncertainty to the data by adding
noise or reporting positions as geometric regions rather than points.
(See the survey by Krumm~\cite{krumm:2009} and the references therein.)

Here we consider both the continuous and discrete Fr\'echet distance for uncertain curves.
Given the applications above, our uncertain input is given as a sequence of compact regions, from
which a polygonal curve is \emph{realised} by selecting one point from each region.
Our goal is to find, for a given pair of uncertain curves, the upper bound, lower bound, and
expected Fr\'echet distance, where the upper (resp.\@ lower) bound Fr\'echet distance is the maximum
(resp.\@ minimum) distance over any realisation.
For the expected Fr\'echet distance we assume a probability distribution is provided that describes
how each vertex on a curve is chosen from the compact region.
Previously, Ahn et al.~\cite{ahn:2012} considered the lower bound problem for the discrete Fr\'echet
distance, giving a polynomial-time algorithm for points in constant dimension.
The authors also gave efficient approximation algorithms for the discrete upper bound Fr\'echet
distance for uncertain inputs, where the approximation factor depends on the spread of the region
diameters or how well-separated they are.
Subsequently, Fan and Zhu showed that the discrete upper bound Fr\'echet distance is \np-hard for
uncertain inputs modelled as thin rectangles~\cite{fan:2018}.
To our knowledge, we are the first to consider either variant for the continuous Fr\'echet case, and
the first to consider the expected Fr\'echet distance.

\subsection{Previous Work}
\subparagraph*{Geometric computing under uncertainty:}
The two most common models of geometric uncertainty are the locational model~\cite{loeffler:2009}
and the existential model~\cite{suri:2013, yiu:2009}.
In the existential model the location of an uncertain point is known, but the point may not be
present; in the locational model we know that each uncertain point exists, but not its exact
location.

In this paper we consider the locational model.
Each uncertain point is a set of potential locations. 
We call an uncertain point \emph{indecisive} if the set of potential locations is finite, or
\emph{imprecise} if the set is not finite but is a convex region.
A \emph{realisation} of a set of uncertain points is a selection of one point from each uncertain
point.
The goal is typically to compute the realisation of a set of uncertain points that minimises or
maximises some quantity (e.g.\@ area, distance, perimeter) of some underlying geometric structure
(e.g.\@ convex hull, MST).
A large number of minimisation and maximisation variants for imprecise points can be found in the
thesis of Maarten L\"offler~\cite{loeffler:2009} and other works~\cite{knauer:2011, loeffler:2014,
loeffler:2006}.
For indecisive points such problems are often called colour-spanning problems, as each indecisive
point can be viewed as a colour and the goal is to select a point of each colour to minimise or
maximise some quantity~\cite{abellanas:2001, arkin:2018, das:2009, fan:2013}.
Besides finding tight upper and lower bounds for various measures, there have also been studies on
visibility~\cite{leizhen:1997}, imprecise terrains~\cite{driemel:2013, gray:2012}, and Voronoi
diagrams~\cite{sember:2008} and Delaunay triangulations~\cite{buchin:2011, loeffler:2010,
kreveld:2010}.

By assigning a probability distribution to uncertain points, one can also consider the expectation
or distribution of various measures~\cite{agarwal:2016, agarwal:2017, buchin:2019-2, jorgensen:2011,
pei:2007}.
Finally, imprecision has also been studied from a movement perspective, with the focus on the
imprecision between measurements~\cite{buchin:2012} and how imprecision grows and shrinks as time
passes and new location information becomes available~\cite{evans:2013}.

\subparagraph*{Fr\'echet distance:}
Computing the Fr\'echet distance between two precise curves can be done in near-quadratic
time~\cite{agarwal:2014, alt:1995, buchin:2017}, and assuming the Strong Exponential Time Hypothesis
(SETH) it cannot be computed or even approximated well in strongly subquadratic
time~\cite{bringmann:2014, buchin:2019}.
However, for several restricted versions the Fr\'echet distance can be calculated more quickly, for
example for $c$-packed curves~\cite{driemel:2012}, when the edges are long~\cite{gudmundsson:2019},
or when the alignment of curves is restricted~\cite{buchin:2010, maheshwari:2011}.
Many variants of the problem have been considered: Fr\'echet distance with
shortcuts~\cite{buchin:2014, driemel:2018}, weak Fr\'echet distance~\cite{alt:1995},
discrete Fr\'echet distance~\cite{agarwal:2014, eiter:1994}, Fr\'echet gap distance~\cite{fan:2017},
Fr\'echet distance under translations~\cite{bringmann:2019, filtser:2018}, and more.

There are also numerous applications of different variants of Fr\'echet distance in common curve and
trajectory analysis tasks, such as clustering~\cite{buchin:2019-3, buchin:2019-4} or curve
simplification~\cite{kerkhof:2019, kreveld:2018}.

\begin{table}
\centering
\caption{Hardness results for the decision problems in this paper.
Ahn et al.~\cite{ahn:2012} solve the lower bound problem for disks, but their algorithm extends to
the indecisive curves as well as line segment imprecision.}
\begin{tabular}{c c | c c c}
\hline
&&\multirow{2}{*}{indecisive} & \multicolumn{2}{c}{imprecise}\\
&&& disks & line segments\\
\hline
\multirow{3}{*}{discrete Fr\'echet distance} 
& LB & Polynomial~\cite{ahn:2012} & Polynomial~\cite{ahn:2012} & Polynomial~\cite{ahn:2012}\\
& UB & \np-complete & \np-complete & \np-complete\\
& Exp & \nump-hard & --- & \nump-hard\\
\hline
\multirow{3}{*}{Fr\'echet distance}
& LB & Polynomial & --- & \np-complete \\
& UB & \np-complete & \np-complete & \np-complete\\
& Exp & \nump-hard & --- & ---\\
\hline
\end{tabular}
\label{tab:hard}
\end{table}

\subsection{Our Contributions}
In this paper we present an extensive study of the Fr\'echet distance for uncertain curves.
We provide a wide range of hardness results and present several approximations and polynomial-time
solutions to restricted versions.
We are the first to consider the continuous Fr\'echet distance in the uncertain setting, as well as
the first to consider the expected Fr\'echet distance.

On the negative side, we present a plethora of hardness results (see Table~\ref{tab:hard}; details
follow in Section~\ref{sec:hardness}).
The hardness of the lower bound case is curious: while the variants discrete Fr\'echet distance on
imprecise inputs~\cite{ahn:2012} and, as we prove, continuous Fr\'echet distance on indecisive
inputs both permit a simple dynamic programming solution, the variant continuous Fr\'echet distance
on imprecise input has just enough (literal) wiggle room to show \np-hardness by reduction from
\textsc{SubsetSum}.

We complement the lower bound hardness result by two approximation algorithms
(Section~\ref{sec:lowerbound}).
The first is a FPTAS for general uncertain curves in constant dimension when the ratio between the
diameter of the uncertain points and the lower bound Fr\'echet distance is polynomially bounded.
The second is a $3$-approximation for separated imprecise curves, but uses a simpler greedy approach
that runs in near-linear time.

The \np-hardness of the upper bound by a reduction from \textsc{CNF-SAT} is less surprising, but
requires a careful set-up and analysis of the geometry to then extend it to a reduction from
\textsc{\#CNF-SAT} to the expected (discrete or continuous) Fr\'echet distance.
However, by adding the common constraint that the alignment between the curves needs to stay within
a Sakoe--Chiba~\cite{sakoe:1978} band of constant width (see Section~\ref{sec:sakoe-chiba} for
definition and results), we can solve these problems in polynomial time for indecisive curves.
Sakoe--Chiba bands are frequently used for time-series data~\cite{berndt:1994, keogh:2005,
sakoe:1978} and trajectories~\cite{buchin:2010, devogele:2017}, when the alignment should (or is
expected to) not vary too much from a certain `natural' alignment.

\section{Preliminaries}\label{sec:notation}
In this section, we introduce the notation relevant to the rest of this paper, as well as recall the
definitions of (discrete) Fr\'echet distance.

\subsection{Curves}
Denote $[n] \equiv \{1, 2, \dots, n\}$.
Consider a \emph{sequence} of $d$-dimensional points $\pi = \langle p_1, p_2, \dots, p_n\rangle$.
A \emph{polygonal curve} $\pi$ is defined by these points by linearly interpolating between the
successive points and can be seen as a continuous function: $\pi(i + \alpha) = (1 - \alpha) p_i +
\alpha p_{i + 1}$ for $i \in [n - 1]$ and $\alpha \in [0, 1]$.
The \emph{length} of such a curve is the number of its vertices, $\lvert \pi\rvert = n$.
Where we deem important to distinguish between points that are a part of the curve and other points,
we denote the polygonal curve by $\pi = \langle\pi_1, \pi_2, \dots, \pi_n\rangle$.
We denote the \emph{concatenation} of two polygonal curves $\pi$ and $\sigma$ of lengths $n$ and $m$
by $\pi \concat \sigma$; the new curve follows $\pi$, then has a segment between $\pi(n)$ and
$\sigma(1)$, and then follows $\sigma$.
Similarly, $p \concat q$ (or simply $pq$) denotes the line segment between points $p$ and $q$.
We can generalise this notation:
\[\Concat_{i \in [n]} p_i \equiv p_1 \concat p_2 \concat \dots \concat p_n \equiv \pi\,.\]
We denote a \emph{subcurve} from vertex $i$ to $j$ of curve $\pi$ as $\pi[i: j] \equiv p_i \concat
p_{i + 1} \concat \dots \concat p_j$.

\subsection{Metrics Definitions}
Given two points $x, y \in \R^d$, denote their Euclidean distance by $\lVert x - y\rVert$.
For two compact sets $X, Y\subset \R^d$, denote their distance by $\lVert X - Y\rVert =
\min_{x \in X, y \in Y} \lVert x - y\rVert$.
Throughout we treat the dimension $d$ as a small constant.

Let $\Phi_n$ denote the set of all \emph{reparametrisations} of length $n$, defined as continuous
non-decreasing functions $\phi: [0, 1] \to [1, n]$ where $\phi(0) = 1$ and $\phi(1) = n$.
Given a pair of curves $\pi$ and $\sigma$ of lengths $n$ and $m$, respectively, and corresponding
reparametrisations $\phi_1\in \Phi_n$ and $\phi_2\in\Phi_m$, define
$\width_{\phi_1,\phi_2}(\pi,\sigma) = \max_{t \in [0, 1]}
\lVert \pi(\phi_1(t)) - \sigma(\phi_2(t))\rVert$.

The width represents the maximum distance between two points traversing the curves from start to end
according to $\phi_1$ and $\phi_2$ (which allow varying speed, but no backtracking).
The \emph{Fr\'echet distance} $\fr(\pi, \sigma)$ is defined as the minimum possible width over all
such traversals:
\[\fr(\pi, \sigma) = 
\inf_{\phi_1 \in \Phi_n, \phi_2 \in \Phi_m} \width_{\phi_1,\phi_2}(\pi,\sigma) 
=\inf_{\phi_1 \in \Phi_n, \phi_2 \in \Phi_m} \max_{t \in [0, 1]}
\lVert \pi(\phi_1(t)) - \sigma(\phi_2(t))\rVert\,.\]

The \emph{discrete Fr\'echet distance} $\dfr(\pi, \sigma)$ is defined similarly, except that we do
not traverse edges of the curves, but must jump from one vertex to the next on either or both
curves.
We define a valid \emph{coupling} as a sequence $c = \langle(p_1, q_1), \dots, (p_r, q_r)\rangle$
of pairs from $[n]\times [m]$ where $(p_1, q_1) = (1, 1)$, $(p_r, q_r) = (n, m)$, and, for any
$i \in [r - 1]$ we have
$(p_{i + 1}, q_{i + 1}) \in \{(p_i + 1, q_i), (p_i, q_i + 1), (p_i + 1, q_i + 1)\}\,.$
Let $\mathcal{C}$ be the set of all valid couplings on curves of lengths $n$ and $m$, then
\[\dfr(\pi, \sigma) = \inf_{c \in \mathcal{C}} \max_{s \in [\lvert c\rvert]} \lVert \pi(p_s) -
\sigma(q_s)\rVert\,,\]
where $c_s = (p_s, q_s)$ for all $s \in [\lvert c\rvert]$.
Both distances are illustrated in Figure~\ref{fig:frechet}.

\begin{figure}
\centering
\begin{tikzpicture}
\begin{scope}
\draw (0, 0) -- (0, 1) -- (0, 2) -- (2, 4);
\draw (2, 0) -- (1, 1) -- (3, 2) -- (4, 4);
\draw[red, thin, dashed] (0, 0) -- (2, 0) (0, 1) -- (1, 1) (0, 2) -- (1, 1) (2, 4) -- (4, 4);
\draw[red, very thick, dashed] (2, 4) -- (3, 2);

\fill[black]  (0, 0) node[left] {$(0, 0)$} circle[radius=2pt]
              (0, 1) node[left] {$(0, 1)$} circle[radius=2pt]
              (0, 2) node[left] {$(0, 2)$} circle[radius=2pt]
              (2, 4) node[left] {$(2, 4)$} circle[radius=2pt];
\fill[black]  (2, 0) node[right] {$(2, 0)$} circle[radius=2pt]
              (1, 1) node[right,xshift=6pt] {$(1, 1)$} circle[radius=2pt]
              (3, 2) node[right] {$(3, 2)$} circle[radius=2pt]
              (4, 4) node[right] {$(4, 4)$} circle[radius=2pt];
\end{scope}
\begin{scope}[xshift=6cm]
\draw (0, 0) -- (0, 1) -- (0, 2) -- (2, 4);
\draw (2, 0) -- (1, 1) -- (3, 2) -- (4, 4);
\fill[green,opacity=.1] (0, 0) -- (0, 1) -- (0, 2) -- (2, 4) -- (4, 4) -- (3, 2) -- (1, 1) -- (2, 0)
                      -- cycle;
\draw[green,thin,dashed] (0, 0) -- (2, 0) (0, 1) -- (1, 1) -- (0, 2) (4, 4) -- (2, 4) -- (3.6, 3.2);
\draw[green, very thick, dashed] (1.5, 3.5) -- (3, 2);

\fill[black]  (0, 0) node[left] {$(0, 0)$} circle[radius=2pt]
              (0, 1) node[left] {$(0, 1)$} circle[radius=2pt]
              (0, 2) node[left] {$(0, 2)$} circle[radius=2pt]
              (2, 4) node[left] {$(2, 4)$} circle[radius=2pt];
\fill[black]  (2, 0) node[right] {$(2, 0)$} circle[radius=2pt]
              (1, 1) node[right,xshift=6pt] {$(1, 1)$} circle[radius=2pt]
              (3, 2) node[right] {$(3, 2)$} circle[radius=2pt]
              (4, 4) node[right] {$(4, 4)$} circle[radius=2pt];
\end{scope}
\end{tikzpicture}
\caption{Left: Discrete Fr\'echet distance, where an optimal coupling is shown in dashed red lines.
Right: Fr\'echet distance, dashed green lines indicate specific values for $\delta$ for optimal
functions $\phi_1$, $\phi_2$.}
\label{fig:frechet}
\end{figure}
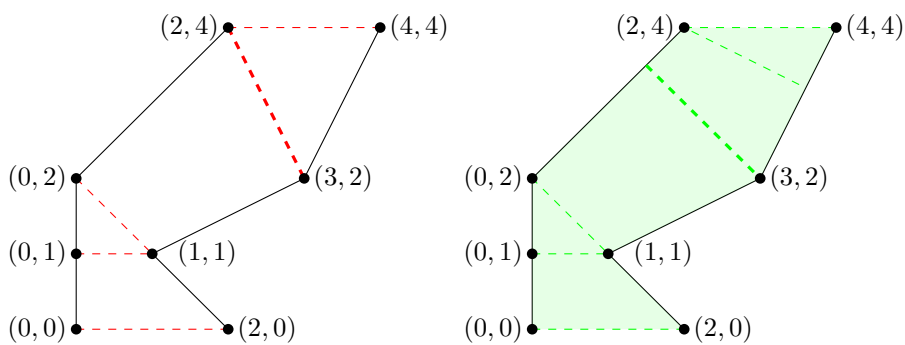

\subparagraph*{Computing discrete Fr\'echet distance:}
We recall the standard dynamic programming approach by Eiter and Mannila~\cite{eiter:1994}.
The algorithm is deduced in a standard manner from the following recursion:
\begin{align*}
\dfr(\pi[1: i + 1], \sigma[1: j + 1]) = &\max(\lVert \pi(i + 1) - \sigma(j + 1)\rVert,\\
&\quad\min(\dfr(\pi[1: i], \sigma[1: j]),\\
&\qquad\dfr(\pi[1: i + 1], \sigma[1: j]),\\
&\qquad\dfr(\pi[1: i], \sigma[1: j + 1])))\,.
\end{align*}
In words, the discrete Fr\'echet distance is the maximum of the distance of the newly added element
in the coupling and the value that was considered best previously.
Due to the coupling restrictions, there are only three possible subproblems that we need to
consider, and we may choose the best of them, thus obtaining the recursion above.
It is straightforward to turn it into a dynamic program.

Table~\ref{tab:dfd} gives the distance matrix and the computation of the discrete Fr\'echet distance
for the example of Figure~\ref{fig:frechet}.
Each cell of the table on the right shows the value of the discrete Fr\'echet distance so far; the
final result can be read out from the top right corner of the table, and the coupling that yields
this result can be read from the sequence of grey cells.
Notice that the table shows the same coupling as Figure~\ref{fig:frechet}.

\begin{table}
\centering
\begin{tabular}{|C C C C|}
\hline
4 & \sqrt{10} & \sqrt{5} & 2\\
2\sqrt{2} & \sqrt{2} & 3 & 2\sqrt{5}\\
\sqrt{5} & 1 & \sqrt{10} & 5\\
2 & \sqrt{2} & \sqrt{13} & 4\sqrt{2}\\
\hline
\end{tabular}
\qquad
\begin{tabular}{|C C C C|}
\hline
4 & \sqrt{10} & \cellcolor{lightgray}\sqrt{5} & \cellcolor{lightgray}\sqrt{5}\\
2\sqrt{2} & \cellcolor{lightgray}2 & 3 & 2\sqrt{5}\\
\sqrt{5} & \cellcolor{lightgray}2 & \sqrt{10} & 5\\
\cellcolor{lightgray}2 & 2 & \sqrt{13} & 4\sqrt{2}\\
\hline
\end{tabular}
\caption{Left: Distance matrix on vertices for example of Figure~\ref{fig:frechet}. Right: Dynamic
program for discrete Fr\'echet distance, filled from bottom left corner. Rows correspond to points
from the left trajectory, columns---to points from the right trajectory. Optimal path is marked in
grey.}
\label{tab:dfd}
\end{table}

Given two trajectories of length $n$ and $m$ in two dimensions, this approach takes $\Theta(mn)$
time to run.
More recently, Agarwal et al.~\cite{agarwal:2014} presented an algorithm that computes discrete
Fr\'echet distance in time $\mathcal{O}\left(\tfrac{mn \log \log n}{\log n}\right)$ in two
dimensions, for $m \leq n$.
However, it is rather complex and does not help the intuition about the problems discussed in this
thesis, so we will not go into further detail.
The decision version of the problem can be solved in a similar fashion, but propagating boolean
values instead.

\subparagraph*{Computing Fr\'echet distance:}
One can use a similar approach to solve the decision version of the Fr\'echet distance problem,
except now we have free and blocked areas within each cell of the table rather than simply having a
boolean value in each cell.
The resulting table is called a \emph{free-space diagram.}
On polygonal curves, each cell becomes an intersection of an ellipse with the cell, with the inside
of the ellipse being free.
The answer to the problem is \True\ if and only if there is a monotone path from the bottom left
corner to the top right corner of the free-space diagram.
A free-space diagram for the example of the two polygonal curves of Figure~\ref{fig:frechet} is
shown in Figure~\ref{fig:freespace}.

Algorithmically this can be checked by keeping the open intervals on the edges of the cells, i.e.\@
the white segments on cell borders shown in Figure~\ref{fig:freespace}.
The algorithm then runs in time $\Theta(mn)$.
For further details the reader is invited to consult the work by Alt and Godau~\cite{alt:1995}
or previous work on the same topic~\cite{godau:1991}.

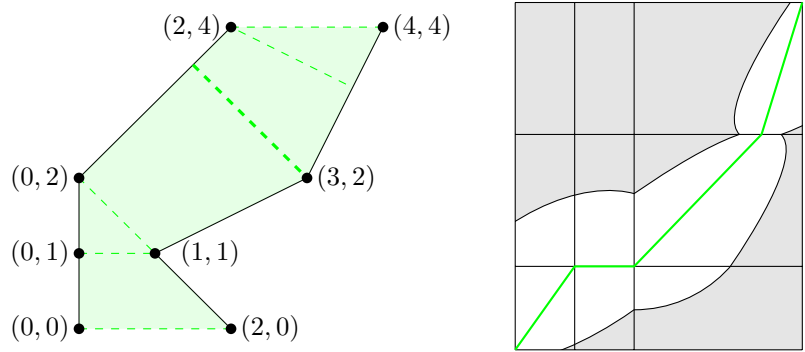
\begin{figure}
\centering
\begin{tikzpicture}
\draw (0, 0) -- (0, 1) -- (0, 2) -- (2, 4);
\draw (2, 0) -- (1, 1) -- (3, 2) -- (4, 4);
\fill[green,opacity=.1] (0, 0) -- (0, 1) -- (0, 2) -- (2, 4) -- (4, 4) -- (3, 2) -- (1, 1) -- (2, 0)
            -- cycle;
\draw[green,thin,dashed] (0, 0) -- (2, 0) (0, 1) -- (1, 1) -- (0, 2) (4, 4) -- (2, 4) -- (3.6, 3.2);
\draw[green,very thick,dashed] (1.5, 3.5) -- (3, 2);

\fill[black]  (0, 0) node[left] {$(0, 0)$} circle[radius=2pt]
        (0, 1) node[left] {$(0, 1)$} circle[radius=2pt]
        (0, 2) node[left] {$(0, 2)$} circle[radius=2pt]
        (2, 4) node[left] {$(2, 4)$} circle[radius=2pt];
\fill[black]  (2, 0) node[right] {$(2, 0)$} circle[radius=2pt]
        (1, 1) node[right,xshift=6pt] {$(1, 1)$} circle[radius=2pt]
        (3, 2) node[right] {$(3, 2)$} circle[radius=2pt]
        (4, 4) node[right] {$(4, 4)$} circle[radius=2pt];
\end{tikzpicture}%
\qquad%
\begin{tikzpicture}[ipe import,scale=0.88]
  \begin{scope}[scale=0.6]
    \fill[gray,fill opacity=.2]
      (224, 518) rectangle (426.7939, 765.2267);
    \begin{scope}[shift={(224, 518)}]
      \begin{scope}
        \clip
          (0, 0) rectangle (42, 59.397);
        \filldraw[fill=white]
          (84, 118.7939) ellipse[x radius=166.8526, y radius=69.1126, rotate=45];
      \end{scope}
      \begin{scope}
        \clip
          (0, 59.397) rectangle (42, 153.3118);
        \filldraw[fill=white]
          (21, 12.4395) ellipse[x radius=121.4532, y radius=75.0622, rotate=45];
      \end{scope}
      \begin{scope}
        \clip
          (42, 0) rectangle (84, 59.397);
        \filldraw[fill=white]
          (84, 118.7939) ellipse[x radius=166.8526, y radius=69.1126, rotate=45];
      \end{scope}
      \begin{scope}
        \clip
          (42, 59.397) rectangle (84, 153.3118);
        \filldraw[fill=white]
          (21, 12.4395) ellipse[x radius=121.4532, y radius=75.0622, rotate=45];
      \end{scope}
      \begin{scope}
        \clip
          (84, 0) rectangle (202.7939, 59.397);
        \filldraw[fill=white]
          (84, 118.7939) circle[radius=90.3];
      \end{scope}
      \begin{scope}
        \clip
          (84, 59.397) rectangle (202.7939, 153.3118);
        \filldraw[fill=white]
          (-94.1909, -128.4327) ellipse[x radius=398.6194, y radius=64.687, rotate=45];
      \end{scope}
      \begin{scope}
        \clip
          (84, 153.3118) rectangle (202.7939, 247.2267);
        \filldraw[fill=white]
          (440.3818, 435.0564) ellipse[x radius=398.6194, y radius=64.687, rotate=45];
      \end{scope}
    \end{scope}
    \draw
      (224, 518)
       -- (224, 765.2267)
      (266, 518)
       -- (266, 765.2267)
      (308, 518)
       -- (308, 765.2267)
      (426.7939, 518)
       -- (426.7939, 765.2267)
      (224, 518)
       -- (426.7939, 518)
      (224, 577.397)
       -- (426.7939, 577.397)
      (224, 671.3118)
       -- (426.7939, 671.3118)
      (224, 765.2267)
       -- (426.7939, 765.2267);
    \draw[green,thick] (224, 518) -- (266, 577.397) -- (308, 577.397) -- (398, 671.3118) -- 
    (426.7939, 765.2267);
  \end{scope}
\end{tikzpicture}
\caption{Left: Visualisation of Fr\'echet distance. Right: Free-space diagram for threshold
$\varepsilon = 2.15$. One can draw a monotonous \textcolor{green}{path} from the lower left corner to
the upper right corner of the diagram, so the Fr\'echet distance between trajectories is below the
threshold.}
\label{fig:freespace}
\end{figure}

\subsection{Uncertainty Model}
An \emph{uncertain} point is commonly represented as a compact region $U \subset \mathbb{R}^d$.
Usually, it is a finite set of points, a disk, a rectangle, or a line segment.
The intuition is that only one point from this region represents the true location of the point;
however, we do not know which one.
A \emph{realisation} $p$ of such a point is one of the points from the region $U$.
When needed we assume the realisations are drawn from $U$ according to a known probability
distribution $\mathbb{P}$.
We denote the diameter of any compact set (e.g.\@ an uncertain point) $U \subset \mathbb{R}^d$ by
$\diam{U} = \max_{p, q\in U} \lVert p - q \rVert$.
An \emph{indecisive} point is a special case of an uncertain point: it is a set of points $U =
\{p^1, \dots, p^k\}$, with each point $p^i \in \mathbb{R}^d$ for $i \in [k]$.
Similarly, an \emph{imprecise} point is a compact convex region $U \subset \mathbb{R}^d$.
We will often use disks or line segments as such regions.
Note that a precise point is a special case of an indecisive point (set of size one) and an
imprecise point (disk of radius zero).

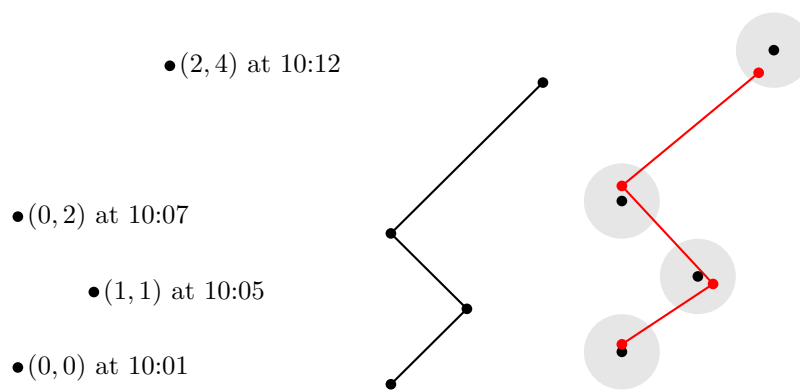
\begin{figure}
\centering
\begin{tikzpicture}
\fill[black]  (0, 0) node[right] {$(0, 0)$ at 10:01} circle[radius=2pt]
        (1, 1) node[right] {$(1, 1)$ at 10:05} circle[radius=2pt]
        (0, 2) node[right] {$(0, 2)$ at 10:07} circle[radius=2pt]
        (2, 4) node[right] {$(2, 4)$ at 10:12} circle[radius=2pt];
\end{tikzpicture}
\quad
\begin{tikzpicture}
\draw[thick] (0, 0) -- (1, 1) -- (0, 2) -- (2, 4);
\fill[black]  (0, 0) circle[radius=2pt]
        (1, 1) circle[radius=2pt]
        (0, 2) circle[radius=2pt]
        (2, 4) circle[radius=2pt];
\end{tikzpicture}
\quad
\begin{tikzpicture}
\fill[black,opacity=.1] (0, 0) circle[radius=.5]
                        (1, 1) circle[radius=.5]
                        (0, 2) circle[radius=.5]
                        (2, 4) circle[radius=.5];
\draw[red,thick] (0, 0.1) -- (1.2, 0.9) -- (0, 2.2) -- (1.8, 3.7);
\fill[black]    (0, 0) circle[radius=2pt]
                (1, 1) circle[radius=2pt]
                (0, 2) circle[radius=2pt]
                (2, 4) circle[radius=2pt];
\fill[red]  (0, 0.1) circle[radius=2pt]
            (1.2, 0.9) circle[radius=2pt]
            (0, 2.2) circle[radius=2pt]
            (1.8, 3.7) circle[radius=2pt];
\end{tikzpicture}
\caption{Left: Trajectory data. Centre: Polygonal curve on the data. Right: Imprecise curve
with disks as imprecision regions and real curve.}
\label{fig:traj}
\end{figure}

\subsection{Uncertain Curves and Distances}
Define an \emph{uncertain curve} as a sequence of uncertain points
$\mathcal{U} = \langle U_1, \dots, U_n\rangle$.
A \emph{realisation} $\pi \Subset \mathcal{U}$ of an uncertain curve is a polygonal curve
$\pi = \langle p_1, \dots, p_n\rangle$, where each $p_i$ is a realisation of the corresponding
uncertain point $U_i$.
We denote the set of all realisations of an uncertain curve $\mathcal{U}$ by $\Real{U}$ (see
Figure~\ref{fig:traj}).
In a probabilistic setting, we write $\pi \reals \mathcal{U}$ to denote that each point of $\pi$
gets drawn from the corresponding uncertainty region independently according to distribution
$\mathbb{P}$.

For uncertain curves $\mathcal{U}$ and $\mathcal{V}$, define the upper bound, lower bound, and
expected discrete Fr\'echet distance (and extend to continuous Fr\'echet distance $\frmax$,
$\frmin$, $\frexpp$ using $\fr$) as:
\begin{align*}
\dfrmax(\mathcal{U}, \mathcal{V}) = \max_{\pi \Subset \mathcal{U}, \sigma \Subset \mathcal{V}}
\dfr(\pi, \sigma)\,,&\qquad
\frmax(\mathcal{U}, \mathcal{V}) = \max_{\pi \Subset \mathcal{U}, \sigma \Subset \mathcal{V}}
\fr(\pi, \sigma)\,,\\
\dfrmin(\mathcal{U}, \mathcal{V}) = \min_{\pi \Subset \mathcal{U}, \sigma \Subset \mathcal{V}}
\dfr(\pi, \sigma)\,,&\qquad
\frmin(\mathcal{U}, \mathcal{V}) = \min_{\pi \Subset \mathcal{U}, \sigma \Subset \mathcal{V}}
\fr(\pi, \sigma)\,,\\
\dfrexpp(\mathcal{U}, \mathcal{V}) = \mathbb{E}_{\pi \reals \mathcal{U}, \sigma \reals \mathcal{V}}
[\dfr(\pi, \sigma)]\,,&\qquad
\frexpp(\mathcal{U}, \mathcal{V}) = \mathbb{E}_{\pi \reals \mathcal{U}, \sigma \reals \mathcal{V}}
[\fr(\pi, \sigma)]\,.
\end{align*}
If the distribution is clear from the context, we write $\frexp$ and $\dfrexp$.
The definitions above also apply if one of the curves is precise, as a precise curve is a
special case of an uncertain curve.

\section{Hardness Results}\label{sec:hardness}
In this section, we first discuss the hardness results for the upper bound and expected value of the
continuous and discrete Fr\'echet distance for indecisive and imprecise curves.
We then show hardness of finding the lower bound continuous Fr\'echet distance on imprecise curves.

\subsection{Upper Bound and Expected Fr\'echet Distance}\label{sec:proofs}
We present proofs of \np-hardness and \nump-hardness for the upper bound and expected Fr\'echet
distance for both indecisive and imprecise curves by showing polynomial-time reductions from
\textsc{CNF-SAT} (satisfiability of a boolean formula) and \textsc{\#CNF-SAT} (counting version).
We consider the upper bound problem for indecisive curves and then illustrate how the
construction can be used to show \nump-hardness for the expected Fr\'echet distance (both discrete
and continuous).
We then illustrate how the construction can be adapted to show hardness for imprecise curves.
All our constructions are in two dimensions.

\subsubsection{Upper Bound Fr\'echet Distance: Basic Construction}
Define the following problem:
\begin{problem}
\textsc{Upper Bound Discrete Fr\'echet:} Given two uncertain curves $\mathcal{U}$ and $\mathcal{V}$
and a threshold $\delta \in \mathbb{R}^+$, decide if $\dfrmax(\mathcal{U}, \mathcal{V}) > \delta$.
\end{problem}
We can similarly define its continuous counterpart, using $\frmax$ instead:
\begin{problem}
\textsc{Upper Bound Continuous Fr\'echet:} Given two uncertain curves $\mathcal{U}$ and $\mathcal{V}$
and a threshold $\delta \in \mathbb{R}^+$, decide if $\frmax(\mathcal{U}, \mathcal{V}) > \delta$.
\end{problem}

We first give some extra definitions to make the proofs clearer.
Suppose we are given a \cnfsat\ \emph{formula} $C$ with
\[C = \bigwedge_{i \in [n]} C_i\,,\qquad C_i = \bigvee_{j \in J \subseteq [m]} x_j \lor \bigvee_{k
\in K \subseteq [m] \setminus J} \neg x_k\quad\text{for all $i \in [n]$.}\]
Here $n$ and $m$ are the number of clauses and variables, respectively, and $x_j$ for any
$j \in [m]$ is a boolean variable.
Such a variable may be assigned `true' or `false'; an \emph{assignment} is a function $a: \{x_1,
\dots, x_m\} \to \{\True, \False\}$ that assigns a value to each variable, $a(x_j) = \True$ or $a
(x_j) = \False$ for any $j \in [m]$.
We denote by $C[a]$ the result of substituting $x_j \mapsto a(x_j)$ in $C$ for all $j \in [m]$.
As an aid to the reader, the problem we reduce from is:
\begin{problem}
\textsc{CNF-SAT:} Given a \cnfsat\ formula $C$, decide if there is an assignment $a$ such that
$C[a] = \True$.
\end{problem}

We pick some value $0 < \varepsilon < 0.25$.\footnote{This range is determined by the
relative distances in the construction.}
Construct a \emph{variable curve}, where each variable corresponds to an indecisive point
with locations $(0, 0.5 + \varepsilon)$ and $(0, -0.5 - \varepsilon)$; the locations are
interpreted as assigning the variable \True\ and \False.
Any realisation of the curve corresponds to a variable assignment.

\subparagraph*{Literal level}
Define a \emph{variable gadget,} where an indecisive point corresponds to a variable and is followed
by a precise point far away, to force synchronisation with the other curve:
\[\mathrm{VG}_j = \{(0, 0.5 + \varepsilon), (0, -0.5 - \varepsilon)\} \concat (2, 0)\,.\]

Consider a specific clause $C_i$ of the formula.
We define an \emph{assignment gadget} $\mathrm{AG}_{i, j}$ for each variable $x_j$ and clause $C_i$
depending on how the variable occurs in the clause.
\[\mathrm{AG}_{i, j} =
\begin{cases}
(0, -0.5) \concat (1, 0) & \text{if $x_j$ is a literal of $C_i$,}\\
(0, 0.5) \concat (1, 0) & \text{if $\neg x_j$ is a literal of $C_i$,}\\
(0, 0) \concat (1, 0) & \text{otherwise.}
\end{cases}
\]
Note that if assignment $x_j = \True$ makes a clause $C_i$ true, then the first precise point of the
corresponding assignment gadget appears at distance $1 + \varepsilon$ from the realisation
corresponding to setting $x_j = \True$ of the indecisive point in $\mathrm{VG}_j$.

We now show the relation between the gadgets.
To do so, we introduce the \emph{one-to-one coupling} as a valid coupling $c = \langle(p_1, q_1),
\dots, (p_r, q_r)\rangle$, where the condition is restricted to $(p_{s + 1}, q_{s + 1}) = (p_s + 1,
q_s + 1)$ for all $s \in [r - 1]$.
Necessarily, such a coupling can only exist for curves of equal length.
\begin{lemma}\label{lemma:distgadgetsvgag}
Suppose we are given a clause $C_i$ and a variable $x_j$ that both occur in the \cnfsat\ formula
$C$, and we restrict the set of valid couplings $\mathcal{C}$ to only contain one-to-one couplings.
We only get the discrete Fr\'echet distance equal to $1 + \varepsilon$ if the realisation of
$\mathrm{VG}_j$ we pick corresponds to the assignment of $x_j$ that ensures the clause $C_i$ is
satisfied; otherwise, the discrete Fr\'echet distance is $1$.
In other words, if we consider $\pi \Subset \mathrm{VG}_j$ that corresponds to setting $a(x_j)$ to
some values, then
\[\dfr(\pi, \mathrm{AG}_{i, j}) =
\begin{cases}
1 + \varepsilon & \text{iff $C_i[a] = \True$,}\\
1 & \text{otherwise.}
\end{cases}
\]
\end{lemma}
\begin{proof}
First of all, observe that as we only consider one-to-one couplings, the second points of both
gadgets must be coupled; the distance between them is $\lVert (2, 0) - (1, 0)\rVert = 1$; thus, the
discrete Fr\'echet distance between the curves must be at least $1$.

Now consider the possible realisations of $\mathrm{VG}_j$.
Say, we pick the realisation $(0, 0.5 + \varepsilon) \concat (2, 0)$, which corresponds to assigning
$a(x_j) = \True$.
If $x_j$ is a literal of $C_i$, so $C_i[a] = \True$, then by construction we know that $\mathrm{AG}_
{i, j}$ is $(0, -0.5) \concat (1, 0)$.
Since we consider only the one-to-one couplings, we must couple the first points together, yielding
the distance $\lVert (0, 0.5 + \varepsilon) - (0, -0.5)\rVert = 1 + \varepsilon > 1$, so the
discrete Fr\'echet distance in this case is $1 + \varepsilon$, and indeed we picked the assignment
that ensures that $C_i$ is satisfied.
If instead $\neg x_j$ is a literal of $C_i$, so $C_i[a] = \False$, then we know that
$\mathrm{AG}_{i, j}$ is $(0, 0.5) \concat (1, 0)$, and it is easy to see that, as $\lVert (0, 0.5 +
\varepsilon) - (0, 0.5)\rVert = \varepsilon < 1$, we get the discrete Fr\'echet distance of $1$, and
that we picked an assignment that does not ensure that $C_i$ is satisfied.

A symmetric argument can be applied when we consider the realisation $(0, -0.5 - \varepsilon)
\concat (2, 0)$ for $\mathrm{VG}_j$: if $\neg x_j$ is a literal of $C_i$, then we get the discrete
Fr\'echet distance of $1 + \varepsilon$ and we picked an assignment that surely satisfies $C_i$.

Finally, consider the case when $\mathrm{AG}_{i, j} = (0, 0) \concat (1, 0)$.
This implies that assigning a value to $x_j$ has no effect on $C_i$, i.e.\@ a literal involving
$x_j$ does not occur in $C_i$, so neither assignment (and neither realisation of $\mathrm{VG}_j$)
would ensure that $C_i$ is satisfied.
Also observe that $\lVert (0, 0.5 + \varepsilon) - (0, 0)\rVert = \lVert (0, -0.5 - \varepsilon) - 
(0, 0)\rVert = 0.5 + \varepsilon < 1$, so both realisations yield the discrete Fr\'echet distance of
$1$.

So, we can conclude that we get the distance $1 + \varepsilon$ if and only if the partial assignment
of a value to $x_j$ ensures that $C_i$ is satisfied; otherwise, we get the distance $1$.
\end{proof}

\subparagraph*{Clause level}
We can repeat the construction, yielding a \emph{variable clause gadget} and an \emph{assignment
clause gadget:}
\[\mathrm{VCG} = (-2, 0) \concat \Concat_{j \in [m]} \mathrm{VG}_j\,,\qquad
\mathrm{ACG}_i = (-1, 0) \concat \Concat_{j \in [m]} \mathrm{AG}_{i, j}\,.\]
Consider the Fr\'echet distance between the two gadgets.
Observe that matching a synchronisation point from one gadget with a non-synchronisation point in
the other yields a distance larger than $1 + \varepsilon$, whereas matching synchronisation points
pairwise and non-synchronisation points pairwise will yield the distance at most $1 + \varepsilon$.
So we only consider one-to-one couplings, i.e.\@ we match point $i$ on one curve to point $i$ on the
other curve, for all $i$.

Now, if a realisation corresponds to a satisfying assignment, then for some $x_j$ we have picked the
realisation that is opposite from the coupled point on the clause curve, yielding the bottleneck
distance of $1 + \varepsilon$.
If the realisation corresponds to a non-satisfying assignment, then the
synchronisation points establish the bottleneck, yielding the distance $1$.
So, we can clearly distinguish between a satisfying and a non-satisfying assignment for a clause.

It is crucial now that we show the following:
\begin{lemma}\label{lemma:onetoonevcgacg}
Given a \cnfsat\ formula $C$ containing some clause $C_i$ and $m$ variables $x_1, \dots, x_m$,
consider curves $\alpha_1 \concat \mathrm{VCG} \concat \alpha_1'$ and $\alpha_2 \concat
\mathrm{ACG}_i \concat \alpha_2'$ for arbitrary precise curves $\alpha_1$, $\alpha_1'$, $\alpha_2$,
$\alpha_2'$ with $\lvert \alpha_1\rvert = k$ and $\lvert \alpha_2\rvert = l$.
If some optimal coupling between $\alpha_1 \concat \mathrm{VCG} \concat \alpha_1'$ and $\alpha_2
\concat \mathrm{ACG}_i \concat \alpha_2'$ for any realisation of $\mathrm{VCG}$ has a pair $(k + 1,
l + 1)$, then there is an optimal coupling that has pairs $(k + s, l + s)$ for all $s \in [2m + 1]$,
i.e.\@ there is an optimal coupling that is one-to-one on the gadgets for any realisation of
$\mathrm{VCG}$.
\end{lemma}
\begin{proof}
Observe that both gadgets have exactly $2m + 1$ points.
Suppose the optimal coupling $\mathrm{Opt}$ has a pair $(k + 1, l + 1)$, so it matches the first
points of $\mathrm{VCG}$ and $\mathrm{ACG}_i$.
If $\mathrm{Opt}$ is already one-to-one for all $s \in [2m + 1]$, there is nothing to be done.
Suppose now that it is one-to-one until some $1 \leq r < 2m + 1$, so it has pairs $(k + s, l +
s)$ for all $s \in [r]$, but it does not have a pair $(k + (r + 1), l + (r + 1))$.
Then one of the following cases occurs.
\begin{itemize}
    \item $r = 2q + 2$ is even; then we know that the point $(2, 0)$ in $\mathrm{VG}_{q + 1}$ is
    not coupled with the point $(1, 0)$ in $\mathrm{AG}_{i, q + 1}$, but the preceding indecisive
    point is coupled with the assignment point.
    Then either $(2, 0)$ is coupled to an assignment point, with the distance at least $2$, or $(1,
    0)$ is coupled to an indecisive point, yielding the distance of $\sqrt{1 + (0.5 +
    \varepsilon)^2} > 1$.
    If we eliminate that pair and instead couple $(2, 0)$ to $(1, 0)$, we will still have a valid
    coupling and obtain the distance of $1$ on this pair; thus, the new coupling is not worse that
    the original one, and so it is also an optimal coupling that is one-to-one for all
    $s \in [r + 1]$.
    \item $r = 2q + 1$ is odd; then we know that the indecisive point in $\mathrm{VG}_{q + 1}$ is
    not coupled with the assignment point in $\mathrm{AG}_{i, q + 1}$, but the preceding $(2, 0)$
    and $(1, 0)$ (or $(-2, 0)$ and $(-1, 0)$) are coupled.
    Then either $\mathrm{Opt}$ has a pair of the indecisive point and $(1, 0)$, or it has a pair of
    the assignment point and $(2, 0)$.
    (The cases for $(-1, 0)$ and $(-2, 0)$ are symmetrical.)
    In either case, we want to eliminate that pair from the coupling and instead add the pair of the
    indecisive point and the assignment point, yielding a valid coupling that is one-to-one for all
    $s \in [r + 1]$.
    To complete the proof for this case, we need to show that such coupling is optimal.

    Consider the first possible coupling.
    The distance between the indecisive point and $(1, 0)$ is $\sqrt{1 + (0.5 + \varepsilon)^2}$,
    whereas the distance between the indecisive and the assignment point is $\varepsilon$, $0.5 +
    \varepsilon$, or $1 + \varepsilon$.
    As $\varepsilon < 0.25$, note that
    \begin{align*}
    0.25 + \varepsilon &> 2\varepsilon\\
    1 + 0.25 + \varepsilon + \varepsilon^2 &> 1 + 2\varepsilon + \varepsilon^2\\
    1 + (0.5 + \varepsilon)^2 &> (1 + \varepsilon)^2\\
    \sqrt{1 + (0.5 + \varepsilon)^2} &> 1 + \varepsilon\,,
    \end{align*}
    so our change to the optimal coupling will replace a pair with a pair with lower distance, so
    the new coupling is at least as good as the original one, and thus optimal.

    Now consider the second coupling.
    The distance between the assignment point and $(2, 0)$ is at least $2$, and $2 > 1 +
    \varepsilon > 0.5 + \varepsilon > \varepsilon$, so again our change yields an optimal coupling.
\end{itemize}
By induction, we conclude that the statement of the lemma holds.
\end{proof}

\noindent We can now use the two previous results to show the following.
\begin{lemma}\label{lemma:distvcgacg}
Given a \cnfsat\ formula $C$ containing some clause $C_i$ and $m$ variables $x_1, \dots, x_m$,
construct curves $\alpha_1 \concat \mathrm{VCG} \concat \alpha_1'$ and $\alpha_2 \concat
\mathrm{ACG}_i \concat \alpha_2'$ for arbitrary precise curves $\alpha_1$, $\alpha_1'$, $\alpha_2$,
$\alpha_2'$ with $\lvert \alpha_1\rvert = k$ and $\lvert \alpha_2\rvert = l$.
If some optimal coupling between $\alpha_1 \concat \mathrm{VCG} \concat \alpha_1'$ and $\alpha_2
\concat \mathrm{ACG}_i \concat \alpha_2'$ for any realisation of $\mathrm{VCG}$ has a pair $(k + 1,
l + 1)$ and $\dfr(\alpha_1, \alpha_2) \leq 1$ and $\dfr(\alpha_1', \alpha_2') \leq 1$, then the
discrete Fr\'echet distance between the curves is $1 + \varepsilon$ for realisations of
$\mathrm{VCG}$ that correspond to satisfying assignments for $C_i$, and $1$ for realisations that
do not.
In other words, if $\pi \Subset \mathrm{VCG}$ corresponds to assignment $a$ and we only consider
the restricted couplings, then \[\dfr(\alpha_1 \concat \pi \concat \alpha_1', \alpha_2 \concat
\mathrm{ACG}_i \concat \alpha_2') =
\begin{cases}
1 + \varepsilon & \text{iff $C_i[a] = \True$,}\\
1 & \text{otherwise.}
\end{cases}
\]
\end{lemma}
\begin{proof}
First of all, since some optimal coupling between $\alpha_1 \concat \mathrm{VCG} \concat \alpha_1'$
and $\alpha_2 \concat \mathrm{ACG}_i \concat \alpha_2'$ for any realisation of $\mathrm{VCG}$ has a
pair $(k + 1, l + 1)$, we can use Lemma~\ref{lemma:onetoonevcgacg} to find an optimal coupling
$\mathrm{Opt}$ that is one-to-one on the subcurves corresponding to the gadgets.
That means that we can, essentially, split the curves, if we consider only such restricted
couplings:
\begin{align*}
\dfr(\alpha_1 \concat \pi \concat \alpha_1', \alpha_2 \concat \mathrm{ACG}_i \concat \alpha_2')
&= \max(\dfr (\alpha_1, \alpha_2), \dfr(\pi, \mathrm{ACG}_i), \dfr(\alpha_1', \alpha_2'))\\
&= \max(1, \dfr(\pi, \mathrm{ACG}_i))\,,
\end{align*}
where the last equality follows from the fact that $\dfr(\pi, \mathrm{ACG}_i) \geq 1$, since the
first points are in a coupling and have the distance $1$, and from the assumption that
$\dfr(\alpha_1, \alpha_2) \leq 1$ and $\dfr(\alpha_1', \alpha_2') \leq 1$.
Note that here we do not restrict the coupling on $\alpha_1, \alpha_2$ and $\alpha_1', \alpha_2'$.

To obtain the end result, we need to consider the distance between $\pi$ and $\mathrm{ACG}_i$ under
a one-to-one coupling.
Using Lemma~\ref{lemma:distgadgetsvgag}, it is easy to see that if we have $a(x_j) = \True$ for
some variable $x_j$ and $x_j$ is a literal in $C_i$, then $C_i[a] = \True$, and $\dfr(\pi,
\mathrm{ACG}_i) = 1 + \varepsilon$; similarly, if $a(x_j) = \False$ for some variable $x_j$ and
$\neg x_j$ is a literal in $C_i$, then $C_i[a] = \True$, and $\dfr(\pi, \mathrm{ACG}_i) = 1 +
\varepsilon$.
If there is no such $x_j$, then $C_i[a] = \False$ and $\dfr(\pi, \mathrm{ACG}_i) = 1$.
We can thus conclude that the lemma holds.
\end{proof}

\subparagraph*{Formula level}
Next, we define the \emph{variable curve} and the \emph{clause curve} as follows:
\[\mathrm{VC} = (0, 0) \concat \mathrm{VCG} \concat (0, 0)\,,\qquad
\mathrm{CC} = \Concat_{i \in [n]} \mathrm{ACG}_i\,.\]
Observe that the synchronisation points at $(-2, 0)$ and $(-1, 0)$ ensure that for any optimal
coupling we match up $\mathrm{VCG}$ with some $\mathrm{ACG}_i$ as described before.
Also note that all the points on $\mathrm{CC}$ are within distance $1$ from $(0, 0)$.
Therefore, we can always pick any one of $n$ clauses to align with $\mathrm{VCG}$, and couple the
remaining points to $(0, 0)$; the bottleneck distance will then be determined by the distance
between $\mathrm{VCG}$ and the chosen $\mathrm{ACG}_i$.

Now consider a specific realisation of $\mathrm{VCG}$.
If the corresponding assignment does not satisfy $C$, then we can synchronise $\mathrm{VCG}$ with a
clause that is false to obtain the distance of~$1$. If the assignment corresponding to the
realisation satisfies all clauses, we must synchronise $\mathrm{VCG}$ with a satisfied clause, which
yields a distance of $1 + \varepsilon$.

We show the following important property of our construction.
\begin{lemma}\label{lemma:disttraj}
Given a \cnfsat\ formula $C$ with $n$ clauses and $m$ variables, construct the curves $\mathrm{VC}$
and $\mathrm{CC}$ as defined above and consider a realisation $(0, 0) \concat \pi \concat (0, 0)$ of
curve $\mathrm{VC}$, corresponding to some assignment $a$.
Then, under no restrictions on the couplings except those imposed by the definition,
\[\dfr((0, 0) \concat \pi \concat (0, 0), \mathrm{CC}) =
\begin{cases}
1 + \varepsilon & \text{iff $C[a] = \True$,}\\
1 & \text{iff $C[a] = \False$.}
\end{cases}
\]
In other words, the discrete Fr\'echet distance is $1 + \varepsilon$ if and only if the realisation
corresponds to a satisfying assignment, and is $1$ otherwise.
\end{lemma}
\begin{proof}
We can show this by proving that the premises of Lemma~\ref{lemma:distvcgacg} are satisfied.

First of all, note that all the points of $\mathrm{CC}$ are within distance $1$ from $(0, 0)$.
Furthermore, note that we can always give a coupling with the distance at most $1 + \varepsilon$:
couple $(0, 0)$ to $(-1, 0)$ from $\mathrm{ACG}_1$, then walk along realisation of $\mathrm{VCG}$
and $\mathrm{ACG}_1$ in a one-to-one coupling, and then couple the remaining points in $\mathrm{CC}$
to $(0, 0)$.
As all the points of $\mathrm{CC}$ are within distance $1$ from $(0, 0)$ and as this is otherwise
the construction of Lemma~\ref{lemma:distvcgacg}, this coupling yields the discrete Fr\'echet
distance of at most $1 + \varepsilon$ for any realisation of $\mathrm{VC}$.
Therefore, any coupling that has pairs further away than $1 + \varepsilon$ cannot be optimal.
Observe that the only point within that distance from $(-2, 0)$ is $(-1, 0)$.
Therefore, we only need to consider couplings that couple the first point of realisation of
$\mathrm{VCG}$ with the first point of some $\mathrm{ACG}_i$ as possibly optimal.
Thus, for each of the $n$ couplings we get, we can apply Lemma~\ref{lemma:distvcgacg}.
There are two cases to consider.
\begin{itemize}
    \item There is some gadget $\mathrm{ACG}_i$ with the distance $1$ to $\pi$ under the one-to-one
    coupling.
    Then we can choose that gadget to align with $\pi$ and couple all the other points in
    $\mathrm{CC}$ to $(0, 0)$ at the beginning or at the end of $\mathrm{VC}$ as suitable.
    As all the points of $\mathrm{CC}$ are within distance $1$ from $(0, 0)$, this coupling will
    yield distance $1$; as lower distance is impossible, this coupling is optimal, so then
    $\dfr((0, 0) \concat \pi \concat (0, 0), \mathrm{CC}) = 1$.
    Observe that by our construction this situation corresponds to the case when $C_i[a] = \False$,
    by Lemma~\ref{lemma:distvcgacg}, and so indeed $C[a] = \False$.
    \item The distance between any gadget $\mathrm{ACG}_i$ and $\pi$ under the one-to-one coupling
    is $1 + \varepsilon$.
    Then, no matter which gadget we choose to align with $\pi$, we will get the distance of $1 +
    \varepsilon$, so in this case $\dfr((0, 0) \concat \pi \concat (0, 0), \mathrm{CC}) = 1 +
    \varepsilon$.
    Note that, by our construction, this means that $C_i[a] = \True$ for all $i \in [n]$; therefore,
    indeed $C[a] = \True.$
\end{itemize}
As we have covered all the possible cases, we conclude that the lemma holds.
\end{proof}
We illustrate the gadgets of the construction in Figure~\ref{fig:gadgets}.
We also show an example of the correspondence between a boolean formula and our construction in
Figure~\ref{fig:real}.

{\def\epsln{0.15}
\begin{figure}
\centering
\begin{tikzpicture}[scale=2]
\draw[black,thin] (-1.7, -0.65) rectangle (2, 0.65) node[right,yshift=-12pt]{$\mathrm{ACG}$};
\draw[black,thin] (-0.8, -0.55) rectangle (1.5, 0.55) node[right,yshift=-12pt]{$\mathrm{AG}$};
\fill[red]   (0, 0) node[right] {$(0, 0)$} circle[radius=1pt]
             (0, 0.5) node[below left] {$(0, 0.5)$} circle[radius=1pt]
             (0, -0.5) node[above left] {$(0, -0.5)$} circle[radius=1pt]
             (1, 0) node[right] {$(1, 0)$} circle[radius=1pt]
             (-1, 0) node[left] {$(-1, 0)$} circle[radius=1pt];
\end{tikzpicture}

\vspace{\floatsep}

\begin{tikzpicture}[scale=2]
\draw[black,thin] (-2.3, -0.8) rectangle (2.8, 0.8) node[right,yshift=-12pt]{$\mathrm{VCG}$};
\draw[black,thin] (-1.2, -0.7) rectangle (2.3, 0.7) node[right,yshift=-12pt]{$\mathrm{VG}$};
\fill[olive] (0, 0.5+\epsln) node[below left] {$(0, 0.5 + \varepsilon)$} circle[radius=1pt]
             (0, -0.5-\epsln) node[above left] {$(0, -0.5 - \varepsilon)$} circle[radius=1pt]
             (2, 0) node[below] {$(2, 0)$} circle[radius=1pt]
             (-2, 0) node[below] {$(-2, 0)$} circle[radius=1pt];
\end{tikzpicture}
\caption{Illustration of the gadgets used in the basic construction.}
\label{fig:gadgets}
\end{figure}
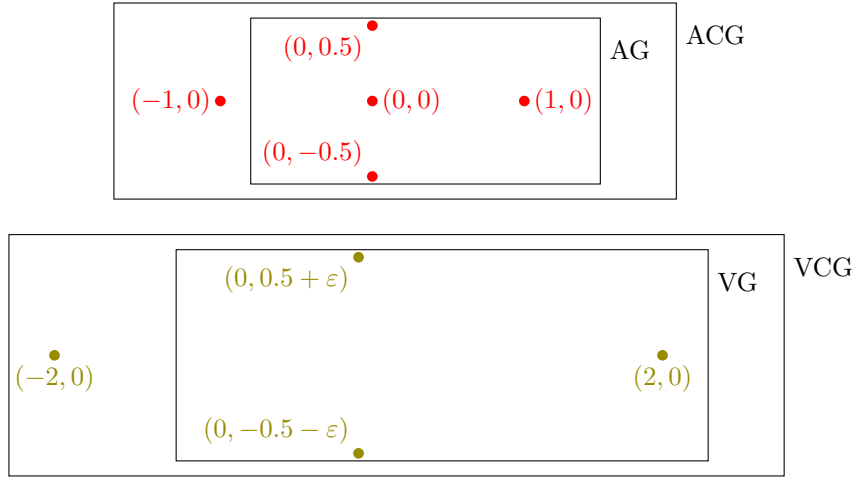}

{\def\rds{2pt}
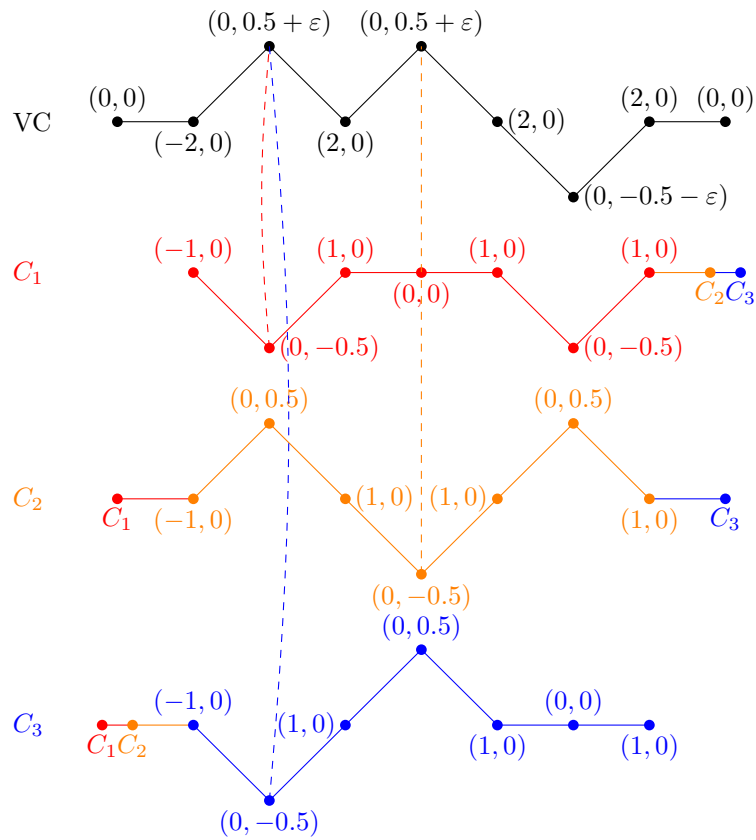
\begin{figure}
\centering
\begin{tikzpicture}
\draw[black,thin] (0, 10) -- (1, 10) -- (2, 11) -- (3, 10) -- (4, 11) --
                  (5, 10) -- (6, 9) -- (7, 10) -- (8, 10);
\fill[black]    (0, 10) node[above] {$(0, 0)$} circle[radius=\rds]
                (1, 10) node[below] {$(-2, 0)$} circle[radius=\rds]
                (2, 11) node[above] {$(0, 0.5 + \varepsilon)$} circle[radius=\rds]
                (3, 10) node[below] {$(2, 0)$} circle[radius=\rds]
                (4, 11) node[above] {$(0, 0.5 + \varepsilon)$} circle[radius=\rds]
                (5, 10) node[right] {$(2, 0)$} circle[radius=\rds]
                (6, 9) node[right] {$(0, -0.5 - \varepsilon)$} circle[radius=\rds]
                (7, 10) node[above] {$(2, 0)$} circle[radius=\rds]
                (8, 10) node[above] {$(0, 0)$} circle[radius=\rds];

\draw[red,thin] (1, 8) -- (2, 7) -- (3, 8) -- (5, 8) -- (6, 7) -- (7, 8);
\draw[orange,thin] (7, 8) -- (7.8, 8);
\draw[blue,thin] (7.8, 8) -- (8.2, 8);
\fill[red]      (1, 8) node[above] {$(-1, 0)$} circle[radius=\rds]
                (2, 7) node[right] {$(0, -0.5)$} circle[radius=\rds]
                (3, 8) node[above] {$(1, 0)$} circle[radius=\rds]
                (4, 8) node[below] {$(0, 0)$} circle[radius=\rds]
                (5, 8) node[above] {$(1, 0)$} circle[radius=\rds]
                (6, 7) node[right] {$(0, -0.5)$} circle[radius=\rds]
                (7, 8) node[above] {$(1, 0)$} circle[radius=\rds];

\draw[red,thin] (0, 5) -- (1, 5);
\draw[orange,thin] (1, 5) -- (2, 6) -- (4, 4) -- (6, 6) -- (7, 5);
\draw[blue,thin] (7, 5) -- (8, 5);
\fill[orange]   (1, 5) node[below] {$(-1, 0)$} circle[radius=\rds]
                (2, 6) node[above] {$(0, 0.5)$} circle[radius=\rds]
                (3, 5) node[right] {$(1, 0)$} circle[radius=\rds]
                (4, 4) node[below] {$(0, -0.5)$} circle[radius=\rds]
                (5, 5) node[left] {$(1, 0)$} circle[radius=\rds]
                (6, 6) node[above] {$(0, 0.5)$} circle[radius=\rds]
                (7, 5) node[below] {$(1, 0)$} circle[radius=\rds];

\draw[red,thin] (-0.2, 2) -- (0.2, 2);
\draw[orange,thin] (0.2, 2) -- (1, 2);
\draw[blue,thin] (1, 2) -- (2, 1) -- (4, 3) -- (5, 2) -- (7, 2);
\fill[blue]     (1, 2) node[above] {$(-1, 0)$} circle[radius=\rds]
                (2, 1) node[below] {$(0, -0.5)$} circle[radius=\rds]
                (3, 2) node[left] {$(1, 0)$} circle[radius=\rds]
                (4, 3) node[above] {$(0, 0.5)$} circle[radius=\rds]
                (5, 2) node[below] {$(1, 0)$} circle[radius=\rds]
                (6, 2) node[above] {$(0, 0)$} circle[radius=\rds]
                (7, 2) node[below] {$(1, 0)$} circle[radius=\rds];

\fill[red]      (0, 5) node[below] {$C_1$} circle[radius=\rds]
                (-0.2, 2) node[below] {$C_1$} circle[radius=\rds]
                (-1.5, 8) node[right] {$C_1$};
\fill[orange]   (7.8, 8) node[below] {$C_2$} circle[radius=\rds]
                (0.2, 2) node[below] {$C_2$} circle[radius=\rds]
                (-1.5, 5) node[right] {$C_2$};
\fill[blue]     (8.2, 8) node[below] {$C_3$} circle[radius=\rds]
                (8, 5) node[below] {$C_3$} circle[radius=\rds]
                (-1.5, 2) node[right] {$C_3$};
\fill[black]    (-1.5, 10) node[right] {$\mathrm{VC}$};

\draw[dashed,thin,blue] (2, 1) to[bend right=5] (2, 11);
\draw[dashed,thin,orange] (4, 4) -- (4, 11);
\draw[dashed,thin,red] (2, 7) to[bend left=5] (2, 11);
\end{tikzpicture}
\caption{Realisation of $\mathrm{VC}$ for assignment $x_1 = \True$, $x_2 = \True$, $x_3 = \False$
and the $\mathrm{CC}$ for formula $C = (x_1 \lor x_3) \land (\neg x_1 \lor x_2 \lor \neg x_3) \land
(x_1 \lor \neg x_2)$.
Note that $C = \True$ with the given variable assignment.
Also note that we can choose any of $C_1$, $C_2$, $C_3$ to align with $\mathrm{VC}$; we always get
the bottleneck distance of $1 + \varepsilon$, as all three are satisfied, so here
$\dfr(\mathrm{VC}, \mathrm{CC}) = 1 + \varepsilon$.}
\label{fig:real}
\end{figure}}

\subsubsection{Upper Bound Discrete Fr\'echet Distance on Indecisive Points}
\begin{theorem}\label{thm:ub_ind_np}
The problem \textsc{Upper Bound Discrete Fr\'echet} for indecisive curves is \np-complete.
\end{theorem}
\begin{proof}
First of all, observe that if two realisations of length $n$ and $m$ are given as a certificate
for a `Yes'-instance of the problem, then one can verify the solution by computing discrete
Fr\'echet distance between the realisations and checking that it is indeed larger than some
threshold $\delta$.
The computation can be done in time $\Theta(mn)$, using the algorithm proposed by Eiter and
Mannila~\cite{eiter:1994}.
Therefore, the problem is in \np.

Now suppose we are given an instance of \textsc{CNF-SAT}, i.e.\@ a \cnfsat\ formula $C$ with $n$
clauses and $m$ variables.
We construct the curves $\mathrm{VC}$ and $\mathrm{CC}$, as described previously, and get an
instance of \textsc{Upper Bound Discrete Fr\'echet} on curves $\mathrm{VC}$ and $\mathrm{CC}$
and threshold $\delta = 1$.
If the answer is `Yes', then we also output `Yes' as an answer to \textsc{CNF-SAT}; otherwise, we
output `No'.

Using Lemma~\ref{lemma:disttraj}, we can see that if there is some assignment $a$ such that $C[a]
= \True$, then for the corresponding realisation the discrete Fr\'echet distance is $1 +
\varepsilon$; the other way around, if for some realisation we get the distance $1 + \varepsilon$,
then by our construction all the clauses are satisfied and $C[a] = \True$; and so
$\dfrmax(\mathrm{VC}, \mathrm{CC}) = 1 + \varepsilon$.
On the other hand, if there is no such assignment $a$, then for any assignment $a$ there is some
$C_i$ with $C_i[a] = \False$, yielding $C[a] = \False$, and also for any realisation of
$\mathrm{VC}$ there is some gadget $\mathrm{ACG}_i$ that yields the discrete Fr\'echet distance of
$1$; and so $\dfrmax(\mathrm{VC}, \mathrm{CC}) = 1$.
Therefore, the formula $C$ is satisfiable if and only if $\dfrmax(\mathrm{VC}, \mathrm{CC}) > 1$,
and so our answer is correct.

Furthermore, observe that the curves have $2m + 2$ and $2mn + n$ points, respectively, and so the
instance of \textsc{Upper Bound Discrete Fr\'echet} that gives the answer to \textsc{CNF-SAT} can be
constructed in polynomial time.
Thus, we conclude that \textsc{Upper Bound Discrete Fr\'echet} for indecisive curves is \np-hard;
combining it with the first part of the proof shows that it is \np-complete.
\end{proof}

\subsubsection{Upper Bound Fr\'echet Distance on Indecisive Points}
We use the same construction as for the discrete Fr\'echet distance.
To do the same proof, we need to present arguments for the continuous case that lead up to an
alternative to Lemma~\ref{lemma:disttraj}.
For the arguments to work, we need to further restrict the range of $\varepsilon$ to be $[0.12,
0.25)$.

{\def\epsln{0.15}
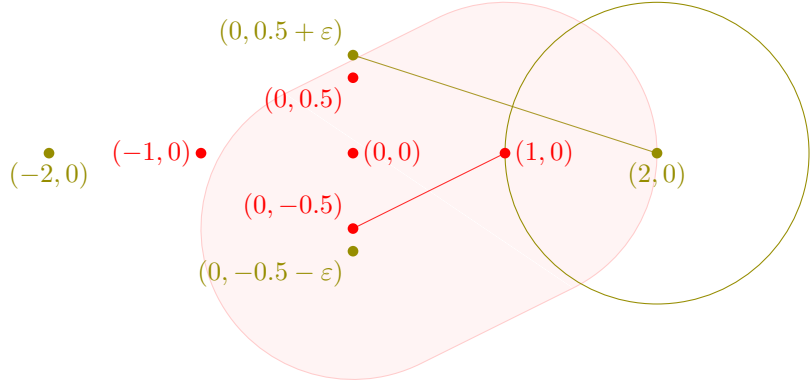
\begin{figure}
\centering
\begin{tikzpicture}[scale=2]
\draw[olive] (2, 0) -- (0, 0.5+\epsln);
\draw[red] (1, 0) -- (0, -0.5);

\path (0, -0.5) +(116.57:1) coordinate (a)
(1, 0) +(116.57:1) coordinate (b);
\draw[red,fill=red!20,opacity=0.2]  (a) -- (b) arc(116.57:-63.43:1) coordinate (c)
                                    (a) arc(116.57:296.57:1) -- (c);
\draw[olive,thin] (2, 0) circle[radius=1];

\fill[red]   (0, 0) node[right] {$(0, 0)$} circle[radius=1pt]
             (0, 0.5) node[below left] {$(0, 0.5)$} circle[radius=1pt]
             (0, -0.5) node[above left] {$(0, -0.5)$} circle[radius=1pt]
             (1, 0) node[right] {$(1, 0)$} circle[radius=1pt]
             (-1, 0) node[left] {$(-1, 0)$} circle[radius=1pt];
\fill[olive] (0, 0.5+\epsln) node[above left] {$(0, 0.5 + \varepsilon)$} circle[radius=1pt]
             (0, -0.5-\epsln) node[below left] {$(0, -0.5 - \varepsilon)$} circle[radius=1pt]
             (2, 0) node[below] {$(2, 0)$} circle[radius=1pt]
             (-2, 0) node[below] {$(-2, 0)$} circle[radius=1pt];
\end{tikzpicture}
\caption{Construction for $\varepsilon = \epsln$. Shaded red area shows the points within distance
$1$ from the segment $(0, -0.5) \concat (1, 0)$. Observe that $(0, 0.5 + \varepsilon)$ is outside
that region, and that $(1, 0)$ is the only red point within distance $1$ from $(2, 0)$.}
\label{fig:constr}
\end{figure}}

Consider the construction drawn in Figure~\ref{fig:constr}.
The key points here are that $(0, 0.5 + \varepsilon)$ is far from any point on the clause
curve, and that $(2, 0)$ is only close enough to $(1, 0)$.
We can present a lemma similar to Lemma~\ref{lemma:distgadgetsvgag}.
\begin{lemma}\label{lemma:cont_disatgadgetsvgag}
Given a clause $C_i$ and a variable $x_j$ that both occur in the \cnfsat\ formula $C$, we only get
the Fr\'echet distance equal to $(1 + \varepsilon) \cdot \frac{2}{\sqrt{5}}$ if the realisation of
$\mathrm{VG}_j$ we pick corresponds to the assignment of $x_j$ that ensures the clause $C_i$ is
satisfied; otherwise, the Fr\'echet distance is $1$.
In other words, if we consider $\pi \Subset \mathrm{VG}_j$ that corresponds to setting $a(x_j)$ to
some values, then
\[\fr(\pi, \mathrm{AG}_{i, j}) =
\begin{cases}
(1 + \varepsilon) \cdot \frac{2}{\sqrt{5}} & \text{iff $C_i[a] = \True$,}\\
1 & \text{otherwise.}
\end{cases}
\]
\end{lemma}
\begin{proof}
Consider the possible realisations of $\mathrm{VG}_j$.
Suppose we pick the realisation $(0, 0.5 + \varepsilon) \concat (2, 0)$, which corresponds to
assigning $a(x_j) = \True$.
If $x_j$ is a literal in $C_i$, so $C_i = \True$, then by construction we know that
$\mathrm{AG}_{i, j}$ is $(0, -0.5) \concat (1, 0)$.
As noted in Figure~\ref{fig:constr}, the distance between $(0, 0.5 + \varepsilon)$ and any point on
$(0, -0.5) \concat (1, 0)$ is larger than $1$.
To be more specific, the distance between the point $(x, y)$ and the line defined by $(x_1, y_1)
\concat (x_2, y_2)$ can be determined using a standard formula as
\[d = \frac{\lvert x(y_2 - y_1) - y(x_2 - x_1) + x_2y_1 - x_1y_2\rvert}{\sqrt{(x_2 - x_1)^2 + (y_2 -
y_1)^2}}\,.\]
In our case, we get
\[d = \frac{\lvert 0 - (0.5 + \varepsilon) \cdot (1 - 0) - 1 \cdot 0.5 - 0 \rvert}{\sqrt{(1 - 0)^2 +
(0 + 0.5)^2}} = \frac{2 \cdot (1 + \varepsilon)}{\sqrt{5}}\,.\]
As the point $(0, 0.5 + \varepsilon)$ must be coupled to some point on $\mathrm{AG}_{i, j}$, the
Fr\'echet distance we get in this case cannot be smaller than $d$.
Furthermore, it is easy to see that the point $(0, 0.5 + \varepsilon)$ is the furthest point from
$\mathrm{AG}_{i, j}$; thus, we get that Fr\'echet distance is exactly $d$.

On the other hand, if $\neg x_j$ is a literal in $C_i$, then by construction we know that
$\mathrm{AG}_{i, j}$ is $(0, 0.5) \concat (1, 0)$.
As noted in Figure~\ref{fig:constr}, the distance between $(2, 0)$ and any point on $(0, 0.5)
\concat (1, 0)$ is at least $1$, with the smallest distance achieved at $(1, 0)$.
It is clear that this is the furthest pair of points on the two gadgets in this case; thus, we get
the Fr\'echet distance of $1$.

A symmetric argument can be applied when we consider the realisation $(0, -0.5 - \varepsilon)
\concat (2, 0)$ for $\mathrm{VG}_j$: if $\neg x_j$ is a literal in $C_i$, then we get the
Fr\'echet distance of $d$ and we picked an assignment that satisfies $C_i$; and in the other case,
we get that $C_i[a]$ is not necessarily satisfied and the Fr\'echet distance is $1$.

Finally, consider the case when $\mathrm{AG}_{i, j} = (0, 0) \concat (1, 0)$.
Again, this implies that assigning a value to $x_j$ has no effect on $C_i$, so neither assignment
(and neither realisation of $\mathrm{VG}_j$) would ensure that $C_i$ is satisfied.
Also observe that both realisations give rise to curves that are entirely within distance $1$
of $(0, 0) \concat (1, 0)$, yielding the Fr\'echet distance of $1$.
\end{proof}

\noindent We can now naturally get a lemma similar to Lemma~\ref{lemma:distvcgacg}.
\begin{lemma}\label{lemma:cont_distvcgacg}
Given a \cnfsat\ formula $C$ containing some clause $C_i$ and $m$ variables $x_1, \dots, x_m$,
construct curves $\alpha_1 \concat \mathrm{VCG} \concat \alpha_1'$ and $\alpha_2 \concat
\mathrm{ACG}_i \concat \alpha_2'$ for arbitrary precise curves $\alpha_1$, $\alpha_1'$, $\alpha_2$,
$\alpha_2'$ with $\lvert \alpha_1\rvert = k$ and $\lvert \alpha_2\rvert = l$.
If some optimal coupling $\phi_1, \phi_2$ between $\alpha_1 \concat \mathrm{VCG} \concat \alpha_1'$
and $\alpha_2 \concat \mathrm{ACG}_i \concat \alpha_2'$ for any realisation of $\mathrm{VCG}$ has
some value $t$ such that $\phi_1(t) = k + 1$ and $\phi_2(t) = l + 1$ and $\fr(\alpha_1, \alpha_2)
\leq 1$ and $\fr(\alpha_1', \alpha_2') \leq 1$, then the Fr\'echet distance between the curves is
$(1 + \varepsilon) \cdot \frac{2}{\sqrt{5}}$ for realisations of $\mathrm{VCG}$ that correspond to
satisfying assignments for $C_i$, and $1$ for realisations that do not.
In other words, if $\pi \Subset \mathrm{VCG}$ corresponds to assignment $a$ and we only consider
the restricted couplings, then
\[\fr(\alpha_1 \concat \pi \concat \alpha_1', \alpha_2 \concat \mathrm{ACG}_i \concat \alpha_2') =
\begin{cases}
(1 + \varepsilon) \cdot \frac{2}{\sqrt{5}} & \text{iff $C_i[a] = \True$,}\\
1 & \text{otherwise.}
\end{cases}
\]
\end{lemma}
\begin{proof}
First of all, observe that as we traverse $\mathrm{VCG}$, we need to couple $(2, 0)$ to $(1, 0)$ to
obtain an optimal coupling.
Therefore, essentially, the traversal can be split into $m$ parts, each of which corresponds to
traversing $\mathrm{VG}_j$ and $\mathrm{AG}_{i, j}$ at the same time for all $j \in [m]$.
We can use Lemma~\ref{lemma:cont_disatgadgetsvgag} to note that if some variable $x_j$ is
assigned a value that makes clause $C_i$ satisfied, then the Fr\'echet distance becomes $(1 +
\varepsilon) \cdot \frac{2}{\sqrt{5}}$; if that is not the case for any variables, then we can
traverse the entire curve, as well as $\alpha_1$ and $\alpha_1'$ by linearly interpolating our
position between the vertices of the curves and otherwise using the coupling of the discrete case,
while staying within distance $1$ of the other curve, yielding the Fr\'echet distance of~$1$.
The distance also cannot be smaller than $1$ due to coupling of $(2, 0)$ and $(1, 0)$.
\end{proof}
While this proof is a bit less formal than that of Lemma~\ref{lemma:distvcgacg}, its validity
should be sufficiently clear from geometric considerations described earlier in this section.

Now we can provide a lemma that mirrors Lemma~\ref{lemma:disttraj}.
\begin{lemma}\label{lemma:cont_disttraj}
Given a \cnfsat\ formula $C$ with $n$ clauses and $m$ variables, construct the curves
$\mathrm{VC}$ and $\mathrm{CC}$ as defined above and consider a realisation $(0, 0) \concat \pi
\concat (0, 0)$ of curve $\mathrm{VC}$, corresponding to some assignment $a$.
Then
\[\fr((0, 0) \concat \pi \concat (0, 0), \mathrm{CC}) =
\begin{cases}
(1 + \varepsilon) \cdot \frac{2}{\sqrt{5}} & \text{iff $C[a] = \True$,}\\
1 & \text{iff $C[a] = \False$.}
\end{cases}
\]
In other words, the Fr\'echet distance is $(1 + \varepsilon) \cdot \frac{2}{\sqrt{5}}$ if and only
if the realisation $\pi$ corresponds to a satisfying assignment, and is $1$ otherwise.
\end{lemma}
\begin{proof}
First of all, observe that any point of $\mathrm{CC}$ is within distance $1$ of $(0, 0)$;
furthermore, when starting to traverse $\pi$, we must couple $(-2, 0)$ to $(-1, 0)$ in an optimal
coupling.
Thus, the premise of Lemma~\ref{lemma:cont_distvcgacg} is satisfied, and, using reasoning similar to
that of Lemma~\ref{lemma:disttraj}, we observe that an optimal coupling chooses one of the clauses
to traverse in parallel with the variable curve, and so if there is a clause that is not satisfied,
then we get the Fr\'echet distance of $1$, and if all of them are satisfied, then all of them yield
the Fr\'echet distance of $(1 + \varepsilon) \cdot \frac{2}{\sqrt{5}}$.
Thus, we conclude that the lemma holds.
\end{proof}

\noindent Finally, we can show the main result.
\begin{theorem}\label{thm:cont_ub_ind_np}
The problem \textsc{Upper Bound Continuous Fr\'echet} for indecisive curves is \np-complete.
\end{theorem}
\begin{proof}
First of all, observe that if two realisations of length $n$ and $m$ are given as a certificate
for a `Yes'-instance of the problem, then one can verify the solution by checking that the
Fr\'echet distance between the realisations is indeed larger than some threshold $\delta$.
The computation can be done in time $\Theta(mn)$, using the algorithm proposed by Alt and
Godau~\cite{alt:1995, godau:1991}.
Therefore, the problem is in \np.

Now suppose we are given an instance of \textsc{CNF-SAT}, i.e.\@ a \cnfsat\ formula $C$ with $n$
clauses and $m$ variables.
We construct the curves $\mathrm{VC}$ and $\mathrm{CC}$, as described previously, and get an
instance of \textsc{Upper Bound Continuous Fr\'echet} on curves $\mathrm{VC}$ and $\mathrm{CC}$ and
threshold $\delta = 1$.
If the answer is `Yes', then we also output `Yes' as an answer to \textsc{CNF-SAT}; otherwise,
we output `No'.

Using Lemma~\ref{lemma:cont_disttraj}, we can see that if there is some assignment $a$ such that
$C[a] = \True$, then for the corresponding realisation the Fr\'echet distance is $(1 + \varepsilon)
\cdot \frac{2}{\sqrt{5}}$; the other way around, if for some realisation we get the distance $(1 +
\varepsilon) \cdot \frac{2}{\sqrt{5}}$, then by our construction all the clauses are satisfied and
$C[a] = \True$; and so $\frmax(\mathrm{VC}, \mathrm{CC}) = (1 + \varepsilon) \cdot
\frac{2}{\sqrt{5}}$.
On the other hand, if there is no such assignment $a$, then for any assignment $a$ there is some
$C_i$ with $C_i[a] = \False$, yielding $C[a] = \False$, and also for any realisation of
$\mathrm{VC}$ there is some gadget $\mathrm{ACG}_i$ that yields the Fr\'echet distance of $1$; and
so $\frmax(\mathrm{VC}, \mathrm{CC}) = 1$.
Therefore, the formula $C$ is satisfiable if and only if $\frmax(\mathrm{VC}, \mathrm{CC}) > 1$, and
so our answer to the \textsc{CNF-SAT} instance is correct.

Furthermore, as before, the instance of \textsc{Upper Bound Discrete Fr\'echet} that gives the
answer to \textsc{CNF-SAT} can be constructed in polynomial time.
Thus, we conclude that \textsc{Upper Bound Continuous Fr\'echet} for indecisive curves is \np-hard;
combining it with the first part of the proof shows that it is \np-complete.
\end{proof}

\subsubsection{Expected Fr\'echet Distance on Indecisive Points}
We show that finding expected discrete Fr\'echet distance is \nump-hard by providing a
polynomial-time reduction from \textsc{\#CNF-SAT}, i.e.\@ the problem of finding the number of
satisfying assignments to a \cnfsat\ formula.
Define the following problem and its continuous counterpart:
\begin{problem}
\textsc{Expected Discrete Fr\'echet:} Find $\dfrexpp[U](\mathcal{U}, \mathcal{V})$ for uncertain
curves $\mathcal{U}$ and $\mathcal{V}$.
\end{problem}
\begin{problem}
\textsc{Expected Continuous Fr\'echet:} Find $\frexpp[U](\mathcal{U}, \mathcal{V})$ for uncertain
curves $\mathcal{U}$ and $\mathcal{V}$.
\end{problem}
The main idea is to derive an expression for the number of satisfying assignments in
terms of $\dfrexpp[U](\mathrm{VC}, \mathrm{CC})$.
This works, since there is a one-to-one correspondence between boolean variable assignment and a
choice of realisation of $\mathrm{VC}$, so counting the number of satisfying assignments corresponds
to finding the proportion of realisations yielding large Fr\'echet distance.
We can establish the result for \textsc{Expected Continuous Fr\'echet} similarly.
\begin{theorem}\label{thm:exp_ind_nump_full}
The problems \textsc{Expected Discrete Fr\'echet} and \textsc{Expected Continuous Fr\'echet} for
indecisive curves are \nump-hard.
\end{theorem}
\begin{proof}
Suppose we are given an instance of the \textsc{\#CNF-SAT} problem, i.e.\@ a \cnfsat\ formula $C$
with $n$ clauses and $m$ variables.
Denote the (unknown) number of satisfying assignments of $C$ by $N$.
We can construct indecisive curves $\mathrm{VC}$ and $\mathrm{CC}$ in the same way as
previously.
We then get an instance of \textsc{Expected Discrete Fr\'echet} on indecisive curves under
uniform distribution.
Assuming we solve it and get $\dfrexpp[U](\mathrm{VC}, \mathrm{CC}) = \mu$, we can now compute $N$:
\[N = (\mu - 1) \cdot \frac{2^m}{\varepsilon}\,.\]
$N$ is then the output for the instance of \textsc{\#CNF-SAT} that we were given.
Clearly, construction of the curves can be done in polynomial time; so can
the computation of $N$; hence, the reduction takes polynomial time.

We still need to show that the result we obtain is correct.
For each assignment, there is exactly one realisation of the curve $\mathrm{VC}$.
Furthermore, as we choose the realisation of each indecisive point uniformly and independently, all
the realisations of $\mathrm{VC}$ have equal probability of $2^{-m}$.
There are $N$ satisfying assignments; and each of the corresponding realisations yields the discrete
Fr\'echet distance of $1 + \varepsilon$.
In the remaining $2^m - N$ cases, the distance is $1$.
Using the definition of expected value, we can derive
\[\mu = \dfrexpp[U](\mathrm{VC}, \mathrm{CC}) = N \cdot 2^{-m} \cdot (1 + \varepsilon) + (2^m - N)
\cdot 2^{-m} \cdot 1 = 1 + \frac{N \cdot \varepsilon}{2^m}\,.\]
Then it is easy to see that indeed $N = (\mu - 1) \cdot \frac{2^m}{\varepsilon}$.
So, we get the correct number of satisfying assignments, if we know the expected value under
uniform distribution.
Therefore, \textsc{Expected Discrete Fr\'echet} for indecisive curves is \nump-hard.

One can derive a very similar formula to show that \textsc{Expected Continuous Fr\'echet} is also
\nump-hard for indecisive curves.
We can use almost the same reduction as for the discrete case, so given an instance of
\textsc{\#CNF-SAT} (\cnfsat\ formula $C$ with $n$ clauses and $m$ variables), we construct the two
curves, solve \textsc{Expected Continuous Fr\'echet} to obtain the value of $\mu$, and compute
\[N = 2^m \cdot (\mu - 1) \cdot \frac{\sqrt{5}}{2(1 + \varepsilon) - \sqrt{5}}\]
as the output for \textsc{\#CNF-SAT}.

To show that the output is correct, note that
\begin{align*}
\mu &= 2^{-m} \cdot N \cdot \frac{2}{\sqrt{5}} \cdot (1 + \varepsilon) + 2^{-m} \cdot (2^m - N)
\cdot 1\\
&= 1 + 2^{-m} \cdot N \cdot \left(\frac{2}{\sqrt{5}} (1 + \varepsilon) - 1\right)\,,
\end{align*}
so we can express $N$ as
\[N = 2^m \cdot (\mu - 1) \cdot \frac{\sqrt{5}}{2(1 + \varepsilon) - \sqrt{5}}\,.\]
Again, the reduction is correct and can be done in polynomial time, so \textsc{Expected Continuous
Fr\'echet} for indecisive curves is \nump-hard.
\end{proof}

\subsubsection{Upper Bound Discrete Fr\'echet Distance on Imprecise Points}
Here we consider imprecise points modelled as disks and as line segments; the results and their
proofs turn out to be very similar.
We denote the disk with the centre at $p \in \mathbb{R}^d$ and radius $r \geq 0$ as $D(p, r)$.
We denote the line segment between points $p_1$ and $p_2$ by $S(p_1, p_2)$.

\paragraph*{Disks}
We use a construction very similar to that of the indecisive points case, except now we change the
gadget containing a non-degenerate indecisive point so that it contains a non-degenerate imprecise
point, for all $j \in [m]$:
\[\mathrm{VG}_j = D((0, 0), 0.5 + \varepsilon) \concat (2, 0)\,.\]
Essentially, the two original indecisive points are now located on the points realising the diameter
of the disk.

We can reuse the proof leading up to Theorem~\ref{thm:ub_ind_np}, if we can show the following:
\begin{lemma}\label{lemma:imp_reuse_discr}
Suppose $\dfrmax(\mathrm{VC}, \mathrm{CC}) = \nu$.
If one considers all realisations $\pi$ of $\mathrm{VC}$ that yield $\dfr(\pi, \mathrm{CC}) = \nu$,
then among them there will always be a realisation that only places the imprecise point realisations
at either $(0, 0.5 + \varepsilon)$ or $(0, -0.5 - \varepsilon)$.
\end{lemma}
\begin{proof}
First of all, note that the points $(2, 0)$ and $(1, 0)$ are still in the curves in the same quality
as before, so they must be coupled, and hence the lowest discrete Fr\'echet distance achievable with
any realisation is $1$.

Now consider a realisation of an imprecise point.
Suppose that all the clause assignment points for that imprecise point are placed at $(0, -0.5)$.
Then geometrically it is obvious that the distance is maximised by placing the realisation at $(0,
0.5 + \varepsilon)$; if there is a realisation that achieves the best possible value $\nu$ without
doing this, then we can move this point and still get $\nu$.

Suppose that some clause assignment points are at $(0, -0.5)$ and some at $(0, 0.5)$.
As the realisation comes from the disk of radius $0.5 + \varepsilon$, there is no realisation that
is further than $1$ away from both assignment points; therefore, to maximise the distance we have to
choose one of the two locations, and then the previous case applies.

So, it is clear that, from an arbitrary optimal realisation, moving to the (correct) indecisive
point realisation will still yield an optimal realisation for the maximum discrete Fr\'echet
distance; thus, the statement of the lemma holds.
\end{proof}

\paragraph*{Line Segments}
We use a very similar construction, except now we change the gadget to be, for all $j \in [m]$:
\[\mathrm{VG}_j = S((0, -0.5 - \varepsilon), (0, 0.5 + \varepsilon)) \concat (2, 0)\,.\]
Again, the two original indecisive points are now located on the ends of the segment; moreover, the
segment is a strict subset of the disk.

We can state a similar lemma.
\begin{lemma}\label{lemma:imp_reuse_discr_s}
Suppose $\dfrmax(\mathrm{VC}, \mathrm{CC}) = \nu$.
If one considers all realisations $\pi$ of $\mathrm{VC}$ that yield $\dfr(\pi, \mathrm{CC}) = \nu$,
then among them there will always be a realisation that only places the imprecise point realisations
at either $(0, 0.5 + \varepsilon)$ or $(0, -0.5 - \varepsilon)$.
\end{lemma}
\begin{proof}
Since the line segments include these points and are subsets of the disks, the statement of
Lemma~\ref{lemma:imp_reuse_discr} immediately yields this result.
\end{proof}

So, now we can state the following for both models:
\begin{theorem}\label{thm:ub_imp_np}
The problem \textsc{Upper Bound Discrete Fr\'echet} for imprecise curves modelled as line
segments or as disks is \np-complete.
\end{theorem}
\begin{proof}
As shown in the proof of Theorem~\ref{thm:ub_ind_np}, the problem is in \np\ for any uncertain
curves.

Furthermore, as we have shown in Lemma~\ref{lemma:imp_reuse_discr} and
Lemma~\ref{lemma:imp_reuse_discr_s}, for the same \cnfsat\ formula the upper bound discrete
Fr\'echet distance on indecisive and imprecise points is equal for our construction.
So, trivially, \textsc{Upper Bound Discrete Fr\'echet} is \np-hard for imprecise curves.
Therefore, it is \np-complete.
\end{proof}

\subsubsection{Upper Bound Fr\'echet Distance on Imprecise Points}
We use exactly the same construction as in the previous section.
The argument here follows the previous ones very closely, so we can immediately state the following
theorem.

\begin{theorem}\label{thm:cont_ub_imp_np}
The problem \textsc{Upper Bound Continuous Fr\'echet} for imprecise curves modelled as line
segments or as disks is \np-complete.
\end{theorem}
\begin{proof}
Note that we can apply exactly the same argument as the one in Lemma~\ref{lemma:imp_reuse_discr} and
Lemma~\ref{lemma:imp_reuse_discr_s} to reduce this problem to the one on indecisive points.
Then, we can apply the same argument as in the proof of Theorem~\ref{thm:ub_imp_np} to conclude that
the problem is \np-hard.

We have shown in Theorem~\ref{thm:cont_ub_ind_np} that the problem is in \np\ for all uncertain
curves; thus, we conclude that it is \np-complete.
\end{proof}

\subsubsection{Expected Discrete Fr\'echet Distance on Imprecise Points}
We can also consider the value of expected Fr\'echet distance on imprecise points.
We show the result only for points modelled as line segments; in principle, we believe that for
disks a similar result holds, but the specifics of our reduction do not allow for clean
computations.

We cannot immediately use our construction: we treat subsegments at the ends of the imprecision
segments as \True\ and \False, but we have no interpretation for points in the centre part of a
segment.
So, we want to separate the realisations that pick any such invalid points.
To that aim, we introduce extra gadgets to the clause curve that act as clauses, but catch these
invalid realisations, so each of them yields the distance of $1$.
Now we have three distinct cases: realisation is satisfying, non-satisfying, or invalid.

We use the same construction as for the indecisive case, but we add a new gadget, which makes the
resulting distance predictable.
For every $j \in [m]$, define
\[\mathrm{FG}_j = (-1, 0) \concat\quad\Concat_{\mathclap{k \in [j - 1]}}
\Bigl((0, 0) \concat (1, 0)\Bigr) \concat (0, 0.5) \concat (0, -0.5) \concat (1, 0)
\concat\quad\Concat_{\mathclap{k \in [m] \setminus [j]}} \Bigl((0, 0) \concat (1, 0)\Bigr)\,.\]
So, we define a clause gadget that ignores all the variables except for $x_j$ and then features both
`true' and `false' for $x_j$.
We then define the clause curve as
\[\mathrm{CC} = \Concat_{i \in [n]} \mathrm{ACG}_i \concat \Concat_{j \in [m]} \mathrm{FG}_j\,.\]
We can now choose to align one of $\mathrm{FG}$ clauses with the variable curve.
As before, due to the synchronisation points we can never get the Fr\'echet distance below $1$.
If one of the realisations $x_j$ of the segments falls into the interval $[(0, -0.5), (0, 0.5)]$,
then it will be not further away than $1$ from both the corresponding points on $\mathrm{FG}_j$; all
the other points, being in the middle at $(0, 0)$, are guaranteed to be at most $0.5 + \varepsilon <
1$ away from their coupled point; so, the one-to-one coupling\footnote{Technically, it is one-to-one
on all points except the realisation corresponding to $x_j$; that one has to be coupled to both
$(0, 0.5)$ and $(0, -0.5)$ in $\mathrm{FG}_j$.} will yield the discrete Fr\'echet distance of~$1$;
thus, the optimal discrete Fr\'echet distance in this case is~$1$.
Therefore, we only need to consider the situations when all the realisations happen to fall in
either the interval $((0, 0.5), (0, 0.5 + \varepsilon)]$ or $[(0, -0.5 - \varepsilon), (0, -0.5))$.
We will treat the first interval as \True\ and the second interval as \False.
Denote the number of satisfying assignments by $N$.
To find the expression for the expected discrete Fr\'echet distance, we need to consider three
cases:
\begin{itemize}
    \item At least one realisation of $m$ variables falls within the $y$-interval $[-0.5, 0.5]$.
    Note that the realisation on each segment is uniform and independent of other segments.
    We get
    \[\Pr[\text{at least one realisation from $[-0.5, 0.5]$}]
    = 1 - \prod_{j \in [m]} \frac{2\varepsilon}{1 + 2\varepsilon}
    = 1 - \left(\frac{2\varepsilon}{1 + 2\varepsilon}\right)^m\,.\]
    Note that in each such case we get the discrete Fr\'echet distance of $1$, as discussed before.
    \item All realisations fall outside the $y$-interval $[-0.5, 0.5]$, and they correspond to a
    non-satisfying assignment.
    Each specific non-satisfying assignment corresponds to picking values on the specific interval,
    either $((0, 0.5), (0, 0.5 + \varepsilon)]$ or $[(0, -0.5 - \varepsilon), (0, -0.5))$, so:
    \[\Pr[\text{specific assignment}] = \prod_{j \in [m]} \frac{\varepsilon}{1 + 2\varepsilon} =
    \left(\frac{\varepsilon}{1 + 2\varepsilon}\right)^m\,.\]
    There are $2^m - N$ such assignments, and each of them contributes the value of $1$.
    \item All realisations fall outside the $y$-interval $[-0.5, 0.5]$, and they correspond to a
    satisfying assignment.
    Again, the probability of getting a particular assignment is
    $\left(\frac{\varepsilon}{1 + 2\varepsilon}\right)^m$, and there are $N$ such assignments.
    Now they contribute values distinct from $1$; still, the optimum is contributed by one of the
    new clauses, and then it will be defined by the realisation closest to $(0, 0)$.
    This is shown in the following lemma.\medskip

    \begin{lemma}
    Consider some realisation $\pi \Subset \mathrm{VC}$ where each value can be interpreted either
    as \True\ or \False\ and the corresponding assignment satisfies the formula.
    Pick $j$ such that the subcurve of $\pi$ realising $\mathrm{VG}_j$ contains the point closest
    to $(0, 0)$, at location $(0, 0.5 + \varepsilon')$ or $(0, -0.5 - \varepsilon')$ for
    some $\varepsilon' > 0$.
    Then the optimal coupling establishes a matching between $\pi$ and $\mathrm{FG}_j$, and the
    discrete Fr\'echet distance is $\dfr(\pi, \mathrm{CC}) = 1 + \varepsilon'$.
    \end{lemma}
    \begin{proof}
    First of all, note that we still have to couple the synchronisation points and we cannot have
    discrete Fr\'echet distance below $1$.
    So, we need to consider only the couplings of $\pi$ with the gadgets of $\mathrm{CC}$.
    Note that if we align $\mathrm{FG}_j$ with $\pi$, we get discrete Fr\'echet distance of $1 + 
    \varepsilon'$.
    Recall that we consider only satisfying assignments, so, if we consider an arbitrary
    subcurve $\mathrm{ACG}_i$, then there is some variable $x_j$ that satisfies the
    corresponding clause, and so the realisation of that variable is $1 + \varepsilon''$ away from
    the corresponding assignment point.
    Therefore, such a coupling will yield the discrete Fr\'echet distance of $1 + \varepsilon''
    \geq 1 + \varepsilon'$.
    Finally, it is easy to see that choosing some $\mathrm{FG}_k$ with $k \neq j$ will also yield
    some distance $1 + \varepsilon'' \geq 1 + \varepsilon'$.
    So, the statement of the lemma holds.
    \end{proof}
    So, here we need to find $\mathbb{E}[\min_{j \in [m]} (1 + \varepsilon'_j)]$ with
    $\varepsilon'_j$ sampled uniformly from $(0, \varepsilon]$; we can rephrase this to
    $1 + \varepsilon \cdot \mathbb{E}[\min_{j \in [m]} u_j]$ with $u_j$ sampled uniformly from
    $(0, 1]$.
    It is a standard result that the minimum now is geometrically distributed, so we get
    $\mathbb{E}[\min_{j \in [m]} u_j] = \frac{1}{1 + m}$, and hence the expected contribution is
    $1 + \frac{\varepsilon}{1 + m}$.
\end{itemize}
We can bring the three cases together to find
\begin{align*}
&\dfrexp(\mathrm{VC}, \mathrm{CC})\\
&= 1 \cdot \left(1 - \left(\frac{2\varepsilon}{1 + 2\varepsilon}\right)^m\right)
+ 1 \cdot (2^m - N) \cdot \left(\frac{\varepsilon}{1 + 2\varepsilon}\right)^m
+ \left(1 + \frac{\varepsilon}{1 + m}\right) \cdot N \cdot \left(\frac{\varepsilon}{1 +
2\varepsilon}\right)^m\\
&= 1 + N \cdot \frac{\varepsilon^{m + 1}}{(1 + m) \cdot (1 + 2\varepsilon)^m}\,.
\end{align*}
So, if we were to compute $\dfrexp(\mathrm{VC}, \mathrm{CC}) = \mu$, then the number of
satisfying assignments is
\[N = (\mu - 1) \cdot \frac{(1 + m) \cdot (1 + 2\varepsilon)^m}{\varepsilon^{m + 1}}\,.\]
This is easy to compute in polynomial time, and our construction can still be done in polynomial
time; hence, the result follows.

\begin{theorem}\label{thm:exp_imp_nump}
The problem \textsc{Expected Discrete Fr\'echet} for imprecise curves modelled as line
segments is \nump-hard.
\end{theorem}

\subsection{Lower Bound Fr\'echet Distance}\seclab{hard}
In this section, we prove that computing the lower bound continuous Fr\'echet distance is \np-hard
for uncertainty modelled with line segments.
This contrasts with the algorithm for indecisive curves, given in \secref{exact_lowerbound}, and
with the algorithm previously suggested by Ahn el al.~\cite{ahn:2012} for the discrete Fr\'echet
distance.
Unlike the upper bound proofs, this reduction uses the \np-hard problem \textsc{Subset-Sum}.
We consider the following problems.

\begin{problem}\problab{decision}
\textsc{Lower Bound Continuous Fr\'echet:} Given a polygonal curve $\curveA$ with $n$ vertices, an
uncertain curve $\U$ with $m$ vertices, and a threshold $\delta > 0$, decide if
$\frmin(\curveA, \U) \leq \delta$.
\end{problem}

\begin{problem}
\textsc{Subset-Sum:} Given a set $S = \{s_1, \dots, s_n\}$ of $n$ positive integers and a target
integer $\tau$, decide if there exists an index set $I$ such that $\sum_{i \in I} s_i = \tau$.
\end{problem}
As a polygonal curve is an uncertain curve, proving \probref{decision} is \np-hard implies the
corresponding problem with two uncertain curves is also \np-hard.

\subsubsection{An Intermediate Problem}
We start by reducing \textsc{Subset-Sum} to a more geometric intermediate curve-based problem.
\begin{definition}\deflab{respect}
Let $\alpha > 0$ be some value, and let $\curveB = \langle\curveB_1, \dots, \curveB_{2n + 1}\rangle$
be a polygonal curve.
Call $\curveB$ an \emph{$\alpha$-regular curve} if for all $1\leq i\leq 2n + 1$, the $x$-coordinate
of $\curveB_i$ is $i \cdot \alpha$.
Let $Y = \{y_1, \dots, y_n\}$ be a set of $n$ positive integers.
Call $\curveB$ a \emph{$Y$-respecting curve} if:
\begin{enumerate}
\item For all $1\leq i\leq n$, $\curveB$ passes through the point $((2i + 1/2)\alpha, 0)$.
\item For all $1\leq i\leq n$, $\curveB$ either passes through the point $((2i - 1/2)\alpha, 0)$ or
$((2i - 1/2)\alpha, -y_i)$.
\end{enumerate}
\end{definition}
Intuitively, the above definition requires $\curveB$ to pass through $((2i + 1/2)\alpha, 0)$ as it
reflects the $y$-coordinate about the line $y = 0$ (see \figref{double}).
Thus, if the curve also passes through $((2i - 1/2)\alpha, 0)$, the two reflections cancel each
other.
If it passes through $((2i - 1/2)\alpha, -y_i)$, the lemma below argues that $y_i$ shows up in the
final vertex height.

\begin{figure}[tpb]
\centering
\includegraphics[width=0.6\textwidth]{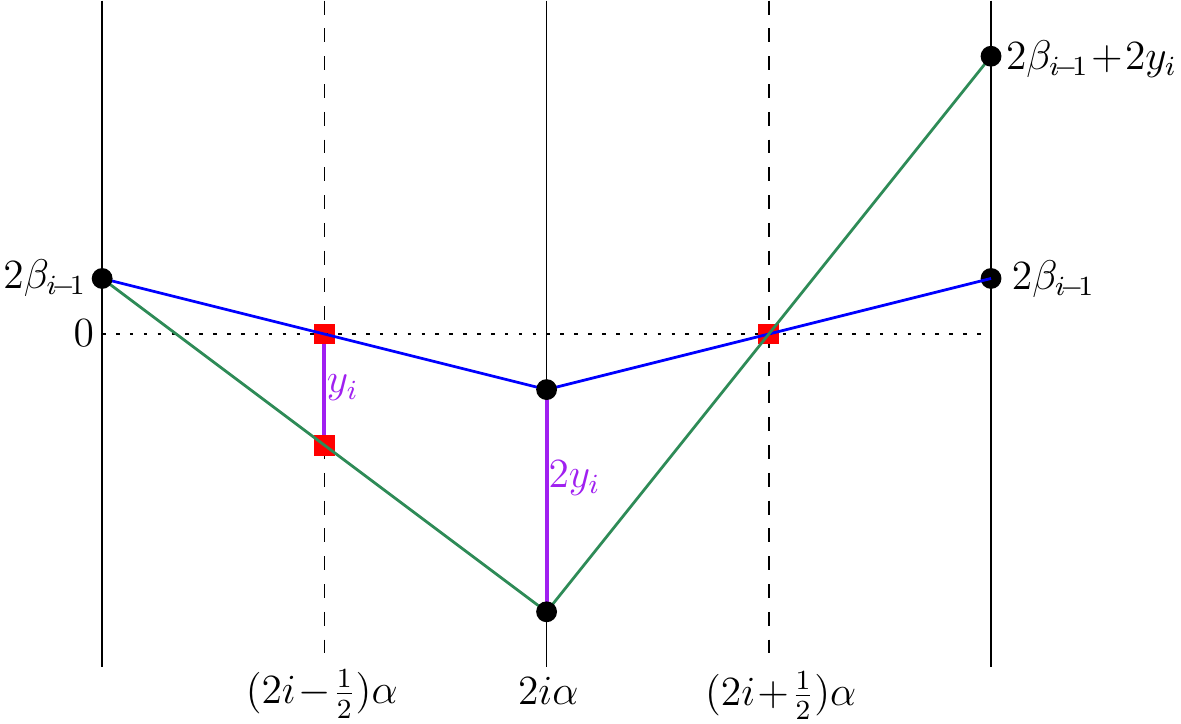}
\caption{Passing through $((2i - 1/2)\alpha, 0)$ does not change the height, and passing through
$((2i - 1/2)\alpha, -y_i)$ adds $2y_i$.}
\figlab{double}
\end{figure}

\begin{lemma}\lemlab{geometry}
Let $\curveB$ be a $Y$-respecting $\alpha$-regular curve, and let $I$ be the subset of indices $i$
such that $\curveB$ passes through $((2i - 1/2)\alpha, -y_i)$.
If $\curveB_1 = (\alpha, 0)$, then $\curveB_{2n + 1} = ((2n + 1)\alpha, 2\sum_{i\in I} y_i)$.
\end{lemma}
\begin{proof}
For $1\leq j \leq n$, let $I_j = \{i\in I \mid i\leq j\}$, and let $\beta_j = \sum_{i\in I_j} y_i$
(where $\beta_0 = 0$).
We argue by induction that $\curveB_{2j + 1} = ((2j + 1)\alpha, 2\beta_j)$, thus yielding the lemma
statement when $j = n$.
For the base case, $j = 0$, the statement becomes $\curveB_1 = (\alpha, 0)$ which is true by
assumption of the lemma statement.

So assume that $\curveB_{2j - 1} = ((2j - 1)\alpha, 2\beta_{j - 1})$.
First suppose that $j\notin I$.
In this case, since $\curveB$ is $Y$-respecting, it passes through points $((2j - 1/2)\alpha, 0)$
and $((2j + 1/2)\alpha,0)$.
This implies $\curveB_{2j} = (2j\alpha, -2\beta_{j - 1})$ and $\curveB_{2j + 1} = ((2j + 1)\alpha,
2\beta_{j - 1}) = ((2j + 1)\alpha, 2\beta_j)$.
Now suppose that $j\in I$.
In this case, it must pass through points $((2j - 1/2)\alpha, -y_j)$ and $((2j + 1/2)\alpha,0)$. 
This implies $\curveB_{2j} = (2j\alpha, 2\beta_{j - 1} - 2(2\beta_{j - 1} + y_j)) = (2j\alpha,
-2(\beta_{j - 1} + y_j))$ and $\curveB_{2j + 1} = ((2j + 1)\alpha, 2(\beta_{i - 1} + y_j)) =
((2j + 1)\alpha, 2\beta_j)$.
See \figref{double}.
\end{proof}
The following is needed in the next section, and follows from the proof of the above.
\begin{corollary}\corlab{bounded}
For a set $Y = \{y_1, \dots, y_n\}$, let $M = \sum_{i = 1}^n y_i$.
For any vertex $\curveB_i$ of a $Y$-respecting $\alpha$-regular curve, its $y$-coordinate is at most
$2M$ and at least $-2M$.
\end{corollary}

\begin{problem}\problab{intermediate}
\textsc{RR-Curve:} Given a set $Y = \{y_1, \dots, y_n\}$ of $n$ positive integers, a value
$\alpha = \alpha(Y) > 0$, and an integer $\tau$, decide if there is a $Y$-respecting
$\alpha$-regular curve $\curveB = \langle\curveB_1, \dots, \curveB_{2n + 1}\rangle$ such that
$\curveB_1 = (\alpha, 0)$ and $\curveB_{2n + 1}=((2n + 1)\alpha, 2\tau)$.
\end{problem}
By \lemref{geometry}, \textsc{Subset-Sum} immediately reduces to the above problem by setting
$Y = S$.
Note that for this reduction it suffices to use any positive constant for $\alpha$; however, we
allow $\alpha$ to depend on $Y$, as this will ultimately be needed in our reduction to
\probref{decision}.
\begin{theorem}\thmlab{curvehard}
For any $\alpha(Y) > 0$, \textsc{RR-Curve} is \np-hard.
\end{theorem}

\subsubsection{Reduction to Lower Bound \Frechet Distance}\seclab{rrtolb}
Let $\alpha$, $\tau$, $Y = \{y_1, \dots, y_n\}$ be an instance of \textsc{RR-Curve}.
In this section, we show how to reduce it to an instance $\delta$, $\curveA$, $\U$ of
\probref{decision}, where the uncertain regions in $\U$ are vertical line segments.
The main idea is to use $\U$ to define an $\alpha$-regular curve, and use $\curveA$ to enforce that
it is $Y$-respecting.
Specifically, let $M = \sum_{i=1}^n y_i$.
Then $\U = \langle v_1, \dots, v_{2n + 1}\rangle$, where $v_i$ is a vertical segment, whose
horizontal coordinate is $i\alpha$ and whose vertical extent is given by the interval $[-2M, 2M]$. 
By \corref{bounded}, we have the following simple observation.
\begin{observation}
The set of all $Y$-respecting $\alpha$-regular curves is a subset of $\Real{U}$.
\end{observation}
Thus, the main challenge is to define $\curveA$ to enforce that the realisation is $Y$-respecting.
To that end, we first describe a gadget forcing the realisation to pass through a specified point.

\begin{definition}
For any point $p = (x,y)\in \R^2$ and value $\delta > 0$, let the \emph{$\delta$ gadget} at $p$,
denoted by $\g_{\delta}(p)$, be the curve:
$(x, y) \concat (x, y + \delta) \concat (x, y - \delta) \concat (x, y + \delta) \concat (x, y)$.
See \figref{g}.
\end{definition}

\begin{figure}[tpb]
\centering
\hspace{.1in}
\begin{subfigure}{0.1035\linewidth}
\centering
\includegraphics[width=\linewidth]{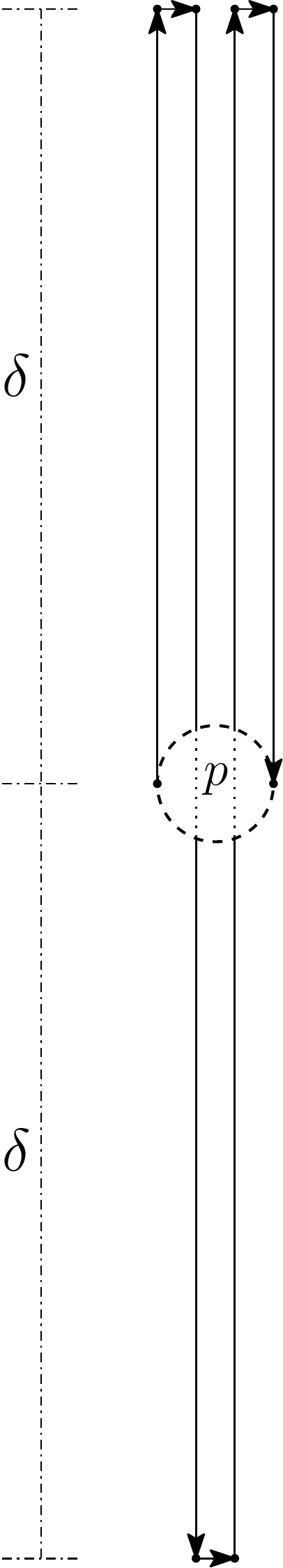}
\caption{$\g_\delta(p)$}
\figlab{g}
\end{subfigure}
\hspace{.5in}
\begin{subfigure}{0.4\linewidth}
\centering
\includegraphics[width=.88\linewidth]{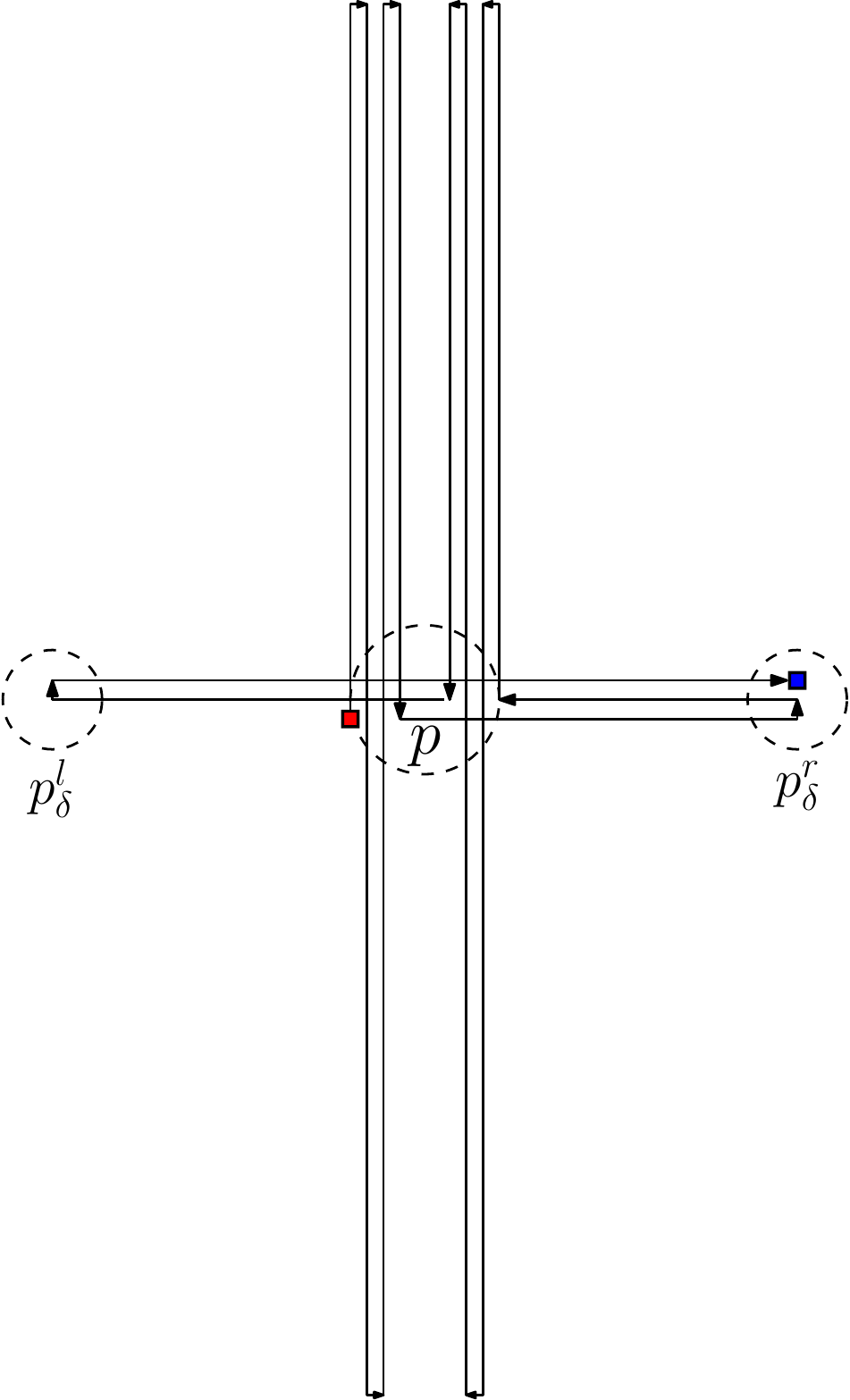}
\caption{$\lcg_\delta(p)$}
\figlab{lcg}
\end{subfigure}
\hspace{.21in}
\begin{subfigure}{0.3\linewidth}
\centering
\includegraphics[width=.83\linewidth]{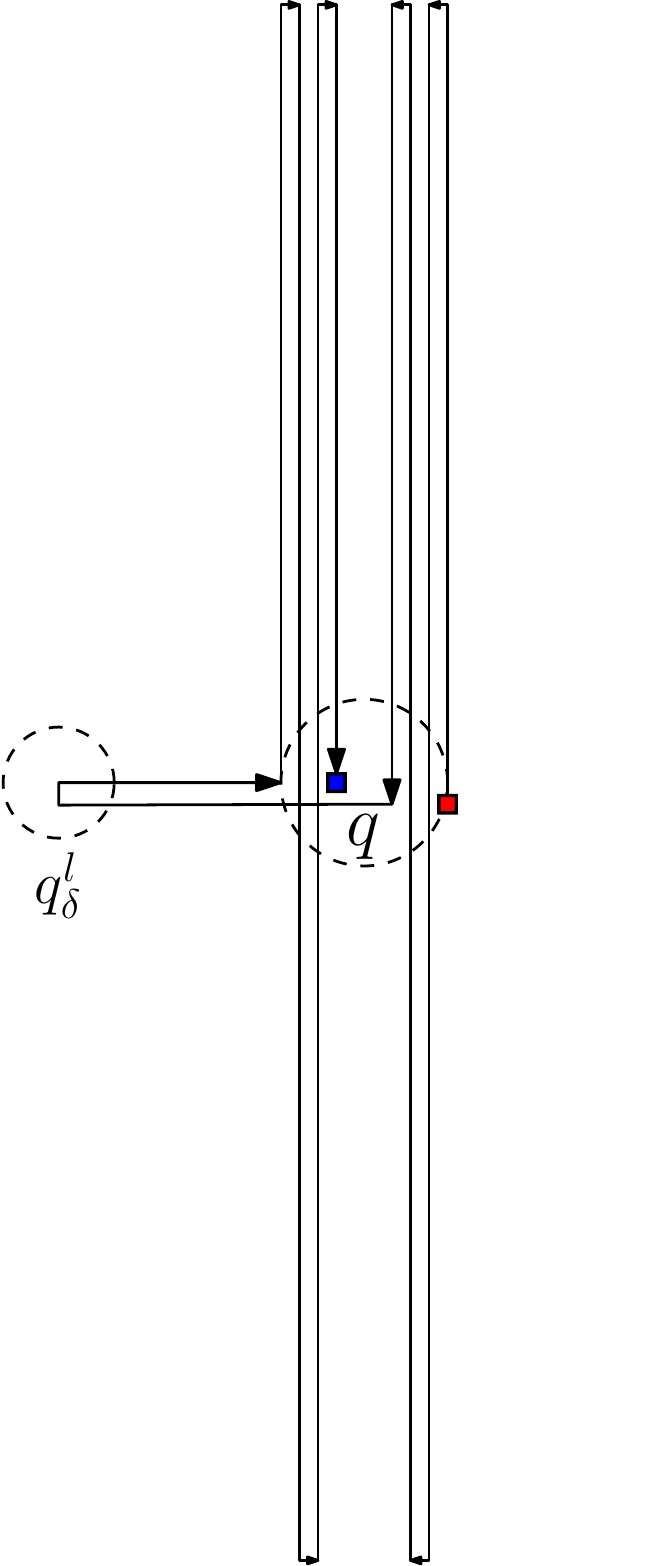}
\caption{$\ucg_\delta(q)$}
\figlab{ucg}
\end{subfigure}
\caption{Depiction of gadgets $\g_\delta(p)$, $\lcg_\delta(p)$, and $\ucg_\delta(p)$.
Circles represent zero-area points.
For the right two figures, the red / blue square represents the starting / ending point.}
\end{figure}

\begin{lemma}\lemlab{passpoint}
Let $p = (x,y)\in \R^2$ be a point, and let $\ell$ be any line segment.
Then if $\fr(\ell, \g_{\delta}(p)) \leq \delta$, then $\ell$ must pass through $p$.
\end{lemma}
\begin{proof}
In order, $\g_\delta(p)$ visits the points $(x, y + \delta)$, $(x,y - \delta)$, and $(x,y + \delta)$.
Let $a$, $b$, $c$ be the points from $\ell$ which get mapped to these respective points under an
optimal \Frechet mapping.
If the \Frechet distance is at most $\delta$, then the $y$-coordinate of $a$ and $c$ must be at
least $y$ and the $y$-coordinate of $b$ must be at most $y$.
This implies that if $\ell$ is non-horizontal then $a = b = c$.
However, if $a = b = c$, then this point must be $p$ itself, as $p$ is the only point with distance
at most $\delta$ from both $(x, y + \delta)$ and $(x, y - \delta)$.
If $\ell$ is horizontal, then one again concludes $a = b = c = p$, as this is the only point on a
horizontal segment matching $(x, y + \delta)$ and $(x, y - \delta)$.  
\end{proof}

For our uncertain curve to be $Y$-respecting, it must pass through all points of the form
$((2i + 1/2)\alpha, 0)$.
This condition is satisfied by the lemma above by placing a $\delta$ gadget at each such point.
The second condition of a $Y$-respecting curve is that it passes through $((2i - 1/2)\alpha, 0)$ or
$((2i - 1/2)\alpha, -y_i)$.
This condition is much harder to encode, and requires putting several $\delta$ gadgets together to
create a composite gadget, which we now describe. 

\begin{definition}\deflab{ulgadgets}
For any point $p = (x,y)\in \R^2$ and value $\delta > 0$, let $p^l_\delta = (x - \delta/2, y)$ and
$p^r_\delta = (x + \delta/2, y)$.
Define the \emph{$\delta$ lower composite gadget} at $p$, denoted $\lcg_\delta(p)$, to be the curve
$\g_\delta(p) \concat p^r_\delta \concat \g_\delta(p) \concat p^l_\delta \concat p^r_\delta$.
See \figref{lcg}.
Define the \emph{$\delta$ upper composite gadget} at $q$, denoted $\ucg_\delta(q)$, to be the curve
$\g_\delta(q) \concat q^l_\delta \concat \g_\delta(q)$.
See \figref{ucg}.
Define the \emph{$\delta$ composite gadget} of $p$ and $q$, denoted $\cg_\delta(p,q)$, to be the
curve $\lcg_\delta(p)$ followed by $\ucg_\delta(q)$:
$\lcg_{\delta}(p) \concat \ucg_{\delta}(q)$. 
\end{definition}
To use this composite gadget we centre the lower gadget at height $-y_i$, and the upper gadget
directly above it at height zero.
As the two gadgets are on top of each other, ultimately we require our uncertain curve to go back
and forth once between consecutive vertical line segments, for which we have the following key
property.

\begin{lemma}\lemlab{three}
Let $p = (x_p, -y_p)$ and $q = (x_p, 0)$ be points in $\R^2$.
Let $\curveB = \langle a, b, c, d\rangle$ be a three-segment curve such that $b_x > x_p + \delta$
and $c_x < x_p-\delta$.
If $\fr(\curveB, \cg_\delta(p, q)) \leq \delta$, then:
\begin{enumerate}[(i)]
    \item the segment $ab$ must pass through $p$,
    \item the segment $cd$ must pass through $q$, and 
    \item the segment $bc$ must either pass through $p$ or through $q$.
\end{enumerate}
\end{lemma}
\begin{proof}
Recall from \defref{ulgadgets} that $\cg_\delta(p, q) = \g_\delta(p) \concat p^r_\delta \concat
\g_\delta(p) \concat p^l_\delta \concat p^r_\delta \concat \g_\delta(q) \concat q^l_\delta \concat
\g_\delta(q)$, and that the gadgets $\g_\delta(p)$ and $\g_\delta(q)$ lie entirely on the vertical
line at $x_p = x_q$.
Thus, as $b_x > x_p + \delta$ and $c_x < x_p - \delta$, each occurrence of $\g_\delta(p)$ or
$\g_\delta(q)$ in $\cg_\delta(p, q)$ must map either entirely before or after $b$, and similarly
entirely before or after $c$.

Moreover, as $\cg_\delta(p, q)$ starts with $\g_{\delta}(p)$ and $b_x > x_p + \delta$, this implies
that $\g_\delta(p)$ maps to the segment $ab$, which by \lemref{passpoint} implies that $ab$ passes
through $p$.
Similarly, as $\cg_\delta(p, q)$ ends with $\g_\delta(q)$ and $c_x < x_q - \delta$, $cd$ passes
through $q$.

Finally, the portion of $\cg_{\delta}(p, q)$ that maps to the segment $bc$ must contain a point on
the vertical line at $x_p = x_q$ (since $b_x > x_p + \delta$ and $c_x < x_p - \delta$).
By the construction of $\cg_{\delta}(p, q)$, this point must lie on one of the (middle)
$\g_\delta(p)$ or $\g_\delta(q)$ gadgets.
As we already argued, such gadgets must map entirely to one side of $b$ or $c$, so
\lemref{passpoint} implies that $bc$ must pass through $p$ or $q$.
\end{proof}
As $bc$ shares an endpoint with $ab$ and $cd$, the following corollary is immediate.
It will be used to argue that while our uncertain curve goes back and forth between consecutive
vertical lines, it defines an $\alpha$-regular curve.
(See \figref{solutions} used for \thmref{mainhard}.)
\begin{corollary}\corlab{same}
If $\fr(\curveB, \cg_\delta(p,q)) \leq \delta$, then either $ab$ and $bc$ are on the same line,
or $cd$ and $bc$ are on the same line.
\end{corollary}

The following lemma acts as a rough converse of \lemref{three}.
\begin{lemma}\lemlab{otherdirection}
Let $p = (x_p, -y_p)$ and $q=(x_p, 0)$ be points in $\R^2$, with $y_p \leq \delta/4$.
Let $\curveB = \langle p, b, c, q\rangle$ be a curve such that
$x_p + \delta < b_x \leq x_p + 1.1\delta$, $x_p - 1.1\delta \leq c_x < x_p - \delta$, and
$-\delta/2 \leq b_y, c_y \leq \delta/2$.
If $bc$ passes through either $p$ or $q$, then $\fr(\curveB, \cg_\delta(p,q)) \leq \delta$.
\end{lemma}
\begin{proof}
Recall that $\cg_\delta(p, q) = \g_\delta(p) \concat p^r_\delta \concat \g_\delta(p) \concat
p^l_\delta \concat p^r_\delta \concat \g_\delta(q) \concat q^l_\delta \concat \g_\delta(q)$.
First, observe that all points on the prefix $\g_\delta(p) \concat p_\delta^r$ of $\cg_\delta(p, q)$
are at most $\delta$ away from $p$, and thus can all be mapped to the starting point of $\curveB$.
Similarly, all points on the suffix $q_\delta^l \concat \g_\delta(q)$ of $\cg_\delta(p, q)$ are at
most $\delta$ away from $q$, and thus can all be mapped to the ending point of $\curveB$.
Thus, it suffices to argue that $\fr(\curveB, \curveA) \leq \delta$, where
$\curveA = p^r_\delta \concat \g_\delta(p) \concat p^l_\delta \concat p^r_\delta \concat
\g_\delta(q) \concat q^l_\delta$.

It is easiest to describe the rest of the mapping in a similar manner, that is, as an alternating
sequence of moves, where we stand still at a single point on one curve while moving along a
contiguous subcurve from the other curve, and then switching curves.
We now describe this sequence, which differs based on whether $bc$ passes through $p$ or $q$.
Ultimately, the mappings will be valid, since for each move, all points on the subcurve will have
distance at most $\delta$ to the fixed point on the other curve.
Thus, we now simply describe the moves without reiterating this property (distance at most $\delta$)
which is validating each each move.

First suppose that $bc$ passes through $p$, in which case $\curveB = \langle p, b, p, c, q\rangle$.
In this case, we first map the prefix $\langle p, b, p\rangle$ of $\curveB$ to $p^r_\delta$.
Next, we map the prefix $p_\delta^r \concat \g_\delta(p) \concat p_\delta^l$ of $\curveA$ to $p$.
Then we map the suffix $\langle p, c, q\rangle$ of $\curveB$ to $p_\delta^l$.
Finally, we map the suffix $p^l_\delta \concat p^r_\delta \concat \g_\delta(q) \concat q^l_\delta$
of $\curveA$ to $q$.

Now suppose that $bc$ passes through $q$, in which case $\curveB = \langle p, b, q, c, q\rangle$.
In this case, we first map the prefix $p^r_\delta \concat \g_\delta(p) \concat p^l_\delta \concat
p^r_\delta$ of $\curveA$ to $p$.
Next, we map the prefix $\langle p, b,  q\rangle$ of $\curveB$ to $p^r_\delta$.
Then we map the suffix $p^r_\delta \concat \g_\delta(q) \concat q_{\delta}^l$ of $\curveA$ to $q$.
Finally, map the suffix $\langle q, c, q\rangle$ of $\curveB$ to $q_\delta^l$.
\end{proof}

\begin{figure}[tpb]
\centering
\begin{subfigure}[b]{0.52\linewidth}
    \centering
    \vspace{.21in}
    \includegraphics[width=\linewidth]{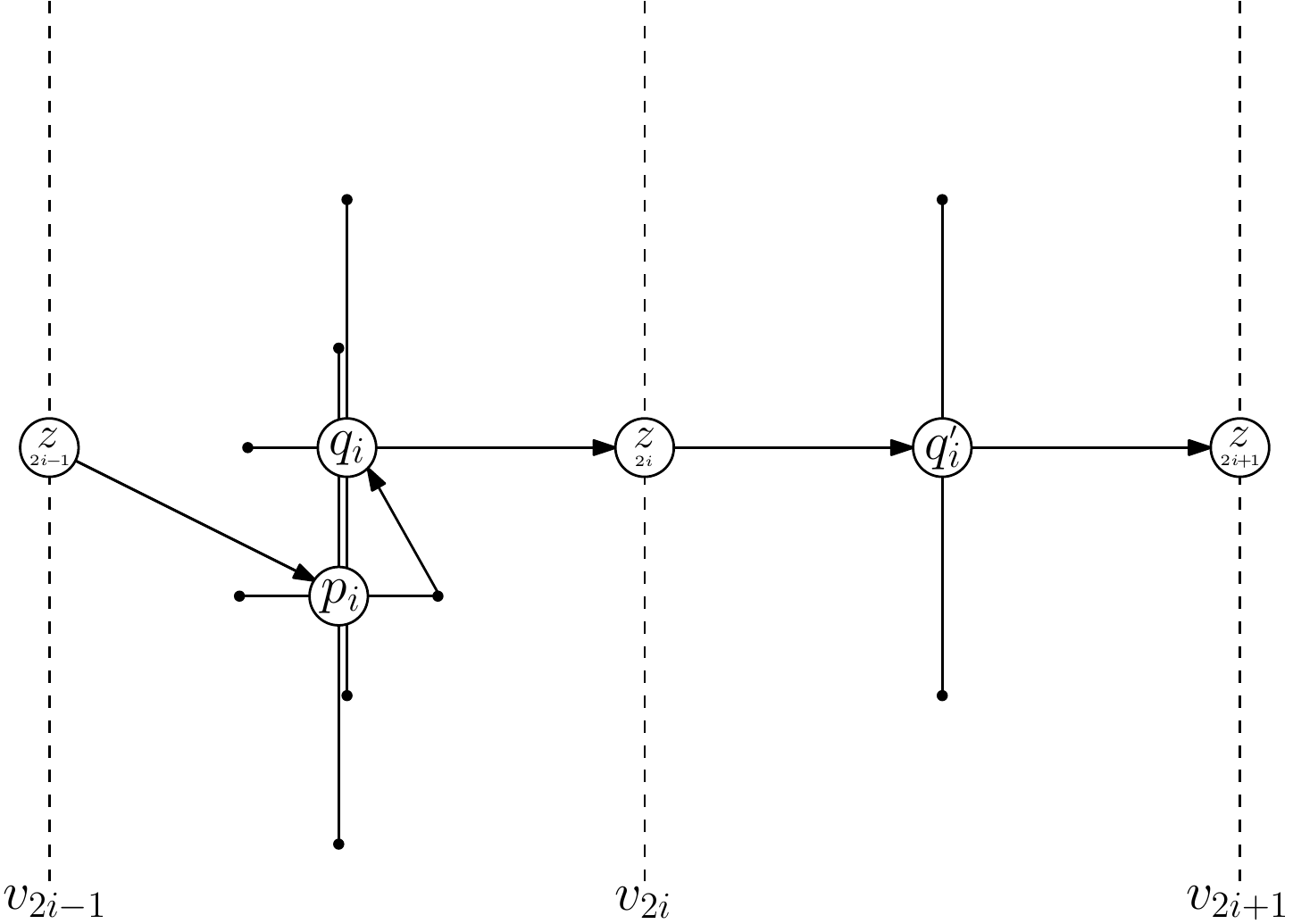}
    \vspace{.12in}
    \caption{Pictorial representation of $\lambda_i$.}
    \figlab{pair}
\end{subfigure}
\hspace{.07in}
\begin{subfigure}[b]{0.43\linewidth}
    \centering
    \includegraphics[width=.932\linewidth]{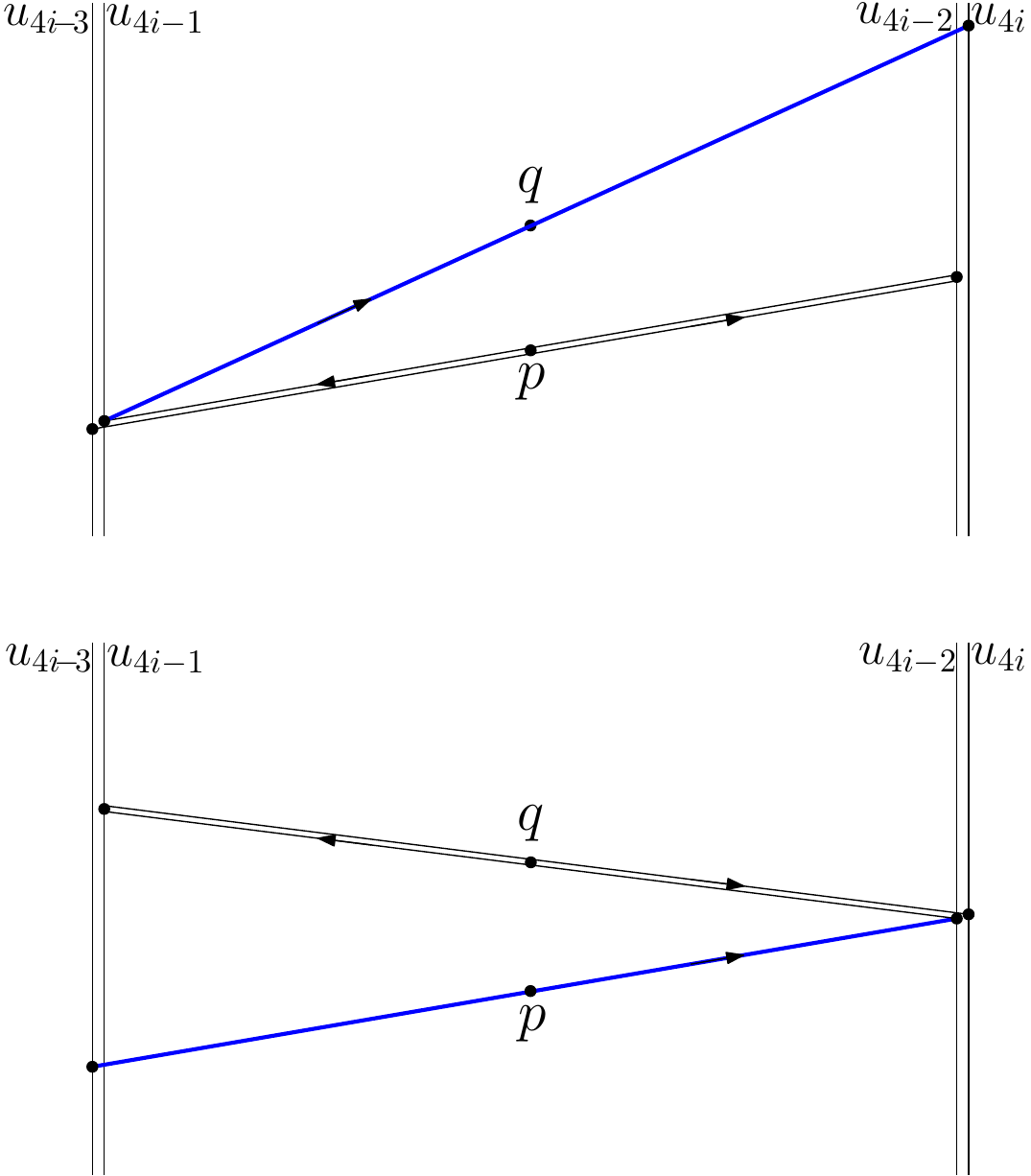}
    \caption{The two solutions.}
    \figlab{twosol}
\end{subfigure}
\caption{On the left, $\lambda_i$.
On the right, the two possible solutions with \Frechet distance at most $\delta$.
The top (resp.\@ bottom) corresponds to an $\alpha$-regular curve passing through $q$ (resp.\@ $p$).}
\figlab{solutions}
\end{figure}

\begin{theorem}\thmlab{mainhard}
\textsc{Lower Bound Continuous \Frechet} (\probref{decision}) is \np-hard, even when the uncertain
regions are all equal-length vertical segments with the same height and the same horizontal distance
(to the left or right) between adjacent uncertain regions.
\end{theorem}
\begin{proof}
To prove \np-hardness, we give a reduction from \textsc{RR-Curve}, which is \np-hard by
\thmref{curvehard}.
Let $\alpha(Y)$, $\tau$, $Y = \{y_1, \dots, y_n\}$ be an instance of \textsc{RR-Curve}.
For the reduction we set $\delta = 4M$, where $M = \sum_{i = 1}^n y_i$.
Note that \thmref{curvehard} allows us to choose how to set $\alpha(Y)$, and in particular we set
$\alpha = 2.1 \delta = 8.4 M$.
(More specifically, the properties we need are that $\alpha > 2\delta$ and $\delta \geq 4M$.)
We now describe how to construct $\U$ and $\curveA$.

Let $V = \{v_1, \dots, v_{2n + 1}\}$ be a set of vertical line segments where all upper
(resp.\@ lower) endpoints of the segments have height $2M$ (resp.\@ $-2M$), and for all $i$, the
$x$-coordinate of $v_i$ is $i \alpha$.
Let $\U = \langle \ur_1, \dots, \ur_{4n + 1}\rangle$ be the uncertain curve such that
$\ur_{4n + 1} = v_{2n+1}$, and for all $1 \leq i \leq n$, $\ur_{4i - 3} = v_{2i - 1}$,
$\ur_{4i - 2} = v_{2i}$, $\ur_{4i - 1} = v_{2i - 1}$, and $\ur_{4i} = v_{2i}$.

For $1 \leq i \leq 2n + 1$, define the points $z_i = (i\alpha, 0)$, and for $1 \leq i \leq n$,
define $q_i = ((2i - 1/2)\alpha, 0)$, $q_i'=((2i + 1/2)\alpha, 0)$, and
$p_i=((2i - 1/2)\alpha, -y_i)$.
For a given value $1 \leq i \leq n$, consider the curve $\lambda_i = z_{2i - 1} \concat
\cg_\delta(p_i, q_i) \concat z_{2i} \concat \g_\delta(q_i')$ (see \figref{pair}).
Let $s = (\alpha, 0)$ and $t = ((2n + 1)\alpha, 2\tau)$.
Then the curve $\curveA$ is defined as
\[\curveA = \g_\delta(s) \concat \lambda_1 \concat \lambda_2 \concat \dots \concat \lambda_{n-1}
\concat \lambda_n \concat \g_{\delta}(t)\,.\]

First, suppose there is a curve $\curveB' = \langle \curveB_1', \dots, \curveB_{4n+1}'\rangle
\Subset \U$ such that $\fr(\curveA, \curveB') \leq \delta$.
Let $\curveB = \langle \curveB_1, \dots, \curveB_{2n + 1}\rangle$ be the curve such that
$\curveB_{2n + 1} = \curveB_{4n + 1}'$, and for all $1 \leq i \leq n$,
$\curveB_{2i - 1} = \curveB_{4i - 3}$ and $\curveB_{2i} = \curveB_{4i}$.
We argue that $\curveB$ is an $\alpha$-regular $Y$-respecting curve with $\curveB_1 = s$ and
$\curveB_{2n + 1} = t$.

Observe that $\curveB$ is $\alpha$-regular, as by the definition of $\U$, $\curveB_i$ is a point on
the vertical segment $v_i$.
Also, as $\curveA$ begins (resp.\@ ends) with $\g_\delta(s)$ (resp.\@ $\g_\delta(t)$), by
\lemref{passpoint}, $\curveB_1 = \curveB_1' = s$ (resp.\@ $\curveB_{2n + 1} = \curveB_{4n + 1}' = t$).
Thus, it remains to argue that $\curveB$ is $Y$-respecting.
To that end, consider the portion $\lambda_i$ of $\curveA$ for some value $i$.

First consider the gadget $\g_\delta(q_i')$ from $\lambda_i$ lying between $z_{2i}$ and $z_{2i+1}$.
By our choice of $\alpha$, this gadget is strictly more than $\delta$ away from both $v_{2i}$ and
$v_{2i + 1}$, and so the portion of $\curveB'$ matched to $\g_\delta(q_i')$ must lie between
$\curveB'_{4i} = \curveB_{2i}$ and $\curveB'_{4i + 1} = \curveB_{2i + 1}$.
Thus, by \lemref{passpoint}, $\curveB$ must pass through $q_i'$.

Now consider the gadget $\cg_\delta(p_i, q_i) = \lcg(p_i) \concat \ucg(q_i)$ from $\lambda_i$ lying
between $z_{2i - 1}$ and $z_{2i}$.
This gadget is strictly more than $\delta$ away from both $v_{2i - 1}$ and $v_{2i}$, implying both
that the portion of $\curveB'$ matched to $\cg_\delta(p_i, q_i)$ lies between $\curveB'_{4i - 3}$
and $\curveB'_{4i}$, and that all three segments in the subcurve from $\curveB'_{4i - 3}$ to
$\curveB'_{4i}$ must in part map to $\cg_\delta(p_i, q_i)$.
Thus, by \lemref{three}, $\curveB'_{4i - 3} \curveB'_{4i - 2}$ passes through $p_i$, and
$\curveB'_{4i - 1}\curveB'_{4i}$ passes through $q_i$.
By \corref{same}, either $\curveB'_{4i - 2} = \curveB'_{4i}$ or $\curveB'_{4i - 3} =
\curveB'_{4i - 1}$, and thus $\curveB'_{4i - 3} \curveB'_{4i} = \curveB_{2i - 1} \curveB_{2i}$
passes through either $p_i$ or $q_i$ (see \figref{twosol}).
Thus, $\curveB$ is $Y$-respecting.

Now suppose that there is an $\alpha$-regular $Y$-respecting curve $\curveB = \langle \curveB_1,
\dots, \curveB_{2n + 1}\rangle$ such that $\curveB_1 = s$ and $\curveB_{2n + 1} = t$.
Let $\intrv(p_i)$ be the intersection with $v_{2i}$ of the line passing through $\curveB_{2i - 1}$
and $p_i$, and let $\intrv(q_i)$ be the intersection with $v_{2i - 1}$ of the line passing through
$\curveB_{2i}$ and $q_i$.
Let $\curveB' = \langle \curveB_1', \dots, \curveB_{4n + 1}'\rangle$ be the curve such that
$\curveB_{4n + 1}' = \curveB_{2n + 1}$, and for all $1 \leq i \leq n$, $\curveB_{4i - 3}' =
\curveB_{2i - 1}$, $\curveB_{4i - 2}' = \intrv(p_i)$, $\curveB_{4i - 1}' = \rho$, and
$\curveB_{4i}' = \curveB_{2i}$, where $\rho = \curveB_{2i - 1}$ if $\curveB$ passes through $q_i$
and $\rho = \intrv(q_i)$ if $\curveB$ passes through $p_i$.
(See \figref{twosol}.)

Let $\midpt(\ell)$ denote the midpoint of a line segment $\ell$.
Observe that by construction $\midpt(\curveB'_{4i - 3} \curveB'_{4i - 2}) = p_i$,
$\midpt(\curveB'_{4i - 1} \curveB'_{4i}) = q_i$, and $\midpt(\curveB'_{4i - 2} \curveB'_{4i - 1}) =
p_i$ (resp.\@ $q_i$) if $\curveB$ passed through $q_i$ (resp.\@ $p_i$).
Let $\gamma_i= \langle p_i, \curveB'_{4i - 2}, \curveB'_{4i - 1}, q_i\rangle$, which by the previous
argument is a subcurve of $\curveB'$.

To argue that $\fr(\curveB', \curveA) \leq \delta$, we now describe how to walk along the curves
$\curveB'$ and $\curveA$ such that at all times the distance between the positions on the respective
curves is at most $\delta$.
Note that $\gamma_i$ satisfies the conditions of \lemref{otherdirection}, implying that
$\fr(\cg_\delta(p_i, q_i), \gamma_i) \leq \delta$, and thus for all $i$, we can map
$\cg_\delta(p_i, q_i)$ to $\gamma_i$.
For the other parts of the curves, first observe that with the exception of the
$\cg_\delta(p_i, q_i)$ gadgets, $\curveA$ is $x$-monotone, i.e.\@ as we walk along it, the
$x$-coordinate never decreases.
Moreover, with the exception of the $\gamma_i$ portions, $\curveB'$ is $x$-monotone.
Finally, observe that $\cg_{\delta}(p_i, q_i)$ and $\gamma_i$ have the same starting and ending
points, and $\curveB'$ and $\curveA$ both start at $s$ and end at $t$.
Thus, with the exception of the $\cg_{\delta}(p_i, q_i)$ and $\gamma_i$ portions, we can map all
points from $\curveA$ with a given $x$-coordinate to the point on $\curveB'$ with the same
$x$-coordinate.
It is easy to verify that this maps points between the curves that are at most $\delta$ apart.
First, as $\curveB'$ is identical to $\curveB$ outside of the $\gamma_i$, and since $\curveB$ is
$Y$-respecting, $\curveB'$ passes through $s$, $t$, and $q_i'$ for all $i$.
Thus, the mapping stands still on $\curveB'$ at these respective points as $\curveA$ executes the
$\g_\delta(s)$, $\g_\delta(t)$, and $\g_\delta(q_i')$ gadgets.
Outside of these points, it is easy to verify that the vertical distance between the curves is at
most $4M$ by \corref{bounded}, and by construction $4M \leq \delta$.
\end{proof}

\section{Algorithms for Lower Bound Fr\'echet Distance}\label{sec:lowerbound}
In the previous section, we have shown that the decision problem for $\frmin$ is hard, given a
polygonal curve and an uncertain curve with line-segment-based imprecision model.
Interestingly, the same problem is solvable in polynomial time for indecisive curves.
This result highlights a distinction between $\frmin$ and $\frmax$ and between different
uncertainty models.
To tackle $\frmin$ with general uncertain curves, we develop approximation
algorithms.

\subsection{Exact Solution for Indecisive Curves}\label{sec:exact_lowerbound}
The key idea is that we can use a dynamic programming approach similar to that for computing 
Fr\'echet distance~\cite{alt:1995} and only keep track of realisations of the last indecisive
point considered so far.
(Note that one can also reduce the problem to Fr\'echet distance between paths in
\emph{DAG complexes,} studied by Har-Peled and Raichel~\cite{harpeled:2014}, but this yields a
slower running time.)
We first present the approach for an indecisive and a precise curve, and then generalise it to two
indecisive curves.

\subsubsection{Indecisive and Precise}
Consider the setting with an indecisive curve $\mathcal{V} = \langle V_1, \dots, V_n\rangle$ of $n$
points and a precise curve $\pi = \langle p_1, \dots, p_m\rangle$ with $m$ points; each
indecisive point has $k$ possible realisations, $V_i = \{q_i^1, \dots, q_i^k\}$.
We want to solve the decision problem \emph{`Is the lower bound Fr\'echet distance between the
curves below some threshold $\delta$?',} so $\frmin(\pi, \mathcal{V}) \leq \delta?$

Consider the free-space diagram for this problem; suppose $\mathcal{V}$ is positioned along the
horizontal axis, and $\pi$ along the vertical axis.
Just as for precise curve Fr\'echet distance, we are interested in the reachable intervals on the
cell boundary, since the free space in the cell interior is convex; however, now we care about the
different realisations of the points, so we get a set of reachable boundaries instead of a single
cell boundary.
We can adapt the standard dynamic program to deal with this problem.
We propagate reachability column by column.
An important aspect is that we only need to make sure that a reachable point is reachable by a
monotone path in the free-space diagram induced by \emph{some} valid realisation; we do not need to
remember which one, since we never return to the previous points on the indecisive curve, and we
also do not care about the realisations that yield a distance that is higher than $\delta$\dsh a
significant distinction from the upper bound Fr\'echet distance.

First of all, define $\Feas(i, \ell)$ to be the \emph{feasibility column} for realisation
$q_i^\ell$ of $U_i$.
This is a set of intervals on the vertical cell boundary line in the free-space diagram,
corresponding to the subintervals of one curve within distance $\delta$ from a point on the other
curve.
It is computed exactly the same way as for the precise Fr\'echet distance\dsh it depends on the
distance between a point and a line segment and gives a single interval on each vertical cell
boundary.
We can compute feasibility for the right boundary of all cells in a column for a given realisation,
thus obtaining $\Feas(i, \ell)$.

Consider the standard dynamic program for computing Fr\'echet distance on precise curves.
Represent it so that it operates column by column, grouping propagation of reachable intervals
between vertically aligned cells.
Call that procedure $\Prop(R)$, where $R$ is the \emph{reachability column} for point $i$ and the
result is the reachability column for point $i + 1$ on one of the curves.
Again, the reachability column is a set of intervals on a vertical line, indicating the points in
the free-space diagram that are reachable from the lower left corner with a monotone path.

Define $\Reach(i, s)$ to be the reachability column induced by $q_i^s$, where a point is in
a reachability interval if it can be reached by a monotone path for some realisation of the previous
points.
Then we compute
\[\Reach(i + 1, \ell) = \Feas(i + 1, \ell) \cap \bigcup_{\ell' \in [k]} \Prop(\Reach(i, \ell'))\,.\]
So, we iterate over all the realisations of the  previous column, thus getting precise cells, and
simply propagate the reachable intervals as in the precise Fr\'echet distance algorithm.
For the column corresponding to $U_1$, we set one reachable interval of a single point at the bottom
for all realisations $p_1^s$ for which $\lVert q_1^s - p_1\rVert \leq \delta$.

We now show correctness of this approach.
\begin{lemma}
For all $i > 1$,
\[\Reach(i, \ell) = \Bigl\{y \Bigm| \exists_{q_1^{\ell_1}, \dots, q_{i - 1}^{\ell_{i -1}}}
\Bigl[\fr\Bigl(\bigl(\pi[1:\lfloor y\rfloor] \concat\pi(y),
\Concat_{j \in [i - 1]} q_j^{\ell_j}\bigr) \concat q_i^\ell\Bigr) \leq\delta\Bigr]\Bigr\}\,.\]
So, for any point inside a reachability interval there is a realisation that defines a free-space
diagram and a monotone path through that diagram to this point.
\end{lemma}
\begin{proof}
We show this by induction on $i$.
To compute $\Reach(2, \ell)$ for any fixed $\ell \in [k]$, we start from a single point in the
bottom left corner of the free space for the realisations of $U_1$ that are close enough to $q_1$
and we propagate the reachability through the resulting precise free-space column.
Clearly, the statement holds in this case; if some realisation of $U_1$ is too far from $q_1$, then
the reachability column is correctly empty.

Now assume the statement holds for $\Reach(i, \ell')$ for all $\ell' \in [k]$.
Note that all the values that we add to $\Reach(i + 1, \ell)$ for some fixed $\ell$ are feasible,
since we explicitly take the feasibility column and intersect it with the propagated reachability.
Furthermore, any point $y$ in $\Reach(i + 1, \ell)$ comes as a result of propagation from some
$\Reach(i, \ell')$ for some $\ell'$.
So, there is at least one point $y'$ in the reachability column $i$ for realisation $q_i^{\ell'}$
from which there is a monotone path to $y$.
Since we know there was a realisation up to that point of the two curves that enables a monotone
path from the start of the free space diagram to $y'$; and since point $V_{i + 1}$ is independent of
the previous points; and since we have a fixed valid realisation for points $V_i$ and $V_{i + 1}$
that enables the continuation of the monotone path from $y'$ to $y$, we conclude that the statement
of the lemma holds for the column $i + 1$.
\end{proof}
Therefore, querying the upper right corner of all reachability intervals for $V_n$ will correctly
give us the answer to the decision problem.

Now we need to analyse the complexity of the reachability column.
Note that a particular right cell boundary is entirely reachable if the bottom of the cell is
reachable; combined with the feasibility interval, we get one reachability interval per cell.
Furthermore, if a cell is only reachable from the left, since we consider monotone paths, each
realisation of the previous points induces a reachable interval of $[y', 1]$ for some $0 \leq y'
\leq 1$ if you assume the boundary coordinate range to be $[0, 1]$; therefore, taking a union of
such intervals still gives us at most one reachability interval per cell.
So, in the worst case we have $\Theta(mk)$ intervals that we need to store.
To propagate, we consider all combinations of the two successive indecisive points for all cells,
yielding the total running time of $\Theta(mnk^2)$.

Furthermore, observe that we can also store a realisation of the previous point on the indecisive
curve with the interval that corresponds to the lowest reachable point on the current interval.
If we then store all the reachability columns, we can later backtrack and find a specific curve
that realises Fr\'echet distance below the threshold $\delta$.
This increases the storage requirements to $\Theta(mnk)$; the running time stays the same.

We summarise the results:
\begin{theorem}\label{thm:indecisivedecider}
Given an indecisive curve $\mathcal{V} = \langle V_1, \dots, V_n\rangle$, where each indecisive
point has $k$ options, $V_i = \{q_i^1, \dots, q_i^k\}$, a precise curve $\pi = \langle p_1,
\dots, p_m\rangle$, and a threshold $\delta > 0$, we can decide if $\frmin(\pi, \mathcal{V}) \leq
\delta$ in time $\Theta(mnk^2)$ in the worst case, using $\Theta(mk)$ space.
We can also report the realisation of $\mathcal{V}$ realising Fr\'echet distance at most $\delta$,
using $\Theta(mnk)$ space instead.
Call the algorithm that solves the problem and reports a fitting realisation
$\decider(\delta, \pi, \mathcal{V})$.
\end{theorem}

\subsubsection{Indecisive and Indecisive}
Now consider the setting where instead of $\pi$ we are given curve $\mathcal{U} = \langle U_1,
\dots, U_m\rangle$ with $k$ options per indecisive point, $U_i = \{p_i^1, \dots, p_i^k\}$.
We can adapt the algorithm of the previous section in a straightforward way by propagating in
column-major order, but cell by cell.

A cell boundary now depends on three indecisive points, so there are $k^3$ options per boundary to
consider.
We now store the possibilities for $m - 1$ right cell boundaries, $k^3$ realisations per boundary,
and a single horizontal boundary, with also $k^3$ options.
So, we use $\Theta(mk^3)$ storage.

Whenever we propagate to one further cell, we need to find the reachability for the top and the
right boundary of the cell based on the left and the lower boundary of the cell.
We again go over all the combinations of the realisations of the points that define the cell,
yielding $k^4$ possible precise cells to consider.
We aggregate the values as before, as for both the top and the right boundary only three points
matter.

Since we solve the same problem as in the previous section and never have to revisit a previously
considered point, it should be clear that this approach is correct.
However, now we take $\Theta(k^4)$ time per cell, so in the worst case we need $\Theta(mnk^4)$ time
to complete the propagation.
\begin{theorem}\label{thm:indind}
Given two indecisive curves $\mathcal{U} = \langle U_1, \dots, U_n\rangle$ and $\mathcal{V} =
\langle V_1, \dots, V_m\rangle$, where each indecisive point has $k$ options, $U_i = \{p_i^1, \dots,
p_i^k\}$ and $V_i = \{q_i^1, \dots, q_i^k\}$, and a threshold $\delta > 0$, we can decide if
$\frmin(\mathcal{U}, \mathcal{V}) \leq \delta$ in time $\Theta(mnk^4)$ in the worst case, using
$\Theta(mk^3)$ space.
\end{theorem}

\subsection{Approximation by Grids}\seclab{grids}
Given a polygonal curve $\curveA$ and a general uncertain curve $\U$, in this section we show how to
find a curve $\curveB\Subset\U$ such that $\fr(\curveA,\curveB) \leq (1 + \eps) \frmin(\curveA,\U)$. 
This is accomplished by carefully discretising the regions, in effect approximately reducing the
problem to the indecisive case, for which we then can use \thmref{indecisivedecider}. 

For simplicity we assume the uncertain regions have constant complexity.
Throughout, we assume $\frmin(\curveA, \U) > 0$, as justified by the following lemma.
\begin{lemma}\lemlab{equal}
Let $\curveA$ be a polygonal curve with $n$ vertices, and $\U$ an uncertain curve with $m$ vertices.
Then one can determine whether $\frmin(\curveA,\U) = 0$ in $O(mn)$ time.
\end{lemma}
\begin{proof}
If for some $i$, $\curveA_{i}$ lies on the segment $\curveA_{i - 1}\curveA_{i + 1}$, then
$\fr(\curveA, \curveA') = 0$, where $\curveA' = \langle\curveA_1, \dots, \curveA_{i - 1},
\curveA_{i + 1}, \dots, \curveA_n\rangle$.
So we can assume that no vertex of $\curveA$ lies on the segment between its neighbours, as
otherwise we can remove that vertex and get the same result in terms of Fr\'echet distance.  

Thus, at every vertex $\curveA$ turns, implying that if there exists $\curveB \Subset \U$ such that
$\fr(\curveA, \curveB) = 0$, then for all $i$, $\curveA_i$ must match to some $\curveB_j$.

This observation leads to a simple decision procedure.
Define
\[s(i) = \{1 \leq j \leq m \mid \fr(\curveA[1: i], \curveB[1: j])=0\}\,,\]
so a set of indices on $\curveB$ that yield the zero Fr\'echet distance between the correspondent
prefix curves.
Then we can go through $\curveA$ one vertex at a time, maintaining $s(i)$, and ultimately
$\frmin(\curveA, \U) = 0$ if and only if $m \in s(n)$.

Initially, $s(1)=\{1 \leq j \leq m \mid \forall 1 \leq  k \leq j: \curveA_1 \in \ur_k\}$, which is
easy to test and compute.
For $i > 1$, $s(i)$ can be computed from $s(i - 1)$ as follows.
Let $\Stab_i(k)$ be the set of indices $j > k$ such that there exist points $p_{k + 1}, \dots,
p_{j - 1}$, appearing in order along $\curveA_{i - 1}\curveA_i$, where $p_\ell \in \ur_\ell$ for
all $k < \ell < j$.
(Note that we always have $k + 1\in \Stab_i(k)$.)
So, $\Stab_i(k)$ is the set of indices $j$ of uncertainty regions, starting from $k + 1$, such that
all the regions between $k$ and $j$ are stabbed by the segment $\curveA_{i - 1}\curveA_i$ in the
correct order.
Then we have
\[s(i) = \{j \mid \curveA_i \in \ur_j \land j \in \Stab_i(k) \text{ where }
k = \max_{\ell < j} \{\ell \in s(i - 1)\}\}\,.\]
From this definition of $s(i)$ it is easy to see that it can be computed in $O(m)$ time given
$s(i - 1)$, and thus the total time required is $O(mn)$.
In particular, if $s(i - 1)$ is non-empty, then let $z$ be the minimum value in $s(i - 1)$.
We now incrementally loop over values of $j$, where initially $j = z + 1$, and add $j$ to $s(i)$ if
$\curveA_i \in \ur_j$ and $j \in \Stab_i(z)$.
Note that in constant time per iteration we can maintain sufficient information to determine if
$j\in \Stab_i(z)$, as we describe further.
If at any iteration $j = z' + 1$ for $z' \in s(i - 1)$, we forget $\Stab_i(z)$ (as we no longer need
to stab those regions) and start maintaining and checking $\Stab_i(z')$.

Note that the intersection of any $\ur_\ell$ with $\curveA_{i - 1}\curveA_i$ is a constant number of
intervals along $\curveA_{i - 1}\curveA_i$.
Then $\Stab_i(k)$ can be computed incrementally as follows.
First, let $p_{k+1}$ be the earliest point of $\curveA_{i - 1}\curveA_i \cap \ur_{k + 1}$.
For some $j > k + 1$, let $p_j$ be the earliest point of $\curveA_{i - 1}\curveA_i \cap \ur_{j}$,
which is at least as far as along $\curveA_{i - 1}\curveA_i$ as $p_{j - 1}$ (if it exists).
If such $p_j$ exists, then we know that $j \in \Stab_i(k)$. 
Maintaining this information indeed takes constant time per iteration.
\end{proof}

\subsubsection{Decision Procedure}
\seclab{decproc}
We call an algorithm a \emph{$(1 + \eps)$-decider} for \probref{decision}, if when
$\frmin(\curveA, \U) \leq \delta$, it returns a curve $\curveB \Subset \U$ such that
$\fr(\curveA, \curveB) \leq (1 + \eps) \delta$, and when $\frmin(\curveA, \U) > (1 + \eps) \delta$,
it returns \False\ (in between either answer is allowed).
In this section, we present a \emph{$(1 + \eps)$-decider} for \probref{decision}.
We make use of the following standard observation.
\begin{observation}\obslab{perturbation}
Given a curve $\curveA = \langle\curveA_1, \dots, \curveA_n\rangle$, call a curve
$\curveB = \langle\curveB_1, \dots, \curveB_n\rangle$ an $r$-perturbation of $\curveA$ if
$\lVert\curveA_i - \curveB_i\rVert \leq r$ for all $i$.
Since $\lVert\curveA_i - \curveB_i\rVert, \lVert\curveA_{i + 1} - \curveB_{i + 1}\rVert\leq r$, all
points of the segment $\curveB_i \curveB_{i + 1}$ are within distance $r$ of
$\curveA_i \curveA_{i + 1}$.
For segments this implies that $\fr(\curveA_i \curveA_{i + 1}, \curveB_i \curveB_{i + 1}) \leq r$,
which implies that $\fr(\curveA, \curveB) \leq r$ by composing the mappings for all $i$.
\end{observation}
The high-level idea is to replace $\U$ with the set of grid points it intersects, however, as our
uncertain regions may avoid the grid points, we need to include a slightly larger set of points.

\begin{definition}\deflab{expand}
Let $\ur$ be a compact subset of $\R^d$.
We now define the set of points $\EG_r(\ur)$ which we call the \emph{expanded $r$-grid points} of
$\ur$.

Let $B(\sqrt{d} r)$ denote the ball of radius $\sqrt{d} r$, centred at the origin.
Let $\Thick(\ur, r) = \ur \oplus B(\sqrt{d} r)$, where $\oplus$ denotes Minkowski sum.  
Let $G_r$ denote the regular grid of side length $r$, and let $GT_r(\ur)$ denote the subset of grid
vertices from $G_r$ that fall in $\Thick(\ur, r)$.
Finally, we define
\[\EG_r(\ur) = \{p \mid p = \argmin_{q \in \ur} \lVert q - x\rVert \text{ for } x \in GT_r(\ur)\}\,.\]
\end{definition}

\noindent In the following observation and proof we make use of the terms defined above.
\begin{observation}\obslab{distance}
For any $x \in \ur$, there is a point $p \in \EG_r(\ur)$ such that
$\lVert p - x\rVert \leq 2\sqrt{d} r$.
\end{observation}
\begin{proof}
For any point $x \in \ur$, let $g$ be its nearest grid point in $G_r$.
Since $\lVert x - g\rVert \leq \sqrt{d} r$, we know that
$g \in \Thick(\ur, r) = \ur \oplus B(\sqrt{d} r)$.
So let $p$ be the point in $\ur$ which is closest to $g$; thus, $p \in \EG_r(\ur)$.
Therefore, $\lVert x - p\rVert \leq \lVert x - g\rVert + \lVert g - p\rVert
\leq \sqrt{d} r + \sqrt{d} r = 2\sqrt{d} r$.
\end{proof}

\begin{lemma}\lemlab{decider}
There is a $(1 + \eps)$-decider for \probref{decision} with running time
$O(mn \cdot (1 + (\Delta / (\eps\delta))^{2d}))$, for $1 \geq \eps >0$, where
$\Delta = \max_{i} \diam{\ur_i}$ is the maximum diameter of an uncertain region.
\end{lemma}
\begin{proof}
It will help with the analysis if $\eps\delta < \Delta$.
To ensure this we first do the following.
Select an arbitrary curve $x \Subset \U$.
Now using the standard $O(mn)$ time exact decider for \Frechet distance~\cite{alt:1995}, query
whether $\fr(\curveA, x) \leq (1 + \eps)\delta$.
If the decider returns $\fr(\curveA, x) \leq (1 + \eps)\delta$, then we can return $x$ as our
solution.
Otherwise, $\fr(\curveA, x) > (1 + \eps)\delta$, and we next query whether 
$\fr(\curveA, x) \leq \Delta + \delta$.
By \obsref{perturbation} and the triangle inequality,
$\fr(\curveA, x) \leq \Delta + \frmin(\curveA, \U)$.
Thus, if the decider returns $\Delta + \delta < \fr(\curveA, x)$, then $\delta < \frmin(\curveA, \U)$,
and so we return \False.
Otherwise, the two decider calls tell us that
$(1 + \eps) \delta < \fr(\curveA, x) \leq \Delta + \delta$, implying $\eps\delta < \Delta$.

Let $r = (\eps\delta) / (2\sqrt{d})$, and for $\ur_i \in \U$, let $E_i = \EG_r(\ur_i)$ denote the
expanded $r$-grid points of $\ur_i$, as defined in \defref{expand}.
Consider the indecisive curve $\U' = \langle E_1, \dots, E_m\rangle$.
We call the algorithm $\decider((1 + \eps)\delta, \curveA, \U')$ of \thmref{indecisivedecider} and
return whatever it returns, i.e.\@ if it returns a curve, then we return that curve, and if it
returns that $\frmin(\curveA, \U') > (1 + \eps)\delta$, then we return that
$\frmin(\curveA, \U) > (1 + \eps)\delta$.

First observe that $E_i \subseteq \ur_i$, and thus $\frmin(\curveA, \U) \leq \frmin(\curveA, \U')$.
So if $\frmin(\curveA, \U) > (1 + \eps)\delta$, then the decider must return
$\frmin(\curveA, \U') > (1 + \eps)\delta$, as desired.
Now suppose that $\frmin(\curveA, \U)\leq \delta$.
In this case, we argue that our algorithm outputs a curve $\curveB' \Subset \U$ such that
$\fr(\curveA, \curveB') \leq (1 + \eps)\delta$.
To do so it suffices to argue that there exists some curve $\sigma' \in \U'$ such that
$\fr(\curveA, \curveB') \leq (1 + \eps)\delta$, as then \thmref{indecisivedecider} guarantees the
decider outputs a curve (which is in $\Real{U}$, as it is a superset of $\Real{U'}$).
So let $\curveB = \langle\curveB_1, \dots, \curveB_m\rangle$ be the curve in $\Real{U}$ realising
the \Frechet distance to $\curveA$, that is, $\fr(\curveA, \curveB) = \frmin(\curveA, \U)$.
Let $\curveB' = \langle\curveB'_1, \dots, \curveB'_m\rangle$ be the curve such that
$\curveB'_i = \min_{x \in E_i} \lVert x - \curveB_i\rVert$.
Note that by \obsref{distance}, we have $\lVert\curveB_i - \curveB'_i\rVert \leq 2\sqrt{d} r$ for
all $i$.
Thus, $\curveB'$ is a $2\sqrt{d} r$-perturbation of $\curveB$ as described in \obsref{perturbation},
and so $\fr(\curveB, \curveB') \leq 2\sqrt{d} r = \eps\delta$.
As the \Frechet distance satisfies the triangle inequality, we therefore have 
$\fr(\curveA, \curveB') \leq \fr(\curveA, \curveB) + \fr(\curveB, \curveB') \leq \delta + \eps\delta
= (1 + \eps)\delta$.
Thus, as $\curveB' \Subset \U'$, when our algorithm calls $\decider((1 + \eps)\delta, \curveA, \U')$,
it returns a curve.

For the running time, recall we first spent $O(mn)$ time to ensure $\eps\delta<\Delta$, in which
case we must bound the number of points in each $E_i$.
By \defref{expand}, for all $i$, the number of points in $E_i$ is bounded by the number of grid
points in the region $\Thick(\ur_i, r)$. 
This region is the Minkowski sum of a compact set of diameter at most $\Delta$ with a radius
$\sqrt{d} r$ ball, so its diameter is at most $\Delta + 2\sqrt{d} r$.
Thus, the number of grid points and hence $\lvert E_i\rvert$ is 
$O(((\Delta + 2\sqrt{d} r) / r)^d) = O((2\sqrt{d} \Delta / (\eps \delta) + 2\sqrt{d})^d) =
O((\Delta / (\eps \delta) + 1)^d) = O((\Delta / (\eps\delta))^d)$.
Thus, by \thmref{indecisivedecider}, the call to $\decider$ takes time
$O(mn (\Delta / (\eps\delta))^{2d})$, which bounds the total time of our algorithm.
\end{proof}

\subsubsection{Optimisation}
\begin{theorem}\thmlab{optimize}
Let $\curveA$ be a polygonal curve with $n$ vertices, $\U$ an uncertain curve with $m$ vertices, and
$\delta = \frmin(\curveA, \U)$.
Then for any $1 \geq \eps > 0$, there is an algorithm which returns a curve $\curveB \Subset \U$
such that $\fr(\curveA, \curveB) \leq (1 + \eps)\delta$, whose running time is
$O(mn(\log(mn) + (\Delta / (\eps\delta))^{2d}))$, where $\Delta = \max_{i} \diam{\ur_i}$ is the
maximum diameter of an uncertain region.
\end{theorem}
\begin{proof}
Fix an arbitrary curve $x \Subset \U$. 
First, we compute the \Frechet distance between $\curveA$ and $x$.
If $\fr(x, \curveA) \geq \Delta + \Delta / \eps$, then we return $x$ as our solution.
To see why this is valid, let $\hat{\curveB} \Subset \U$ be an optimal solution, that is,
$\fr(\curveA, \hat{\curveB}) = \frmin(\curveA, \U)$.
Note that $x$ is a $\Delta$-perturbation of $\hat{\curveB}$, and thus by the triangle inequality and
\obsref{perturbation},
\[\fr(x, \curveA) \leq \fr(x, \hat{\curveB}) + \fr(\hat{\curveB}, \curveA) \leq
\Delta + \fr(\hat{\curveB}, \curveA)\,.\]
If $\Delta + \Delta / \eps \leq \fr(x, \curveA)$, then plugging in the inequality above implies that
$\Delta \leq \eps \cdot \fr(\hat{\curveB}, \curveA)$, which in turn implies that
\[\fr(x, \curveA) \leq \Delta + \fr(\hat{\curveB}, \curveA) \leq (1 + \eps) \cdot
\fr(\hat{\curveB}, \curveA)\,.\]

So suppose that $\fr(x, \curveA) < (1 + 1/\eps)\Delta$, in which case
\[\frmin(\curveA, \U) = \fr(\curveA, \hat{\curveB}) \leq \fr(\curveA, x) + \fr(x, \hat{\curveB}) <
(1 + 1/\eps)\Delta + \Delta = (2 + 1/\eps)\Delta = \gamma\,.\]
Let $\gridDecider(\curveA, \U, \eps', \delta)$ denote the $(1 + \eps')$-decider of \lemref{decider},
which correctly returns either \False\ (i.e.\@ $\frmin(\curveA, \U) > \delta$) or a curve in
$\Real{U}$ with \Frechet distance at most $(1 + \eps')\delta$ to $\curveA$.
We perform a decreasing exponential search using \gridDecider.
Specifically, starting at $i = 0$, we call
$\gridDecider(\curveA, \U, \eps/4, \gamma / (1 + \eps/4)^i)$.
If $\gridDecider$ returns a curve (i.e.\@ \True), we increment $i$ by $1$ and repeat, otherwise if
$\gridDecider$ outputs \False, we return the curve from iteration $i - 1$.
(Note that \gridDecider cannot return \False\ when $i = 0$, as this would imply that
$\frmin(\curveA, \U) > \gamma$.)

Let $j$ denote the index when the algorithm stops.
So we know that we got \False\ from $\gridDecider(\curveA, \U, \eps/4, \gamma / (1 + \eps/4)^j)$,
and $\gridDecider(\curveA, \U, \eps/4, \gamma / (1 + \eps/4)^{j - 1})$ returned a curve
$\curveB \Subset \U$ such that $\fr(\curveA, \curveB) \leq (1 + \eps/4) \cdot
(\gamma / (1 + \eps/4)^{j - 1})$.
Therefore,
\[\gamma / (1 + \eps/4)^{j} < \frmin(\curveA, \U) \leq \fr(\curveA, \curveB) \leq
(1 + \eps/4) \cdot (\gamma/(1+\eps/4)^{j - 1}) = \gamma / (1 + \eps/4)^{j - 2}\,,\]
which implies that
\[\fr(\curveA, \curveB) \leq (1 + \eps/4)^2 \frmin(\curveA, \U) =
(1 + \eps/2 + \eps^2/16) \cdot \frmin(\curveA, \U) < (1 + \eps) \cdot \frmin(\curveA, \U)\,.\]

As for the running time, by \lemref{decider}, the time for the $i$th call to \gridDecider is
\[O\pth{mn \pth{\frac{(1 + \eps/4)^i \Delta}{\eps \gamma}}^{2d}}
= O\pth{mn \pth{\frac{(1 + \eps/4)^i\Delta}{\eps (2 + 1/\eps)\Delta}}^{2d}}
= O\pth{mn \pth{1 + \eps/4}^{2d\cdot i}}.\]

Let $\delta = \frmin(\curveA, \U)$, and let $j$ be the index the last time $\gridDecider$ is called.
By the argument above, $\delta \leq \gamma / (1 + \eps/4)^{j - 2}$, which implies that
$j - 2 \leq \log_{1 + \eps/4} (\gamma / \delta)$.
As \gridDecider is called $j + 1$ times, and the running times for the calls to \gridDecider form an
increasing geometric series, the total time for all calls to \gridDecider is
\begin{align*}
&\phantom{{}={}} O\pth{mn(1 + \eps/4)^{2d \cdot (3 + \log_{1 + \eps/4} (\gamma / \delta))}}
= O\pth{mn(1 + \eps/4)^{2d \cdot \log_{1 + \eps/4} (\gamma / \delta)}}\\
&= O\pth{mn(\gamma / \delta)^{2d \cdot \log_{1 + \eps/4} (1 + \eps/4)}}
= O\pth{mn\pth{\frac{\gamma}{\delta}}^{2d}}\\
&= O\pth{mn\pth{\frac{(2 + 1/\eps)\Delta}{\delta}}^{2d}}
= O\pth{mn\pth{\frac{\Delta}{\eps\delta}}^{2d}}.
\end{align*}
As it takes $O(mn\log(mn))$ time to initially compute $\fr(x, \curveA)$ using the algorithm of
Alt and Godau~\cite{alt:1995}, the total running time is
$O(mn(\log(mn) + (\Delta / (\eps\delta))^{2d}))$.
\end{proof}

If the polygonal curve $\curveA$ is replaced with an uncertain curve $\W$, it is easy to see that by
discretising both $\W$ and $\U$, the same analysis gives an algorithm to compute $\frmin(\W, \U)$.
The only difference now is that we must cite \thmref{indind} instead of \thmref{indecisivedecider},
yielding the following.
\begin{corollary}
Let $\W$ and $\U$ be uncertain curves with $n$ and $m$ vertices, respectively, and
$\delta = \frmin(\W, \U)$.
Then for any $0 < \eps \leq 1$, there is an algorithm returning curves $\curveA \Subset \W$ and
$\curveB \Subset \U$ such that $\fr(\curveA, \curveB) \leq (1 + \eps)\delta$, whose running time is
$O(mn(\log(mn) + (\Delta / (\eps\delta))^{4d}))$, where $\Delta$ is the maximum diameter of an
uncertain region.
\end{corollary}

\subsection{Greedy Algorithm}\seclab{greedy}
Here we argue that there is a simple $3$-decider for \probref{decision}, running in near-linear time
in the plane.
Roughly speaking, the idea is to greedily and iteratively pick $\curveB_i \in \ur_i$ so as to allow
us to get as far as possible along $\curveA$.
Without any assumptions on $\U$, this greedy procedure may walk too far ahead and get stuck.
Thus, in this section we assume that consecutive $\ur_i$ are separated, so as to ensure optimal
solutions do not lag too far behind.
Here we also assume that $\ur_i$ are convex, i.e.\@ imprecise, and have constant complexity, as it
simplifies certain definitions.
Throughout this section let $\curveA = \langle\curveA_1, \dots, \curveA_n\rangle$ be a polygonal
curve and let $\U = \langle\ur_1, \dots, \ur_m\rangle$ be an uncertain curve.

\begin{definition} 
Call $\U$ \emph{$\gamma$-separated} if for all $1 \leq i < m$, $\lVert \ur_i - \ur_{i + 1}\rVert >
\gamma$ and each $\ur_i$ is convex.
Define an \emph{$r$-visit} of $\ur_i$ to be any maximal-length contiguous portion of
$\curveA \cap (\ur_i \oplus B(2r))$ which intersects $\ur_i \oplus B(r)$, where $\oplus$ denotes
Minkowski sum.
If $\U$ is $\gamma$-separated for $\gamma \geq 4r$, then any $r$-visit of $\ur_i$ is disjoint from
any $r$-visit of $\ur_j$ for $i \neq j$, in which case define the \emph{true $r$-visit} of $\ur_i$
to be the first visit of $\ur_i$ which occurs after the true $r$-visit of $\ur_{i - 1}$.
(For $\ur_1$ it is the first $r$-visit.)
\end{definition}

\begin{lemma}\lemlab{separated}
If $\U$ is $\gamma$-separated for $\gamma \geq 4r$, then for any curve $\curveB \Subset \U$ and any
reparametrisations $f$ and $g$ such that $\width_{f, g}(\curveA, \curveB) \leq r$, $\curveB_i$ must
map to a point on the true $r$-visit of $\ur_i$ for all $i$.
\end{lemma}
\begin{proof}
First, note that since $\width_{f, g}(\curveA, \curveB) \leq r$, $\curveB_i$ must map to a point in
an $r$-visit of $\ur_i$, and thus we only need to prove it is the true $r$-visit.

We prove the claim by induction on $i$.
For $i = 1$, the claim holds, as $\curveB_1$ must map to $\curveA_1$, and $\curveA_1$ is in the
first $r$-visit of $\ur_1$, which is its true $r$-visit.

Now suppose the claim holds for $i-1$.
$\curveB_i$ must map to a point on an $r$-visit of $\ur_i$, and by the induction hypothesis, this
visit must happen after the true $r$-visit of $\ur_{i-1}$ on $\curveA$.
Moreover, as $\U$ is $4r$-separated, the first point in $\ur_i \oplus B(r)$ of the first $r$-visit
of $\ur_i$ that occurs after the true $r$-visit of $\ur_{i - 1}$ (i.e.\@ true $r$-visit of $\ur_i$)
must map to a point $x$ on $\curveB_{i - 1} \curveB_i$.
Note, however, that as both $x$ and $\curveB_i$ map to points in $\ur_i \oplus B(r)$, the portion of
$\curveA$ that the segment $x \curveB_i$ maps to must lie within $\ur_i \oplus B(2r)$, i.e.\@ the
same $r$-visit.
Therefore, all of $x \curveB_i$ is mapped to the true $r$-visit of $\ur_i$, completing the proof.
\end{proof}

\noindent For two points $\alpha$ and $\beta$ on $\curveA$, let $\alpha \leq \beta$ denote that
$\alpha$ occurs before $\beta$, and for any points $\alpha \leq \beta$ let $\curveA(\alpha, \beta)$
denote the subcurve between $\alpha$ and $\beta$.
\begin{definition}\deflab{gsgr}
The \emph{$\delta$-greedy sequence} of $\curveA$ with respect to $\U$, denoted
$\gs(\curveA, \U, \delta)$, is the longest possible sequence $\alpha = \langle \alpha_1, \dots,
\alpha_k\rangle$ of points on $\curveA$, where $\alpha_1 = \curveA_1$, and for any $i > 1$,
$\alpha_i$ is the point furthest along $\curveA$ such that $\lVert \alpha_i - \ur_i\rVert \leq
\delta$ and $\fr(\alpha_{i - 1} \alpha_i, \curveA(\alpha_{i - 1}, \alpha_i)) \leq 2\delta$.
\end{definition}

\begin{observation}\obslab{contained}
For any $i \leq k$, let $\alpha^i = \langle\alpha_1, \dots, \alpha_i\rangle$ be the $i$th prefix of
$\gs(\curveA, \U, \delta)$.
Then $\fr(\alpha^i, \curveA(\alpha_1, \alpha_i)) \leq 2\delta$, and
$\alpha^i \Subset \U_i \oplus B(\delta)$, where $\U_i \oplus B(\delta) =
\langle \ur_1 \oplus B(\delta), \dots, \ur_i \oplus B(\delta)\rangle$.
\end{observation}

\begin{lemma}\lemlab{greedyinduction}
If $\U$ is $10\delta$-separated and $\frmin(\curveA, \U) \leq \delta$, then
$\gs(\curveA, \U, \delta)$ has length $m$ and $\alpha_m = \curveA_n$.
\end{lemma}
\begin{proof}
Let $\gs(\curveA, \U, \delta) = \alpha = \langle \alpha_1, \dots, \alpha_k\rangle$.
Let $\opt = \langle \opt_1, \dots, \opt_n\rangle$ be any curve in $\Real{U}$ such that
$\fr(\curveA, \opt) = \frmin(\curveA, \U)$.
Throughout we fix a mapping realising $\fr(\curveA, \opt)$ and let $\beta_i$ be the point on
$\curveA$ which $\opt_i$ maps to under this mapping.
For the curve $\alpha$ we fix the mapping which is the composition of the maps realising
$\fr(\alpha_{i - 1} \alpha_i, \curveA(\alpha_{i - 1}, \alpha_i)) \leq 2\delta$, and in particular
$\alpha_i$ on $\alpha$ maps to $\alpha_i$ on $\curveA$.

We prove by induction that for $i \leq m$, $\alpha_i$ exists and $\beta_i \leq \alpha_i$.
So assume that $\alpha_{i - 1}$ exists.
By \obsref{contained}, $\alpha^{i - 1} \Subset \U_{i - 1} \oplus B(\delta)$, and, moreover,
$\fr(\curveA(\alpha_1, \alpha_{i - 1}), \alpha^{i - 1}) \leq 2\delta$.
Since $\U$ is $10\delta$-separated, $\U_{i - 1} \oplus B(\delta)$ is $8\delta$-separated, and thus
by \lemref{separated}, $\alpha_{i - 1}$ is on the true $2\delta$-visit of
$\ur_{i - 1} \oplus B(\delta)$ by the prefix curve $\curveA(\alpha_1, \alpha_{i - 1})$.
Observe that the true $2\delta$-visit of $\ur_{i - 1} \oplus B(\delta)$ by the prefix curve
$\curveA(\alpha_1, \alpha_{i - 1})$ is a subset of the true 2$\delta$-visit of
$\ur_{i - 1} \oplus B(\delta)$ by $\curveA$, and thus $\alpha_{i - 1}$ is on the true
$2\delta$-visit of $\ur_{i - 1} \oplus B(\delta)$ by $\curveA$.
We also have that $\opt \Subset \U \oplus B(\delta)$, as $\ur_j \subset \ur_j \oplus B(\delta)$ for
all $j$, so by \lemref{separated}, $\beta_{i - 1}$ and $\beta_i$ are on the true $2\delta$-visit of
$\ur_{i - 1} \oplus B(\delta)$ and $\ur_i \oplus B(\delta)$.
In particular, this implies that $\beta_{i - 1} \leq \alpha_{i - 1} \leq \beta_i$, as the true
$2\delta$-visits of $\ur_{i - 1} \oplus B(\delta)$ and $\ur_i \oplus B(\delta)$ are disjoint.
Thus, some point $x$ on the segment $\opt_{i - 1} \opt_i$ must map to $\alpha_{i - 1}$.
Note that $\fr(x \opt_i, \curveA(\alpha_{i - 1}, \beta_i)) \leq \delta$.
As $\lVert x - \alpha_{i - 1}\rVert \leq \delta$, $\fr(x \opt_i, \alpha_{i - 1} \opt_i) \leq \delta$,
and so by the triangle inequality for \Frechet distance,
$\fr(\alpha_{i - 1} \opt_i, \curveA(\alpha_{i - 1}, \beta_i)) \leq 2\delta$.
Since $\lVert \beta_i - \opt_i\rVert \leq \delta$, $\beta_i$ is a possible choice for $\alpha_i$,
and thus $\alpha_i$ exists and $\beta_i \leq \alpha_i$.
Finally, since $\alpha_i$ exists for all $i \leq m$, $\alpha = \gs(\curveA, \U, \delta)$ has length
$m$, and moreover, since $\beta_m \leq \alpha_m$ and $\beta_m = \curveA_n$, we conclude that
$\alpha_m = \curveA_n$.
\end{proof}

The following lemma is the only place where we require the points to be in $\R^2$.
The proof uses a result from Guibas et al.~\cite{guibas:1993}.
\begin{lemma}\lemlab{stabbing}
For $\curveA$ and $\U$ in $\R^2$, where $\U$ is $10\delta$-separated, $\gs(\curveA, \U, \delta)$ is
computable in $O(m + n\log n)$ time.
\end{lemma}
\begin{proof}
Given $\alpha_i$ from $\gs(\curveA, \U, \delta)$, we describe how to compute $\alpha_{i + 1}$, if it
exists.
Let $\curveA_j$ be the smallest-index vertex such that $\alpha_i < \curveA_j$.
Let $D_j, \dots, D_n$ be the sequence of $2\delta$-radius disks, where $D_l$ is centred at
$\curveA_l$.
Observe that for $\alpha_{i + 1}$ to be able to lie on $\curveA_z \curveA_{z + 1}$, for any
$z \geq j$, we first require that $\fr(\alpha_i \alpha_{i + 1}, \curveA(\alpha_i, \alpha_{i + 1}))
\leq 2\delta$, which occurs if and only if there exist points $p_j, \dots, p_z$ that appear in order
along $\alpha_i \alpha_{i + 1}$ such that $p_l \in D_l$.
Clearly, such points are necessary, but they are also sufficient, as $\fr(p_l p_{l + 1},
\curveA_l \curveA_{l + 1}) \leq 2\delta$.
(As $\alpha_i$ and $\alpha_{i + 1}$ lie on $\curveA$, the same holds for $\alpha_i \curveA_j$ and
$\curveA_z \alpha_{i + 1}$.)
$\gs(\curveA, \U, \delta)$ also requires that $\alpha_{i + 1}$ lie within distance $\delta$ of
$\ur_{i + 1}$.
This is equivalent to requiring that $\curveA_z \curveA_{z + 1}$ intersects
$\ur_{i + 1} \oplus B(\delta)$.
As both $\curveA_z \curveA_{z + 1}$ and $\ur_{i + 1} \oplus B(\delta)$ are convex regions, their
intersection is convex, i.e.\@ a single subsegment of $\curveA_z \curveA_{z + 1}$.
Let $S_{i + 1}(z)$ denote this segment, which we can compute in constant time, as $\ur_{i + 1}$ is a
constant complexity convex region.
Note that $\alpha_{i + 1}$ may lie on the same segment of $\curveA$ as $\alpha_i$,
i.e.\@ $z = j - 1$, which is an easier case, as no disks need to be intersected and
$\fr(\alpha_i \alpha_{i + 1}, \curveA(\alpha_i, \alpha_{i + 1})) \leq 2\delta$ holds.

Given an order sequence of $k$ equal radius disks $D_1, \dots, D_k$, say that a line $\ell$ stabs
the disks if for all $j \leq k$, there exists a point $p_j \in \ell \cap D_j$ such that the $p_j$
appear in order along $\ell$.
Guibas et al.~\cite{guibas:1993} gave an $O(k \log k)$-time incremental algorithm that determines
the set of all stabbing lines.
As follows from the description of our problem, their algorithm can be used to determine
$\alpha_{i + 1}$ given $\alpha_i$ by restricting the stabbing line to first pass through $\alpha_i$
and requiring it to intersect $S_{i + 1}(k)$ at the end.

We now sketch the necessary changes.
Their algorithm inserts the disks in order, maintaining three objects\dsh the support hull, limiting
lines, and line stabbing wedge.
The support hull consists of a pair of upper and lower concave chains that all stabbers must pass
between, and the limiting lines represent the largest and smallest slope stabbers.
The wedge is the set of all points $p$ such that there is a stabber that passes through $p$ after
passing through the required points from the disks.

To modify their approach for our setting, we require the stabber to initially pass through
$\alpha_i$.
This actually simplifies the problem by joining and collapsing the chains of the support hull,
\footnote{Alternatively, one can enforce the condition by defining an initial zero-radius disk
$D_0$ at $\alpha_i$, and indeed the referenced work~\cite{guibas:1993} considers stabbers for more
general collections of convex objects.}
and thus we can focus on the wedge.
After $j$ insertions, the wedge boundary consists of $O(j)$ pieces from the disks, flanked by the
limiting lines.
These ordered boundary pieces are stored in a binary tree to facilitate logarithmic time updates
when a new disk is inserted, and we can simply reuse this structure to determine the intersection of
the wedge with $S_{i + 1}(j)$.

By \defref{gsgr}, the line segment $\curveA_z \curveA_{z + 1}$ that $\alpha_{i + 1}$ lies on must
have $z$ be as large as possible.
Thus, we run the above incremental procedure, where in the $j$th round we check for intersection
with $S_{i + 1}(j)$.
If no such intersection is found before we reach the end of the $\curveA$ or the wedge becomes
empty, then $\alpha_{i + 1}$ does not exist.
Otherwise, $\alpha_{i + 1}$ is defined.
However, the rounds which have intersection with $S_{i + 1}(j)$ need not be contiguous; thus, care
is needed to determine the last such intersection efficiently.

Let $k$ be the largest index such that $\alpha_k$ is defined.
By \obsref{contained}, for any $i \leq k$, $\fr(\alpha^i, \curveA(\alpha_1, \alpha_i)) \leq 2\delta$
and $\alpha^i \Subset \U_i \oplus B(\delta)$.
Since $\U$ is $10\delta$-separated, $\U_i \oplus B(\delta)$ is $8\delta$-separated, and so by
\lemref{separated}, $\alpha_i$ must be in the true $2\delta$-visit of $\ur_i \oplus B(\delta)$ by
$\curveA(\alpha_1, \alpha_k)$.
Thus, when computing $\alpha_i$, we only need to consider vertices from $\curveA$ which occur after
$\alpha_{i - 1}$ and before the end of the true $2\delta$-visit of $\ur_i \oplus B(\delta)$.
If $n_i$ is the number of such vertices, it therefore takes $O(1 + n_i \log n_i)$ time to compute
$\alpha_i$ with the algorithm above.
Moreover, as the true $2\delta$-visits for $\ur_i \oplus B(\delta)$ and $\ur_j \oplus B(\delta)$ for
$i \neq j \leq k$ are disjoint, any vertex of $\curveA$ can be counted by at most two of the $n_i$,
and so $\sum_i n_i = n$.
Thus, the total running time is $O(m + n \log n) + \sum_{i = 1}^k O(1 + n_i \log n_i) =
O(m + n \log n)$, where the leading $O(m + n \log n)$ term accounts for the time to determine if
$\alpha_{k + 1}$ does not exist for $k < m$.
\end{proof}

\begin{theorem}
Let $\U$ be $10r$-separated for some $r > 0$.
There is a $3$-decider for \probref{decision} with running time $O(m + n \log n)$ in the plane that
works for any query value $0 < \delta \leq r$.
\end{theorem}
\begin{proof}
Compute $\gs(\curveA, \U, \delta)$.
If it has length $m$, then let $\curveB = \langle \curveB_1, \dots, \curveB_m\rangle$ be any curve
in $\Real{U}$ such that $\lVert \curveB_i - \alpha_i\rVert \leq \delta$ for all $i$.
If this occurs and if $\alpha_m = \curveA_n$, we output $\curveB$ as our solution, and otherwise we
output \False.
Thus, the running time follows from \lemref{stabbing}.

Observe that if we output a curve $\curveB$, then $\fr(\curveB, \curveA) \leq 3\delta$, using the
triangle inequality:
\[\fr(\curveB, \curveA) \leq \fr(\curveB, \alpha) + \fr(\alpha, \curveA) \leq \delta + 2\delta =
3\delta\,.\]
Thus, we only need to argue that when $\frmin(\curveA, \U)\leq \delta$, a curve is produced, which
is immediate from \lemref{greedyinduction}.
\end{proof}

\section{Algorithms for Upper Bound and Expected \Frechet Distance}\label{sec:sakoe-chiba}
As shown in Section~\ref{sec:proofs}, finding the upper bound and expected discrete and continuous
Fr\'echet distance is hard even for simple uncertainty models.
However, restricting the possible couplings between the curves makes the problem solvable in
polynomial time.
In this section, we use \emph{indecisive} curves.
Define a Sakoe--Chiba time band~\cite{sakoe:1978} in terms of reparametrisations of the curves: for
a band of width $w$ and all $t \in [0, 1]$, if $\phi_1(t) = x$, then $\phi_2(t) \in [x - w, x + w]$.
In the discrete case we can only couple point $i$ on one curve to points $i \pm w$ on the other
curve.

\subsection{Upper Bound Discrete Fr\'echet Distance: Precise and Indecisive}\label{sec:alg_pr_ind}
First of all, let us discuss a simple setting.
Suppose we are given a curve $\sigma = \langle q_1, \dots, q_n\rangle$ of $n$ precise points
and $\mathcal{U} = \langle U_1, \dots, U_n\rangle$ of $n$ indecisive points, each of them having
$\ell$ options, so for all $i \in [n]$ we have $U_i = \{p_i^1, \dots, p_i^\ell\}$.
We would like to answer the following decision problem:
\emph{`If we restrict the couplings to a Sakoe--Chiba band of width $w$, is it true that
$\dfrmax(\mathcal{U}, \sigma) \leq \delta$ for some given threshold $\delta > 0$?'}
So, we want to solve the decision problem for the upper bound discrete Fr\'echet distance between
a precise and an indecisive curve.

In a fully precise setting the discrete Fr\'echet distance can be computed using dynamic
programming~\cite{eiter:1994}.
We create a table where the rows correspond to vertices of one curve, say $\sigma$, and columns
correspond to vertices of the other curve, say $\pi$.
Each table entry $(i, j)$ then contains a \True\ or \False\ value indicating if there is a coupling
between $\sigma[1: j]$ and $\pi[1: i]$ with maximum distance at most $\delta$.
We use a similar approach.

Suppose we position $\mathcal{U}$ to go horizontally along the table, and $\sigma$ to go vertically.
Consider an arbitrary column in the table and suppose that we fix the realisation of $\mathcal{U}$
up to the previous column.
Then we can simply consider the new column $\ell$ times, each time picking a different realisation
for the new point on $\mathcal{U}$, and compute the resulting reachability.
As we do this for the entire column at once, we can ensure consistency of our choice of realisation.
This procedure will give us a set of binary reachability vectors for the new column, each vector
corresponding to a realisation.
The \emph{reachability vector} is a boolean vector that, for the cell $(i, j)$ of the table, states
whether for a particular realisation $\pi$ of $\mathcal{U}[1: i]$ the discrete Fr\'echet distance
between $\pi$ and $\sigma[1: j]$ is below some threshold $\delta$.

An important observation is that we do not need to distinguish between the realisations that give
the same reachability vector: once we start filling out the next column, all we care about is the
existence of some realisation leading to that particular reachability vector.
So, we can keep a \emph{set} of binary vectors corresponding to reachability in the column.

This procedure was suggested for a specific realisation.
However, we can also repeat this for each previous reachability vector, only keeping the unique
results.
As all the realisation choices happen along $\mathcal{U}$, by treating the table column-by-column we
ensure that we do not have issues with inconsistent choices.
Therefore, repeating this procedure $n$ times, we fill out the last column of the table.
At that point, if any vector has \False\ in the top right cell, then there is some realisation $\pi
\Subset \mathcal{U}$ such that $\dfr(\pi, \sigma) > \delta$, and hence
$\dfrmax(\mathcal{U}, \sigma) > \delta$.

{\def\rds{2pt}
\begin{figure}
\begin{subfigure}{.3\textwidth}
\centering
\begin{tikzpicture}
\fill[gray,opacity=.2]  (0, 0) circle[radius=.4]
                        (2, 1) circle[radius=.4];
\draw[red,thin] (0, 0) -- (2, 1);
\fill[red]  (0, 0) node[below] {$1$} circle[radius=\rds]
            (2, 1) node[below] {$2$} circle[radius=\rds];
\fill[blue] (0, .2) node[left] {$1^a$} circle[radius=\rds]
            (0.2, -.1) node[below] {$1^b$} circle[radius=\rds]
            (.2, .2) node[above] {$2^a$} circle[radius=\rds]
            (2.2, 1) node[right] {$2^b$} circle[radius=\rds]
            (2.3, .8) node[below] {$3^a$} circle[radius=\rds]
            (1.8, 1) node[above] {$3^b$} circle[radius=\rds];
\end{tikzpicture}
\end{subfigure}
\begin{subfigure}{.3\textwidth}
\centering
\begin{tabular}{| c | c | c |}
\hline
F F & F T & T T\\
\hline
T T & T F & F F\\
\hline
\end{tabular}
\end{subfigure}
\begin{subfigure}{.3\textwidth}
\centering
\begin{tabular}{| c |}
\hline
F\\
\hline
T\\
\hline
\end{tabular}
$\to$
\begin{tabular}{| c | c |}
\hline
F & T\\
\hline
T & F\\
\hline
\end{tabular}
$\to$
\begin{tabular}{| c |}
\hline
\color{red}{T}\\
\hline
F\\
\hline
\end{tabular}
\end{subfigure}
\caption{Left: An indecisive and a precise curve.
Middle: Distance matrix.
`T T' in the bottom left cell means $\lVert 1 - 1^a\rVert \leq \delta$ and
$\lVert 1 - 1^b\rVert \leq \delta$.
Right: Computing reachability matrix, column by column.
Note two reachability vectors for the second column.}
\label{fig:alg}
\end{figure}
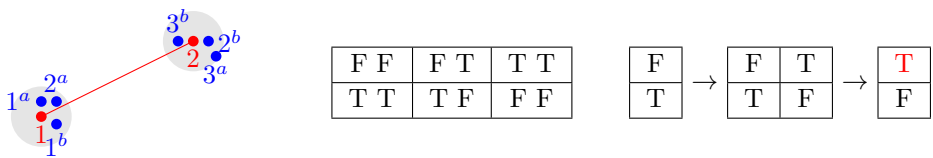}

In more detail, we use two tables, distance matrix $D$ and reachability matrix $R$.
First of all, we initialise the distance matrix $D$ and the reachability of the first column for all
possible locations of $U_1$.
Then we fill out $R$ column-by-column.
We take the reachability of the previous column and note that any cell can be reached either with
the horizontal step or with the diagonal step.
We need to consider various extensions of the curve $\mathcal{U}$ with one of the $\ell$
realisations of the current point; the distance matrix should allow the specific coupling.
Assume we find that a certain cell is reachable; if allowed by the distance matrix, we
can then go upwards, marking cells above the current cell reachable, even if they are not directly
reachable with a horizontal or diagonal step.
Then we just remember the newly computed vector; we make sure to only add distinct vectors.
The computation is illustrated in Figure~\ref{fig:alg}; the pseudocode is given in
Algorithm~\ref{alg:pr_ind}.

\begin{algorithm}[p]
\caption{Finding time-banded upper bound discrete Fr\'echet distance on an indecisive and a precise
curve.}
\label{alg:pr_ind}
\begin{algorithmic}[1]
\Function{TBDFDIndPr}{$\mathcal{U}, \sigma, w, \delta$}
    \LineComment{Input constraint: $\lvert\mathcal{U}\rvert = \lvert\sigma\rvert = n$ and $0 \leq w < n$}
    \State Initialise matrix $D$ of size $n \times \ell \times (2w + 1)$
    \ForAll{$i \in [n]$}
        \ForAll{$k \in [\ell]$}
            \ForAll{$j \in \{\max(1, i - w), \dots, \min(n, i + w)\}$}
                \State $D_{i, k, j} \gets [d(p_i^k, q_j) \leq \delta ?]$
            \EndFor
        \EndFor
    \EndFor
    \State Initialise matrix $R$ of size $n \times 2^{2w + 1} \times 2w + 1$
    \State $R0 \gets \langle r_1 = \True, r_2 = \False, r_3 = \False, \dots, r_{w + 1} =
    \False\rangle$
    \ForAll{$k \in [\ell]$}
        \State $R_{1, k} \gets \Call{Propagate}{R0, D_{1, k}, 1, w, n}$
    \EndFor
    \ForAll{$i \in [n] \setminus \{1\}$}
        \ForAll{$A \in R_{i - 1}$} \Comment{For each reachability vector}
            \State $B \gets A \lor (A << 1)$ \Comment{Horizontal or diagonal step}
            \ForAll{$k \in [\ell]$}
                \State $C \gets \Call{Propagate}{B, D_{i, k}, i, w, n}$
                \State Add $C$ to set $R_{i}$
            \EndFor
        \EndFor
    \EndFor
    \State $r \gets \True$
    \ForAll{$A \in R_n$}
        \State $r \gets r \land A_n$
    \EndFor
    \State\Return $r$
\EndFunction
\Function{Propagate}{$A, B, i, w, n$}
    \LineComment{Propagate the reachability upwards in a column}
    \State $C \gets A \land B$ \Comment{Step and distance matrix}
    \State $r \gets \False$
    \ForAll{$j \in \{\max(1, i - w), \dots, \min(n, i + w)$}
        \If{$B_j \land C_j$}
            \State $r \gets \True$ \Comment{Current cell already reachable}
        \ElsIf{$B_j \land \neg C_j \land r$}
            \State $C_j \gets \True$ \Comment{Vertical step}
        \ElsIf{$\neg B_j$}
            \State $r \gets \False$
        \EndIf
    \EndFor
    \State\Return $C$
\EndFunction
\end{algorithmic}
\end{algorithm}

\paragraph*{Correctness}
We use the following loop invariant to show correctness.
\begin{lemma}\label{lem:alg_basic_correct}
Consider column $i$.
Every reachability vector of this column corresponds to at least one realisation of
$\mathcal{U}[1: i]$ and the discrete Fr\'echet distance between that realisation and
$\sigma[1: \min(n, i + w)]$; and every realisation corresponds to some reachability vector.
\end{lemma}
\begin{proof}
The statement is trivial for the first column: we consider all $\ell$ possible realisations of
$U_1$ and compute reachability of cells $(1, 1)$ to $(1, 1 + w)$ in a straightforward
way.

Now suppose the statement holds for column $i$.
As follows from the recurrence establishing the discrete Fr\'echet distance, the reachability of
column $i + 1$ only depends on the distance matrix for column $i + 1$ and the reachability of column
$i$.
We consider every possible extension of $\mathcal{U}[1: i]$ to $\mathcal{U}[1: i + 1]$, as for every
reachability vector of column $i$ we consider all $\ell$ options for the distance matrix for column
$i + 1$.
Thus, we only consider valid realisations for column $i + 1$, and we consider all of them from the
point of view of reachability.
\end{proof}

\paragraph*{Running time}
First of all, populating the distance matrix takes time $\Theta(\ell n w)$.
A call to \textsc{Propagate} takes $\Theta(w)$ time, so initialisation of first column of
reachability matrix takes $\Theta(\ell w)$ time.
Note that, at any further point, we may have at most $2^{2w + 1}$ distinct reachability vectors; for
each of them, we get $\ell$ calls to \textsc{Propagate}, taking $\Theta(4^w \ell w)$ time per
column, so over all columns we need $\Theta(4^w \ell w n)$ time.
If we assume that adding an element to the set takes amortised constant time, then the previous
value dominates.
Finally, the check at the end takes $\Theta(4^w)$ time.
So, overall the algorithm runs in time $\Theta(4^w \ell n w)$.
This agrees with our hardness result: for a small fixed-width time band, we get the running time of
$\Theta(\ell n)$, whereas if we set $w = n - 1$ to compute the unrestricted distance, we get
$\Theta(4^n \ell n^2)$\dsh clearly, exponential time.
We can also only store vectors that dominate in terms of \False\ values, as we are interested in the
worst case.
This improvement reduces the running time by a factor of $\sqrt{w}$.

\begin{theorem}
Problem \textsc{Upper Bound Discrete Fr\'echet} restricted to a Sakoe--Chiba time band of width $w$
on a precise curve and an uncertain curve on indecisive points with $\ell$ options, both of length
$n$, can be solved in time $\Theta(4^w \ell n \sqrt{w})$ in the worst case.
\end{theorem}

\subsection{Upper Bound Discrete Fr\'echet Distance: Indecisive}\label{app:alg_ind}
Now we extend our previous result to the setting where both curves are indecisive, so instead of
$\sigma$ we have $\mathcal{V} = \langle V_1, \dots, V_n\rangle$, with, for each $j \in [n]$,
$V_j = \{q_j^1, \dots, q_j^\ell\}$.
Suppose we pick a realisation for curve $\mathcal{V}$.
Then we can apply the algorithm we just described.
We cannot run it separately for every realisation; instead, note that the part of the realisation
that matters for column $i$ is the points from $i - w$ to $i + w$, since any previous or further
points are outside the time band.
So, we can fix these $2w + 1$ points and compute the column.
We do so for each possible combination of these $2w + 1$ points.

\begin{lemma}\label{lem:indind}
Any reachability vector we store in column $i$ corresponds to some realisation of the subcurves
$\mathcal{U}[1: i]$ and $\mathcal{V}[1: \min(i + w, n)]$, and every such realisation has the
resulting reachability vector stored in column $i$.
\end{lemma}
\begin{proof}
First of all, consider the statement for column~$1$.
Clearly, we consider all possible realisations of both subcurves, so the statement holds.

Now, as we move from column $i$ to column $i + 1$, we fix the realisation of points $i - w + 1$ to
$i + w + 1$ on curve $\mathcal{V}$ and consider all the vectors stemming from the possible values of
point $i - w$; as in Lemma~\ref{lem:alg_basic_correct}, we cover all realisations of curve
$\mathcal{U}$.

As for curve $\mathcal{V}$, note that we, again, only need the reachability from the previous column
and the distance matrix from the current column, so the points before $i - w + 1$ do not play a role
for the consistency between the two, and thus they can be ignored.

So, we only get reachability vectors corresponding to valid realisations, and we do not miss any,
as required.
\end{proof}

The running time is now $\Theta(4^w \ell^{2w + 1} n w)$, as we consider all combinations of the
$2w + 1$ relevant points on $\mathcal{V}$ with $\ell$ options per point.
For small constant $w$ and $\ell$, we get $\Theta(n)$; for $w = n - 1$, we get
$\Theta(4^n n^2 \ell^{2n})$\dsh exponential time in $n$.
As in the previous algorithm, we can store the boolean vectors more efficiently, reducing the
running time by a factor of $\sqrt{w}$.

\begin{algorithm}[p]
\caption{Finding time-banded upper bound discrete Fr\'echet distance on two indecisive curves.}
\label{alg:ind}
\begin{algorithmic}[1]
\Function{TBDFDIndInd}{$\mathcal{U}, \mathcal{V}, w, \delta$}
    \LineComment{Input constraint: $\lvert\mathcal{U}\rvert = \lvert \mathcal{V}\rvert = n$ and $0
    \leq w < n$}
    \State Initialise matrix $D$ of size $n \times \ell \times (2w + 1) \times \ell$
    \ForAll{$i \in [n]$}
        \ForAll{$k \in [\ell]$}
            \ForAll{$j \in \{\max(1, i - w), \dots, \min(n, i + w)\}$}
                \ForAll{$s \in [\ell]$}
                    \State $D_{i, k, j, s} \gets [d(p_i^k, q_j^s) \leq \delta ?]$
                \EndFor
            \EndFor
        \EndFor
    \EndFor
    \State Initialise matrix $R$ of size $n \times \ell^{2w + 1} \times 2^{2w + 1} \times 2w + 1$
    \State $R0 \gets \langle r_1 = \True, r_2 = \False, r_3 = \False, \dots, r_{w + 1} =
    \False\rangle$
    \ForAll{$s \in [\ell^{w + 1}]$}
        \ForAll{$k \in [\ell]$}
            \State $R_{1, s, k} \gets \Call{Propagate}{R0, D_{1, k}[s], 1, w, n}$
        \EndFor
    \EndFor
    \ForAll{$i \in [n] \setminus \{1\}$}
        \ForAll{$s \in [\ell^{2w + 1}]$} \Comment{Or fewer in edge cases}
            \ForAll{$A \in R_{i - 1}[s]$} \Comment{For each reachability vector with fixed
            realisation}
                \State $B \gets A \lor (A << 1)$
                \ForAll{$k \in [\ell]$}
                    \State $C \gets \Call{Propagate}{B, D_{i, k}[s], i, w, n}$
                    \State Add $C$ to set $R_{i}[s]$
                \EndFor
            \EndFor
        \EndFor
    \EndFor
    \State $r \gets \True$
    \ForAll{$A \in R_n$}
        \ForAll{$s \in [\ell^{2w + 1}]$}
            \State $r \gets r \land A_n[s]$
        \EndFor
    \EndFor
    \State\Return $r$
\EndFunction
\Function{Propagate}{$A, B, i, w, n$}
    \LineComment{Propagate the reachability upwards in a column}
    \State $C \gets A \land B$
    \State $r \gets \False$
    \ForAll{$j \in \{\max(1, i - w), \dots, \min(n, i + w)$}
        \If{$B_j \land C_j$}
            \State $r \gets \True$
        \ElsIf{$B_j \land \neg C_j \land r$}
            \State $C_j \gets \True$
        \ElsIf{$\neg B_j$}
            \State $r \gets \False$
        \EndIf
    \EndFor
    \State\Return $C$
\EndFunction
\end{algorithmic}
\end{algorithm}

\begin{theorem}
Suppose we are given two indecisive curves of length $n$ with $\ell$ options per indecisive
point.
Then we can compute the upper bound discrete Fr\'echet distance restricted to a Sakoe--Chiba band
of width $w$ in time $\Theta(4^w \ell^{2w + 1} n \sqrt{w})$.
\end{theorem}

\subsection{Expected Discrete Fr\'echet Distance}\label{sec:alg_exp}
To compute the expected discrete Fr\'echet distance with time bands, we need two observations:
\begin{enumerate}
    \item For any two precise curves, there is a single threshold $\delta$ where the answer to the
    decision problem changes from \True\ to \False\dsh a \emph{critical value}.
    That threshold corresponds to the distance between some two points on the curves.
    \item We can modify our algorithm to store associated counts with each reachability vector,
    obtaining the fraction of realisations that yield the answer \True\ for a given threshold
    $\delta$.
\end{enumerate}
We can execute our algorithm for each critical value and get the cumulative
distribution function $\mathbb{P}(\dfr(\pi, \sigma) > \delta)$ for $\pi, \sigma \reals[U]
\mathcal{U}, \mathcal{V}$.
As explained in the rest of this section, using the fact that the cumulative distribution function
is a step function, we compute $\dfrexp$.

Consider first the setting of one precise and one indecisive curve.
Note that we store the reachability vectors in a set; instead, we could store a counter with each
reachability vector, so that every time we get an element that is already stored, we increment the
counter.
Notice that we cannot use the improvement that would allow us to discard some vectors, as that would
eschew the count, and we are not interested in the worst possible result now.
We can implement a similar mechanism in the setting of two indecisive curves.
Moreover, we can clearly propagate the count through the algorithm and in the end find the counts
associated with answers \True\ and \False\ to the decision problem.

So, if we store the count of realisations that give us a certain reachability vector, we essentially
obtain, for some value of $\delta$,
\[\mathbb{P}(\dfr(\pi, \sigma) > \delta) \quad \text{when $\pi, \sigma \reals[U] \mathcal{U}, 
\mathcal{V}$.}\]
For any realisation, there is a specific value of $\delta$\dsh a \emph{critical value}\dsh that
acts as a threshold between the answers \True\ and \False\ for that realisation, since if we fix the
realisation we just compute the regular discrete Fr\'echet distance.
Note that that threshold must be a distance between some two points on different curves.
In the case of a precise and an indecisive curve, there are $\ell n (2w + 1)$ such distances
with the time band of width $w$; in the case of two indecisive curves, there are $\ell^2 n 
(2w + 1)$ such distances.
Therefore, if we run our algorithm for each of these critical values and record the counts of \True\
and \False\ for each threshold, we will obtain the complete cumulative distribution function
$\mathbb{P}(\dfr(\pi, \sigma) > \delta)$ for $\pi, \sigma \reals[U] \mathcal{U}, \mathcal{V}$.

Then we can simply find, under the time band restriction,
\[\dfrexpp[U](\mathcal{U}, \mathcal{V}) = \int_0^\infty \mathbb{P}_{\pi, \sigma \reals[U]
\mathcal{U}, \mathcal{V}}(\dfr(\pi, \sigma) > \delta)\,\mathrm{d}\delta\,.\]
For any realisation the answer may change from \True\ to \False\ only at one of the critical values.
So, the distribution of \True\ and \False\ only changes at a finite set of critical values and is
constant between them; therefore, $\mathbb{P}(\dfr(\pi, \sigma) > \delta)$ is a step function.
Hence, finding the integral of interest amounts to multiplying the value of
$\mathbb{P}(\dfr(\pi, \sigma) > \delta)$ by the distance between two successive values of
$\delta$ that match, and summing all the results, i.e.\@ to finding the area under the step
function by summing up areas of rectangles that make it up.

\begin{algorithm}[p]
\caption{Finding time-banded expected discrete Fr\'echet distance on an indecisive and a precise
curve and two indecisive curves.}
\label{alg:exp}
\begin{algorithmic}[1]
\Function{ExpTBDFDIndPr}{$\mathcal{U}, \sigma, w$}
    \LineComment{Input constraint: $\lvert\mathcal{U}\rvert = \lvert\sigma\rvert = n$ and $0 \leq w < n$}
    \State Initialise sorted set $E$
    \ForAll{$i \in [n]$}
        \ForAll{$k \in [\ell]$}
            \ForAll{$j \in \{\max(1, i - w), \dots, \min(n, i + w)\}$}
                \State Add $d(p_i^k, q_j)$ to sorted set $E$
            \EndFor
        \EndFor
    \EndFor
    \State $s \gets E[1]$
    \For{$i \gets 1 \To\ l(E) - 1$}
        \State $\delta \gets E[i], \delta' \gets E[i + 1]$
        \State $p \gets \Call{CntTBDFDIndPr}{\mathcal{U}, \sigma, w, \delta}$
        \State $s \gets s + (1 - p) \cdot (\delta' - \delta)$
    \EndFor
    \State\Return $s$
\EndFunction
\Function{ExpTBDFDIndInd}{$\mathcal{U}, \mathcal{V}, w$}
    \LineComment{Input constraint: $\lvert\mathcal{U}\rvert = \lvert \mathcal{V}\rvert = n$ and $0
    \leq w < n$}
    \State Initialise sorted set $E$
    \ForAll{$i \in [n]$}
        \ForAll{$k \in [\ell]$}
            \ForAll{$j \in \{\max(1, i - w), \dots, \min(n, i + w)\}$}
                \ForAll{$s \in [\ell]$}
                    \State Add $d(p_i^k, q_j^s)$ to sorted set $E$
                \EndFor
            \EndFor
        \EndFor
    \EndFor
    \State $s \gets E[1]$
    \For{$i \gets 1 \To\ l(E) - 1$}
        \State $\delta \gets E[i], \delta' \gets E[i + 1]$
        \State $p \gets \Call{CntTBDFDIndInd}{\mathcal{U}, \mathcal{V}, w, \delta}$
        \State $s \gets s + (1 - p) \cdot (\delta' - \delta)$
    \EndFor
    \State\Return $s$
\EndFunction
\Function{CntTBDFDIndPr}{$\mathcal{U}, \sigma, w, \delta$}
    \LineComment{Like \textsc{TBDFDIndPr}, but returns the fraction of count of \True\ over \False\
    for the final cell.}
\EndFunction
\Function{CntTBDFDIndInd}{$\mathcal{U}, \mathcal{V}, w, \delta$}
    \LineComment{Like \textsc{TBDFDIndInd}, but returns the fraction of count of \True\ over \False\
    for the final cell.}
\EndFunction
\end{algorithmic}
\end{algorithm}

So, clearly, under the time band restriction, we can run one of our algorithms either $\ell n (2w +
1)$ or $\ell^2 n (2w + 1)$ times to obtain the expected discrete Fr\'echet distance.
We show the details in Algorithm~\ref{alg:exp} for the two settings.
We summarise this result as follows.
\begin{theorem}
Suppose we are given an indecisive curve $\mathcal{U}$ and a precise curve $\sigma$ of length $n$
with $\ell$ options per indecisive point and want to find the expected discrete Fr\'echet distance
when constrained to a Sakoe--Chiba band of width $w$.
Then we can run $\textsc{ExpTBDFDIndPr}(\mathcal{U}, \sigma, w)$ to obtain the result in time
$\Theta(4^w \ell^2 n^2 w^2)$ in the worst case.
\end{theorem}
\begin{proof}
First of all, note that from the discussion above it immediately follows that the algorithm is
correct.
In the worst case, every $\delta$ that we have to add to $E$ will be distinct, so we have
$\ell n (2w + 1)$ insertions, taking in total $\Theta(\ell n w \log \ell n w)$ time.
Then, we run \textsc{CntTBDFDIndPr} once per value in $E$, and its running time is the same as that
of \textsc{TBDFDIndPr}, so here we take  time $\Theta(\ell n w \cdot 4^w \ell n w)$ in the worst
case, as claimed.
\end{proof}
We can formalise the result similarly for the other setting.
\begin{theorem}
Suppose we are given two indecisive curves $\mathcal{U}$ and $\mathcal{V}$ of length $n$ with
$\ell$ options per indecisive point and want to find the expected discrete Fr\'echet distance when
constrained to a Sakoe--Chiba band of width $w$.
Then we can run $\textsc{ExpTBDFDIndInd}(\mathcal{U}, \mathcal{V}, w)$ to obtain the result in time
$\Theta(4^w \ell^{2w + 3} n^2 w^2)$ in the worst case.
\end{theorem}
\begin{proof}
Again, note that from the discussion above it immediately follows that the algorithm is correct.
In the worst case, we have $\ell^2 n w$ insertions, taking in total
$\Theta(\ell^2 n w \log \ell n w)$ time.
Then, we run \textsc{CntTBDFDIndInd} once per value in $E$, and its running time is the same as that
of \textsc{TBDFDIndInd}, so here we take time $\Theta(\ell^2 n w \cdot 4^w \ell^{2w + 1} n w)$ in
the worst case, as claimed.
\end{proof}

\subsection{Upper Bound Continuous Fr\'echet Distance}\label{sec:alg_cont}
We can adapt our time band algorithms to handle continuous Fr\'echet distance.
Instead of the boolean reachability vectors, we use vectors of \emph{free space} cells, introduced
by Alt and Godau~\cite{alt:1995, godau:1991}.
We need to now store reachability intervals on cell borders.
The number of these intervals is limited: for any cell, the upper value of the
interval is defined by the distance matrix, so yielding at most $\ell^2$ values; the lower value of
the interval is defined by the distance matrix or by one of the cells from the same row, yielding
exponential dependency on $w$.
However, the algorithm is still polynomial-time in $n$.

In more detail, one could adapt the algorithms for the upper bound discrete Fr\'echet
distance to the case when either both curves are indecisive or one is precise and one is
indecisive, and we are interested in the decision problem for Fr\'echet distance and not discrete
Fr\'echet distance.
Since we are going column-by-column, we would need to store the reachability intervals on the
vertical border of each cell.

It is simpler to see how this would work in the setting of a precise and an indecisive curve: each
column now is a column of a free-space diagram, and we only need to store the intervals on the right
side of the column.
As we progress to the next column, we need to consider all the options from the previous column, so
we need to run the same algorithm, except we store and process vectors of free-space intervals
instead of \True\ and \False.
One other distinction is that we do not consider diagonal steps\dsh for Fr\'echet distance doing so
would not make any sense, as the path is continuous, and the diagonal step is not distinguishable
from a horizontal step followed by a vertical step, if such situation occurs.

In particular, we now take the intervals stored in the distance matrix and compute reachability
based on the previous column: if a cell can be reached horizontally from the previous cell, then
the lower bound of the interval in this cell may need to go up, since we can only use monotone
paths.
\textsc{Propagate} will now take the intervals that correspond to the distance matrix and the
precomputed reachability and make the following adjustment:
if a cell is reachable from below, then the entire interval on the right is actually reachable.
See Figure~\ref{fig:free_adjust} for an example of both cases.

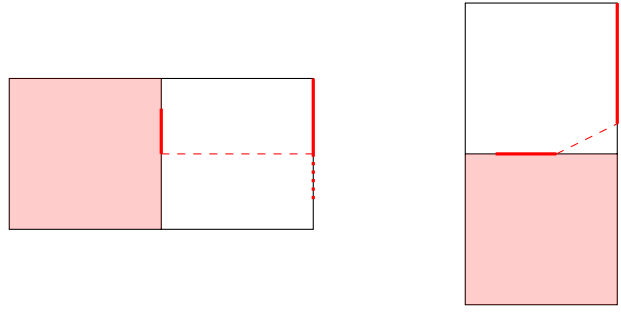
\begin{figure}
\centering
\begin{tikzpicture}[scale=2]
\begin{scope}
    \draw (0, 0) rectangle (2, 1) (1, 0) -- (1, 1);
    \fill[red,opacity=.2] (0, 0) rectangle (1, 1);
    \draw[red,very thick] (1, 0.5) -- (1, 0.8) (2, 0.5)-- (2, 1);
    \draw[red,dotted,very thick] (2, 0.2) -- (2, 0.5);
    \draw[red,dashed] (1, 0.5) -- (2, 0.5);
\end{scope}
\begin{scope}[xshift=3cm]
    \draw (0, -0.5) rectangle (1, 1.5) (0, 0.5) -- (1, 0.5);
    \fill[red,opacity=.2] (0, -0.5) rectangle (1, 0.5);
    \draw[red,very thick] (0.2, 0.5) -- (0.6, 0.5) (1, 0.7)-- (1, 1.5);
    \draw[red,dashed] (0.6, 0.5) -- (1, 0.7);
\end{scope}
\end{tikzpicture}
\caption{Reachability adjustments.
Left: Although the dotted interval is free according to the distance matrix, only the solid interval
is reachable from the cell on the left with a monotone path, assuming entire cell on the left is
free.
Right: The entire interval that is marked as free according to the distance matrix is reachable
with a monotone path from the cell below, assuming the cell below is free.}
\label{fig:free_adjust}
\end{figure}

Other than that, the algorithm is exactly the same; clearly, we can make the same adjustments to the
algorithm handling two indecisive curves.

Notice that we now do not have at most $2^{2w + 1}$ vectors per column, since we store intervals
instead of boolean values, and they can be more varied.
However, the number of values is still limited: for any cell, the upper value of the interval is
defined by the distance matrix, so there can be at most $\ell$ or $\ell^2$ values for the two
settings.
The lower value of the interval is defined by the distance matrix or by one of the cells from the
same row; these may have at most $\ell$ or $\ell^2$ values each, and there are at most $2w$ of them,
so per cell we can have at most $\Theta(\ell w)$ or $\Theta(\ell^2 w)$ lower interval values and
$\Theta(\ell)$ or $\Theta(\ell^2)$ upper interval values, instead of just two possible values in the
discrete case.
The running time changes accordingly, replacing $4^w$ with $(\ell w)^{2w}$, but, importantly, we
still have linear dependency on $n$, so the running time is polynomial for fixed $w$ and $\ell$.

\subsection{Expected Continuous Fr\'echet Distance}\label{app:alg_exp_cont}
We can, of course, again store the associated counts with the vectors of intervals in the algorithm.
As we look at the final cell, we can sum up the counts associated with the cases where the upper
right corner of this cell is reachable, and so we can find the proportion of \True\ to \False\ for a
particular threshold $\delta$.

We can find critical values; now they follow in line with those discussed by Alt and
Godau~\cite{alt:1995, godau:1991}.
The number of the critical values is different: case~1, where we look at the start and end points,
now yields $\Theta(\ell^2)$ events; case~2, where we look at two neighbouring cells, so at the
distance between a segment and a point, yields $\Theta(\ell^3 nw)$ events; and case~3, where we look
at the distance between a segment and two points, yields $\Theta(\ell^4 nw^2)$ events.

Otherwise, we can run Algorithm~\ref{alg:exp} on the new critical values, calling instead the
counting version for the continuous Fr\'echet distance.
This way we can compute the expected Fr\'echet distance restricted to a Sakoe--Chiba band in time
polynomial in $n$ for fixed $w$ and $\ell$.

\begin{theorem}
Suppose we are given two indecisive curves of length $n$ with $\ell$ options per indecisive
point.
Then we can compute the upper bound Fr\'echet distance and expected Fr\'echet distance restricted to
a Sakoe--Chiba band of fixed width $w$ in time polynomial in $n$.
\end{theorem}

\bibliographystyle{plainurl}
\bibliography{references_used}

\begin{thebibliography}{10}

\bibitem{abellanas:2001}
Manuel Abellanas, Ferran Hurtado, Christian Icking, Rolf Klein, Elmar
  Langetepe, Lihong Ma, Bel\'{e}n Palop, and Vera Sacrist\'{a}n.
\newblock Smallest color-spanning objects.
\newblock In {\em Algorithms -- {ESA} 2001}, volume 2161 of {\em Lecture Notes
  in Computer Science}, pages 278--289, Berlin, Germany, 2001. Springer Berlin
  Heidelberg.
\newblock \href {https://doi.org/10.1007/3-540-44676-1_23}
  {\path{doi:10.1007/3-540-44676-1_23}}.

\bibitem{agarwal:2016}
Pankaj~K. Agarwal, Boris Aronov, Sariel Har-Peled, Jeff~M. Phillips, Ke~Yi, and
  Wuzhou Zhang.
\newblock Nearest-neighbor searching under uncertainty {II}.
\newblock {\em {ACM} Transactions on Algorithms ({TALG})}, 13(1):3:1--3:25,
  December 2016.
\newblock \href {https://doi.org/10.1145/2955098} {\path{doi:10.1145/2955098}}.

\bibitem{agarwal:2014}
Pankaj~K. Agarwal, Rinat~Ben Avraham, Haim Kaplan, and Micha Sharir.
\newblock Computing the discrete {F}r\'{e}chet distance in subquadratic time.
\newblock {\em {SIAM} Journal on Computing}, 43(2):429--449, 2014.
\newblock \href {https://doi.org/10.1137/130920526}
  {\path{doi:10.1137/130920526}}.

\bibitem{agarwal:2017}
Pankaj~K. Agarwal, Alon Efrat, Swaminathan Sankararaman, and Wuzhou Zhang.
\newblock Nearest-neighbor searching under uncertainty {I}.
\newblock {\em Discrete \& Computational Geometry}, 58(3):705--745, July 2017.
\newblock \href {https://doi.org/10.1007/s00454-017-9903-x}
  {\path{doi:10.1007/s00454-017-9903-x}}.

\bibitem{ahn:2012}
Hee-Kap Ahn, Christian Knauer, Marc Scherfenberg, Lena Schlipf, and Antoine
  Vigneron.
\newblock Computing the discrete {F}r\'{e}chet distance with imprecise input.
\newblock {\em International Journal of Computational Geometry \&
  Applications}, 22(01):27--44, 2012.
\newblock \href {https://doi.org/10.1142/S0218195912600023}
  {\path{doi:10.1142/S0218195912600023}}.

\bibitem{alt:1995}
Helmut Alt and Michael Godau.
\newblock Computing the {F}r\'{e}chet distance between two polygonal curves.
\newblock {\em International Journal of Computational Geometry and
  Applications}, 5(1):75--91, 1995.
\newblock \href {https://doi.org/10.1142/S0218195995000064}
  {\path{doi:10.1142/S0218195995000064}}.

\bibitem{arkin:2018}
Esther~M. Arkin, Aritra Banik, Paz Carmi, Gui Citovsky, Matthew~J. Katz,
  Joseph~S.B. Mitchell, and Marina Simakov.
\newblock Selecting and covering colored points.
\newblock {\em Discrete Applied Mathematics}, 250:75--86, December 2018.
\newblock \href {https://doi.org/10.1016/j.dam.2018.05.011}
  {\path{doi:10.1016/j.dam.2018.05.011}}.

\bibitem{berndt:1994}
Donald~J. Berndt and James Clifford.
\newblock Using dynamic time warping to find patterns in time series.
\newblock In {\em Proceedings of the 3rd International Conference on Knowledge
  Discovery and Data Mining}, pages 359--370, Palo Alto, CA, USA, 1994. AAAI
  Press.
\newblock \href {https://doi.org/10.5555/3000850.3000887}
  {\path{doi:10.5555/3000850.3000887}}.

\bibitem{bringmann:2014}
Karl Bringmann.
\newblock Why walking the dog takes time: {F}r\'{e}chet distance has no
  strongly subquadratic algorithms unless {SETH} fails.
\newblock In {\em 2014 {IEEE} 55th Annual Symposium on Foundations of Computer
  Science}, pages 661--670, Piscataway, NJ, USA, August 2014. IEEE.
\newblock \href {http://arxiv.org/abs/1404.1448v2} {\path{arXiv:1404.1448v2}},
  \href {https://doi.org/10.1109/FOCS.2014.76}
  {\path{doi:10.1109/FOCS.2014.76}}.

\bibitem{bringmann:2019}
Karl Bringmann, Marvin K\"{u}nnemann, and Andr\'{e} Nusser.
\newblock {F}r\'{e}chet distance under translation: Conditional hardness and an
  algorithm via offline dynamic grid reachability.
\newblock In {\em Proceedings of the Thirtieth Annual {ACM--SIAM} Symposium on
  Discrete Algorithms ({SODA} 2019)}, pages 2902--2921. Society for Industrial
  and Applied Mathematics, January 2019.
\newblock \href {https://doi.org/10.5555/3310435.3310615}
  {\path{doi:10.5555/3310435.3310615}}.

\bibitem{buchin:2010}
Kevin Buchin, Maike Buchin, and Joachim Gudmundsson.
\newblock Constrained free space diagrams: A tool for trajectory analysis.
\newblock {\em International Journal of Geographical Information Science},
  24(7):1101--1125, July 2010.
\newblock \href {https://doi.org/10.1080/13658810903569598}
  {\path{doi:10.1080/13658810903569598}}.

\bibitem{buchin:2017}
Kevin Buchin, Maike Buchin, Wouter Meulemans, and Wolfgang Mulzer.
\newblock Four {S}oviets walk the dog: Improved bounds for computing the
  {F}r\'{e}chet distance.
\newblock {\em Discrete \& Computational Geometry}, 58(1):180--216, 2017.
\newblock \href {https://doi.org/10.1007/s00454-017-9878-7}
  {\path{doi:10.1007/s00454-017-9878-7}}.

\bibitem{buchin:2019-3}
Kevin Buchin, Anne Driemel, Joachim Gudmundsson, Michael Horton, Irina
  Kostitsyna, Maarten L\"{o}ffler, and Martijn Struijs.
\newblock Approximating ($k, \ell$)-center clustering for curves.
\newblock In {\em Proceedings of the Thirtieth Annual {ACM--SIAM} Symposium on
  Discrete Algorithms}, pages 2922--2938, Philadelphia, PA, USA, 2019. SIAM.
\newblock \href {http://arxiv.org/abs/1805.01547v2}
  {\path{arXiv:1805.01547v2}}, \href
  {https://doi.org/10.1137/1.9781611975482.181}
  {\path{doi:10.1137/1.9781611975482.181}}.

\bibitem{buchin:2019-4}
Kevin Buchin, Anne Driemel, Natasja van~de L'Isle, and Andr\'{e} Nusser.
\newblock klcluster: Center-based clustering of trajectories.
\newblock In {\em Proceedings of the 27th {ACM} {SIGSPATIAL} International
  Conference on Advances in Geographic Information Systems}, pages 496--499,
  New York, NY, USA, 2019. ACM.
\newblock \href {https://doi.org/10.1145/3347146.3359111}
  {\path{doi:10.1145/3347146.3359111}}.

\bibitem{buchin:2011}
Kevin Buchin, Maarten L\"{o}ffler, Pat Morin, and Wolfgang Mulzer.
\newblock Preprocessing imprecise points for {D}elaunay triangulation:
  Simplified and extended.
\newblock {\em Algorithmica}, 61(3):674--693, November 2011.
\newblock \href {https://doi.org/10.1007/s00453-010-9430-0}
  {\path{doi:10.1007/s00453-010-9430-0}}.

\bibitem{buchin:2019}
Kevin Buchin, Tim Ophelders, and Bettina Speckmann.
\newblock {SETH} says: Weak {F}r\'{e}chet distance is faster, but only if it is
  continuous and in one dimension.
\newblock In {\em Proceedings of the Thirtieth Annual {ACM--SIAM} Symposium on
  Discrete Algorithms ({SODA} '19)}, pages 2887--2901. Society for Industrial
  and Applied Mathematics, January 2019.
\newblock \href {https://doi.org/10.5555/3310435.3310614}
  {\path{doi:10.5555/3310435.3310614}}.

\bibitem{buchin:2012}
Kevin Buchin, Stef Sijben, T.~Jean~Marie Arseneau, and Erik~P. Willems.
\newblock Detecting movement patterns using {B}rownian bridges.
\newblock In {\em Proceedings of the 20th International Conference on Advances
  in Geographic Information Systems}, pages 119--128, New York, NY, USA, 2012.
  ACM.
\newblock \href {https://doi.org/10.1145/2424321.2424338}
  {\path{doi:10.1145/2424321.2424338}}.

\bibitem{buchin:2014}
Maike Buchin, Anne Driemel, and Bettina Speckmann.
\newblock Computing the {F}r\'{e}chet distance with shortcuts is {NP}-hard.
\newblock In {\em Proceedings of the Thirtieth Annual Symposium on
  Computational Geometry ({SoCG} 2014)}, pages 367--376, New York, NY, USA,
  June 2014. Association for Computing Machinery.
\newblock \href {https://doi.org/10.1145/2582112.2582144}
  {\path{doi:10.1145/2582112.2582144}}.

\bibitem{buchin:2019-2}
Maike Buchin and Stef Sijben.
\newblock Discrete {F}r\'{e}chet distance for uncertain points, 2016.
\newblock Presented at EuroCG 2016, Lugano, Switzerland.
\newblock URL:
  \url{http://www.eurocg2016.usi.ch/sites/default/files/paper_72.pdf} [cited
  2019-07-10].

\bibitem{leizhen:1997}
Leizhen Cai and Mark Keil.
\newblock Computing visibility information in an inaccurate simple polygon.
\newblock {\em International Journal of Computational Geometry \&
  Applications}, 7:515--538, 1997.
\newblock \href {https://doi.org/10.1142/S0218195997000326}
  {\path{doi:10.1142/S0218195997000326}}.

\bibitem{das:2009}
Sandip Das, Partha~P. Goswami, and Subhas~C. Nandy.
\newblock Smallest color-spanning object revisited.
\newblock {\em International Journal of Computational Geometry \&
  Applications}, 19(5):457--478, October 2009.
\newblock \href {https://doi.org/10.1142/S0218195909003076}
  {\path{doi:10.1142/S0218195909003076}}.

\bibitem{devogele:2017}
Thomas Devogele, Laurent Etienne, Maxence Esnault, and Florian Lardy.
\newblock Optimized discrete {F}r{\'e}chet distance between trajectories.
\newblock In {\em Proc.\ 6th {ACM} {SIGSPATIAL} Workshop on Analytics for Big
  Geospatial Data}, pages 11--19, New York, NY, USA, 2017. Association for
  Computing Machinery.
\newblock \href {https://doi.org/10.1145/3150919.3150924}
  {\path{doi:10.1145/3150919.3150924}}.

\bibitem{driemel:2018}
Anne Driemel and Sariel Har-Peled.
\newblock Jaywalking your dog: Computing the {F}r\'{e}chet distance with
  shortcuts.
\newblock {\em {SIAM} Journal on Computing}, 42(5):1830--1866, October 2018.
\newblock \href {http://arxiv.org/abs/1107.1720v4} {\path{arXiv:1107.1720v4}},
  \href {https://doi.org/10.1137/120865112} {\path{doi:10.1137/120865112}}.

\bibitem{driemel:2012}
Anne Driemel, Sariel Har-Peled, and Carola Wenk.
\newblock Approximating the {F}r\'{e}chet distance for realistic curves in near
  linear time.
\newblock {\em Discrete \& Computational Geometry}, 48(1):94--127, July 2012.
\newblock \href {https://doi.org/s00454-012-9402-z}
  {\path{doi:s00454-012-9402-z}}.

\bibitem{driemel:2013}
Anne Driemel, Herman Haverkort, Maarten L\"{o}ffler, and Rodrigo~I. Silveira.
\newblock Flow computations on imprecise terrains.
\newblock {\em Journal of Computational Geometry ({JoCG})}, 4(1):38--78, 2013.
\newblock \href {https://doi.org/10.20382/jocg.v4i1a3}
  {\path{doi:10.20382/jocg.v4i1a3}}.

\bibitem{eiter:1994}
Thomas Eiter and Heikki Mannila.
\newblock Computing discrete {F}r\'{e}chet distance.
\newblock Technical Report CD-TR 94/64, Technishe {U}niversit\"{a}t {W}ien,
  April 1994.
\newblock URL:
  \url{http://www.kr.tuwien.ac.at/staff/eiter/et-archive/cdtr9464.pdf} [cited
  2019-04-23].

\bibitem{evans:2013}
William Evans, David Kirkpatrick, Maarten L\"{o}ffler, and Frank Staals.
\newblock Competitive query strategies for minimising the ply of the potential
  locations of moving points.
\newblock In {\em Proceedings of the Twenty-Ninth Annual Symposium on
  Computational Geometry}, pages 155--164, New York, NY, USA, 2013. ACM.
\newblock \href {https://doi.org/10.1145/2462356.2462395}
  {\path{doi:10.1145/2462356.2462395}}.

\bibitem{fan:2013}
Chenglin Fan, Jun Luo, and Binhai Zhu.
\newblock Tight approximation bounds for connectivity with a color-spanning
  set.
\newblock In {\em Algorithms and Computation ({ISAAC} 2013)}, volume 8283 of
  {\em Lecture Notes in Computer Science}, pages 590--600, Berlin, Germany,
  2013. Springer Berlin Heidelberg.
\newblock \href {https://doi.org/10.1007/978-3-642-45030-3_55}
  {\path{doi:10.1007/978-3-642-45030-3_55}}.

\bibitem{fan:2017}
Chenglin Fan and Benjamin Raichel.
\newblock Computing the {F}r\'{e}chet gap distance.
\newblock In {\em 33rd International Symposium on Computational Geometry
  ({SoCG} 2017)}, volume~77 of {\em Leibniz International Proceedings in
  Informatics ({LIPIcs})}, pages 42:1--42:16, Dagstuhl, Germany, 2017. Schloss
  Dagstuhl -- Leibniz-Zentrum f\"{u}r Informatik.
\newblock \href {https://doi.org/10.4230/LIPIcs.SoCG.2017.42}
  {\path{doi:10.4230/LIPIcs.SoCG.2017.42}}.

\bibitem{fan:2018}
Chenglin Fan and Binhai Zhu.
\newblock Complexity and algorithms for the discrete {F}r\'{e}chet distance
  upper bound with imprecise input, February 2018.
\newblock \href {http://arxiv.org/abs/1509.02576v2}
  {\path{arXiv:1509.02576v2}}.

\bibitem{filtser:2018}
Omrit Filtser and Matthew~J. Katz.
\newblock Algorithms for the discrete {F}r\'{e}chet distance under translation.
\newblock In {\em 16th {S}candinavian Symposium and Workshops on Algorithm
  Theory ({SWAT} 2018)}, volume 101 of {\em Leibniz International Proceedings
  in Informatics ({LIPIcs})}, pages 20:1--20:14, Dagstuhl, Germany, 2018.
  Schloss Dagstuhl -- Leibniz-Zentrum f\"{u}r Informatik.
\newblock \href {https://doi.org/10.4230/LIPIcs.SWAT.2018.20}
  {\path{doi:10.4230/LIPIcs.SWAT.2018.20}}.

\bibitem{godau:1991}
Michael Godau.
\newblock A natural metric for curves: Computing the distance for polygonal
  chains and approximation algorithms.
\newblock In {\em {STACS} 91: Proceedings of 8th Annual Symposium on
  Theoretical Aspects of Computer Science}, volume 480 of {\em Lecture Notes in
  Computer Science}, pages 127--136, Berlin, Germany, 1991. Springer Berlin
  Heidelberg.
\newblock \href {https://doi.org/10.1007/BFb0020793}
  {\path{doi:10.1007/BFb0020793}}.

\bibitem{gray:2012}
Chris Gray, Frank Kammer, Maarten L\"{o}ffler, and Rodrigo~I. Silveira.
\newblock Removing local extrema from imprecise terrains.
\newblock {\em Computational Geometry}, 45(7):334--349, 2012.
\newblock \href {https://doi.org/10.1016/j.comgeo.2012.02.002}
  {\path{doi:10.1016/j.comgeo.2012.02.002}}.

\bibitem{gudmundsson:2019}
Joachim Gudmundsson, Majid Mirzanezhad, Ali Mohades, and Carola Wenk.
\newblock Fast {F}r\'{e}chet distance between curves with long edges.
\newblock {\em International Journal of Computational Geometry \&
  Applications}, 29(2):161--187, 2019.
\newblock \href {https://doi.org/10.1142/S0218195919500043}
  {\path{doi:10.1142/S0218195919500043}}.

\bibitem{guibas:1993}
Leonidas~J. Guibas, John~E. Hershberger, Joseph S.~B. Mitchell, and Jack~S.
  Snoeyink.
\newblock Approximating polygons and subdivisions with minimum-link paths.
\newblock {\em International Journal of Computational Geometry \&
  Applications}, 3(4):383--415, 1993.
\newblock \href {https://doi.org/10.1142/S0218195993000257}
  {\path{doi:10.1142/S0218195993000257}}.

\bibitem{harpeled:2014}
Sariel Har-Peled and Benjamin Raichel.
\newblock The {F}r\'{e}chet distance revisited and extended.
\newblock {\em {ACM} Transactions on Algorithms ({TALG})}, 10(1):3:1--3:22,
  January 2014.
\newblock \href {https://doi.org/10.1145/2532646} {\path{doi:10.1145/2532646}}.

\bibitem{jorgensen:2011}
Allan J{\o}rgensen, Jeff Phillips, and Maarten L\"{o}ffler.
\newblock Geometric computations on indecisive points.
\newblock In {\em Algorithms and Data Structures ({WADS} 2011)}, volume 6844 of
  {\em Lecture Notes in Computer Science}, pages 536--547, Berlin, Germany,
  2011. Springer Berlin Heidelberg.
\newblock \href {https://doi.org/10.1007/978-3-642-22300-6_45}
  {\path{doi:10.1007/978-3-642-22300-6_45}}.

\bibitem{keogh:2005}
Eamonn Keogh and Chotirat~Ann Ratanamahatana.
\newblock Exact indexing of dynamic time warping.
\newblock {\em Knowledge and Information Systems}, 7(3):358--386, 2005.
\newblock \href {https://doi.org/10.1007/s10115-004-0154-9}
  {\path{doi:10.1007/s10115-004-0154-9}}.

\bibitem{knauer:2011}
Christian Knauer, Maarten L\"{o}ffler, Marc Scherfenberg, and Thomas Wolle.
\newblock The directed {H}ausdorff distance between imprecise point sets.
\newblock {\em Theoretical Computer Science}, 412(32):4173--4186, 2011.
\newblock \href {https://doi.org/10.1016/j.tcs.2011.01.039}
  {\path{doi:10.1016/j.tcs.2011.01.039}}.

\bibitem{krumm:2009}
John Krumm.
\newblock A survey of computational location privacy.
\newblock {\em Personal and Ubiquitous Computing}, 13(6):391--399, August 2009.
\newblock \href {https://doi.org/10.1007/s00779-008-0212-5}
  {\path{doi:10.1007/s00779-008-0212-5}}.

\bibitem{loeffler:2009}
Maarten L\"{o}ffler.
\newblock {\em Data Imprecision in Computational Geometry}.
\newblock PhD thesis, Universiteit {U}trecht, October 2009.
\newblock URL:
  \url{https://dspace.library.uu.nl/bitstream/handle/1874/36022/loffler.pdf}
  [cited 2019-06-15].

\bibitem{loeffler:2014}
Maarten L\"{o}ffler and Wolfgang Mulzer.
\newblock Unions of onions: Preprocessing imprecise points for fast onion
  decomposition.
\newblock {\em Journal of Computational Geometry ({JoCG})}, 5(1):1--13, 2014.
\newblock \href {https://doi.org/10.20382/jocg.v5i1a1}
  {\path{doi:10.20382/jocg.v5i1a1}}.

\bibitem{loeffler:2010}
Maarten L\"{o}ffler and Jack Snoeyink.
\newblock {D}elaunay triangulations of imprecise points in linear time after
  preprocessing.
\newblock {\em Computational Geometry: Theory and Applications},
  43(3):234--242, 2010.
\newblock \href {https://doi.org/10.1016/j.comgeo.2008.12.007}
  {\path{doi:10.1016/j.comgeo.2008.12.007}}.

\bibitem{loeffler:2006}
Maarten L\"{o}ffler and Marc van Kreveld.
\newblock Largest and smallest tours and convex hulls for imprecise points.
\newblock In {\em Algorithm Theory -- {SWAT} 2006}, volume 4059 of {\em Lecture
  Notes in Computer Science}, pages 375--387, Berlin, Germany, 2006. Springer
  Berlin Heidelberg.
\newblock \href {https://doi.org/10.1007/11785293_35}
  {\path{doi:10.1007/11785293_35}}.

\bibitem{maheshwari:2011}
Anil Maheshwari, J\"{o}rg-R\"{u}diger Sack, Kaveh Shahbaz, and Hamid
  Zarrabi-Zadeh.
\newblock {F}r\'{e}chet distance with speed limits.
\newblock {\em Computational Geometry}, 44(2):110--120, 2011.
\newblock \href {https://doi.org/10.1016/j.comgeo.2010.09.008}
  {\path{doi:10.1016/j.comgeo.2010.09.008}}.

\bibitem{pei:2007}
Jian Pei, Bin Jiang, Xuemin Lin, and Yidong Yuan.
\newblock Probabilistic skylines on uncertain data.
\newblock In {\em Proceedings of the 33rd International Conference on Very
  Large Data Bases}, pages 15--26. VLDB Endowment, September 2007.
\newblock \href {https://doi.org/10.5555/1325851.1325858}
  {\path{doi:10.5555/1325851.1325858}}.

\bibitem{pfoser:1999}
Dieter Pfoser and Christian~S. Jensen.
\newblock Capturing the uncertainty of moving-object representations.
\newblock In {\em Advances in Spatial Databases}, volume 1651 of {\em Lecture
  Notes in Computer Science}, pages 111--131, Berlin, Germany, June 1999.
  Springer Berlin Heidelberg.
\newblock \href {https://doi.org/10.1007/3-540-48482-5_9}
  {\path{doi:10.1007/3-540-48482-5_9}}.

\bibitem{prasadsistla:1998}
A.~Prasad~Sistla, Ouri Wolfson, Sam Chamberlain, and Son Dao.
\newblock Querying the uncertain position of moving objects.
\newblock In Opher Etzion, Sushil Jajodia, and Suryanarayana Sripada, editors,
  {\em Temporal Databases: Research and Practice}, volume 1399 of {\em Lecture
  Notes in Computer Science}, pages 310--337. Springer Berlin Heidelberg,
  Berlin, Germany, 1998.
\newblock \href {https://doi.org/10.1007/BFb0053708}
  {\path{doi:10.1007/BFb0053708}}.

\bibitem{sakoe:1978}
Hiroaki Sakoe and Seibi Chiba.
\newblock Dynamic programming algorithm optimization for spoken word
  recognition.
\newblock {\em {IEEE} Transactions on Acoustics, Speech, and Signal
  Processing}, 26(1):43--49, February 1978.
\newblock \href {https://doi.org/10.1109/TASSP.1978.1163055}
  {\path{doi:10.1109/TASSP.1978.1163055}}.

\bibitem{sember:2008}
Jeff Sember and William Evans.
\newblock Guaranteed {V}oronoi diagrams of uncertain sites.
\newblock In {\em Proceedings of the 20th {C}anadian Conference on
  Computational Geometry ({CCCG} 2008)}, pages 203--206, 2008.
\newblock URL: \url{http://cccg.ca/proceedings/2008/paper50full.pdf}.

\bibitem{suri:2013}
Subhash Suri, Kevin Verbeek, and Hakan Y{\i}ld{\i}z.
\newblock On the most likely convex hull of uncertain points.
\newblock In {\em Algorithms -- {ESA} 2013}, volume 8125 of {\em Lecture Notes
  in Computer Science}, pages 791--802, Berlin, Germany, 2013. Springer Berlin
  Heidelberg.
\newblock \href {https://doi.org/10.1007/978-3-642-40450-4_67}
  {\path{doi:10.1007/978-3-642-40450-4_67}}.

\bibitem{kerkhof:2019}
Mees van~de Kerkhof, Irina Kostitsyna, Maarten L\"{o}ffler, Majid Mirzanezhad,
  and Carola Wenk.
\newblock Global curve simplification.
\newblock In {\em 27th Annual European Symposium on Algorithms ({ESA} 2019)},
  volume 144 of {\em Leibniz International Proceedings in Informatics
  ({LIPIcs})}, pages 67:1--67:14, Dagstuhl, Germany, 2019. Schloss Dagstuhl --
  Leibniz-Zentrum f\"{u}r Informatik.
\newblock \href {https://doi.org/10.4230/LIPIcs.ESA.2019.67}
  {\path{doi:10.4230/LIPIcs.ESA.2019.67}}.

\bibitem{kreveld:2010}
Marc van Kreveld, Maarten L\"{o}ffler, and Joseph S.~B. Mitchell.
\newblock Preprocessing imprecise points and splitting triangulations.
\newblock {\em {SIAM} Journal on Computing}, 39(7):2990--3000, May 2010.
\newblock \href {https://doi.org/10.1137/090753620}
  {\path{doi:10.1137/090753620}}.

\bibitem{kreveld:2018}
Marc van Kreveld, Maarten L\"{o}ffler, and Lionov Wiratma.
\newblock On optimal polyline simplification using the {H}ausdorff and
  {F}r\'{e}chet distance.
\newblock In {\em 34th International Symposium on Computational Geometry
  ({SoCG} 2018)}, volume~99 of {\em Leibniz International Proceedings in
  Informatics ({LIPIcs})}, pages 56:1--56:14, Dagstuhl, Germany, 2018. Schloss
  Dagstuhl -- Leibniz-Zentrum f\"{u}r Informatik.
\newblock \href {https://doi.org/10.4230/LIPIcs.SoCG.2018.56}
  {\path{doi:10.4230/LIPIcs.SoCG.2018.56}}.

\bibitem{yiu:2009}
Man~Lung Yiu, Nikos Mamoulis, Xiangyuan Dai, Yufei Tao, and Michail Vaitis.
\newblock Efficient evaluation of probabilistic advanced spatial queries on
  existentially uncertain data.
\newblock {\em {IEEE} Transactions on Knowledge and Data Engineering},
  21(1):108--122, 2009.
\newblock \href {https://doi.org/10.1109/TKDE.2008.135}
  {\path{doi:10.1109/TKDE.2008.135}}.

\end{thebibliography}
\end{document}